\documentclass[a4paper,USenglish,modulo,cleveref,autoref,numberwithinsect]{lipics-v2021} 

\newif\iflong
\newif\ifshort

\longtrue

\iflong
\else
\shorttrue
\fi

\nolinenumbers

\usepackage[utf8]{inputenc}
\usepackage{cite,enumitem}
\usepackage{todonotes}
\usepackage{hyperref}

{
\theoremstyle{plain}
\newtheorem{result}{Result}
\newtheorem{fact}[theorem]{Fact}

}

\newcounter{ourcase}
\newenvironment{ourcase}[1][]
	{\stepcounter{ourcase}%
	\medbreak\noindent{\textsf{Case~\arabic{ourcase}%
	\def\tmp{#1}\ifx\tmp\empty.\else~(#1).\fi
	}}}
	{\medbreak}

\makeatletter
\g@addto@macro\bfseries{\boldmath}
\makeatother
\usepackage{mdframed}

\usepackage{macros}

\usepackage{complexity}
\usepackage{thm-restate}

\newcommand{\bigoh}{\ensuremath{{\mathcal O}}}

\newcommand{\cF}{\ensuremath{\mathcal F}}

\newcommand{\Nat}{\mathbb{N}}
\newcommand{\tw}{\operatorname{tw}}

\newcommand{\bw}{\operatorname{bw}}
\newcommand{\td}{\operatorname{td}}
\newcommand{\dist}{\operatorname{dist}}
\newcommand{\decT}{\mathfrak{B}}
\newcommand\decf{\mathcal{L}}

\newcommand{\ph}{partner-hereditary}
\newcommand{\si}{size-identifiable}
\newcommand{\primal}{\emph{\texttt{primal}}\xspace}
\newcommand{\Fbw}{\operatorname{\mathcal{F}-branchwidth}}
\newcommand{\Fsbw}{\mathcal{F^*}\operatorname{-bw}}
\newcommand{\FbwSh}{\operatorname{\mathcal{F}-bw}}
\newcommand{\mimw}{\imatchwidth}

\usepackage{chngcntr}
\counterwithin*{case}{subsection}

\newcommand{\FFF}{\mathcal{F}}

\newcommand\calQ{\mathcal{Q}}
\newcommand\defeq{:=}
\newcommand\card[1]{|#1|}
\newcommand\Fstar{\ensuremath{\FFF^*}}
\newcommand\Fstarval{\Fstar\operatorname{-value}}
\newcommand\decS{\mathfrak{C}}
\newcommand\Fmatch{\ensuremath{\FFF_{\imatch}}\xspace}
\newcommand\Fchain{\ensuremath{\FFF_{\ichain}}\xspace}
\newcommand\Fantimatch{\ensuremath{\FFF_{\iantimatch}}\xspace}

\author{Eduard Eiben}{Department of Computer Science, Royal Holloway, University of London, Egham, United Kingdom}{eduard.eiben@rhul.ac.uk}{https://orcid.org/0000-0003-2628-3435}{}
\author{Robert Ganian}{Algorithms and Complexity Group, TU Wien, Vienna, Austria}{rganian@ac.tuwien.ac.at}{https://orcid.org/0000-0002-7762-8045}{Robert Ganian acknowledges support by the Austrian Science Fund (FWF, projects P31336 and Y1329).}
\author{Thekla Hamm}{Algorithms and Complexity Group, TU Wien, Vienna, Austria}{thamm@ac.tuwien.ac.at}{https://orcid.org/0000-0002-4595-9982}{Thekla Hamm acknowledges support by the Austrian Science Fund (FWF, projects P31336, Y1329 and W1255-N23).}
\author{Lars Jaffke}{Department of Informatics, University of Bergen, Bergen, Norway}{lars.jaffke@uib.no}{}{Lars Jaffke was supported by the Norwegian Research Council, via project 274526.}
\author{O-joung Kwon}{Department of Mathematics, Incheon National University, Incheon, South Korea\\ Discrete Mathematics Group, Institute for Basic Science, Daejeon, South Korea}%
{ojoungkwon@gmail.com}{}{O-joung Kwon was supported by the National Research Foundation of Korea (NRF) grant funded by the Ministry of Education (No. NRF- 2018R1D1A1B07050294), and by the Institute for Basic Science (IBS-R029-C1).}

\authorrunning{E.~Eiben, R.~Ganian, T.~Hamm, L.~Jaffke, O.~Kwon}

\keywords{branchwidth, parameterized algorithms, mim-width, treewidth, clique-width}
\Copyright{Eduard Eiben, Robert Ganian, Thekla Hamm, Lars Jaffke, O-Joung Kwon}
\ccsdesc[300]{Theory of computation~Parameterized complexity and exact algorithms}
\ifshort
\title{A Unifying Framework for Characterizing and Computing Width Measures}
\fi
\iflong
\title{A Unifying Framework for Characterizing and Computing Width Measures
}
\fi

\usepackage{macros}
\newcommand{\phsymb}{\emph{\texttt{ph}}\xspace}
\newcommand{\sisymb}{\emph{\texttt{si}}\xspace}

\newcommand{\calO}{\mathcal{O}}
\newcommand\Fsbwprob{\textsc{$\FFF^*$-Branchwidth}\xspace}

\usepackage{vmargin}
\setpapersize{A4}
\setmarginsrb{20mm}{15mm}{20mm}{20mm}{0pt}{5mm}{0pt}{0mm}
\hideLIPIcs

\begin{document}
\maketitle
\begin{abstract}
Algorithms for computing or approximating optimal decompositions for decompositional parameters such as treewidth or clique-width have so far traditionally been tailored to specific width parameters. Moreover, for mim-width, no efficient algorithms for computing good decompositions were known, even under highly restrictive parameterizations. In this work we identify $\mathcal{F}$-branchwidth as a class of generic decompositional parameters that can capture mim-width, treewidth, clique-width as well as other measures. We show that while there is an infinite number of $\mathcal{F}$-branchwidth parameters, only a handful of these are asymptotically distinct. We then develop fixed-parameter and kernelization algorithms (under several structural parameterizations) that can compute \emph{every possible} $\mathcal{F}$-branchwidth, providing a unifying framework that can efficiently obtain near-optimal tree-decompositions, $k$-expressions, as well as optimal mim-width decompositions.
\end{abstract}

	\section{Introduction}
	
Over the past twenty years, the study of decompositional graph parameters and their algorithmic applications has become a focal point of research in theoretical computer science. \emph{Treewidth}~\cite{RobertsonSeymour86} is, naturally, the most prominent example of such a parameter, but is far from the only parameter of interest: graph classes of unbounded treewidth may still have bounded \emph{clique-width}~\cite{CourcelleMakowskyRotics00}, and graph classes of unbounded clique-width may still have bounded \emph{mim-width}~\cite{vatshelle2012new,JaffkeKT20}. A common feature of such parameters is that they are all defined using some notion of a ``decomposition'' which is also of importance in various algorithmic applications, although the precise definitions of such decompositions vary from parameter to parameter.

Given their importance, it is not surprising that finding efficient algorithms for computing suitable decompositions has become a prominent research direction in the field. While computing optimal decompositions has proven to be a challenging task, in many scenarios it is sufficient to simply obtain an approximately-optimal decomposition instead---i.e., we ask for an algorithm which, given a graph of ``width'' at most $k$, outputs a decomposition witnessing that the ``width'' is at most $f(k)$ for some function $f$. Such approximately-optimal decompositions are often obtained by finding an asymptotically equivalent reformulation of the width parameter\iflong\footnote{Two width parameters are asymptotically equivalent if on every graph class, they are either both bounded or both unbounded.}\fi. Indeed, there has been a series of fixed-parameter algorithms for computing approximately-optimal tree decompositions~\cite{BodlaenderDDFLP16}, and the efficient computation of approximately-optimal decompositions for clique-width was a long-standing open problem that has culminated in the introduction of rank-width~\cite{OumSeymour06}, a parameter which is asymptotically equivalent to clique-width. On the other hand, up to now approximately-optimal mim-width decompositions could only be computed in a few highly restrictive settings~\cite{HogemoTV19,SaetherV16}.

While the algorithmic applications of treewidth, clique-width and mim-width have a number of unifying characteristics and there already exist algorithmic frameworks capable of exploiting multiple of these measures when the respective decompositions are provided~\cite{BergougnouxK19},
to the best of our knowledge no unifying framework for the task of actually computing the respective decompositions was known. Indeed, while all three measures admit asymptotically-equivalent characterizations via the notion of \emph{branchwidth}~\cite{OumSeymour06,JeongST18}, \iflong which is a generic template for defining width parameters based on a notion of \emph{cut-functions}, \fi 
the cut-functions employed in these previously known formulations were fundamentally different.

\smallskip
\noindent \textbf{Contribution.}\quad 
We introduce the notion of $\mathcal{F}$-branchwidth, which is a restriction of branchwidth to cut functions defined via an infinite class $\mathcal{F}$ of obstructions; more precisely, $\mathcal{F}$ can be any class of bipartite graphs with a basic property which we call \emph{partner-hereditarity} (\texttt{\emph{ph}}). 
Intuitively, a class $\mathcal{F}$ of bipartite graphs is \texttt{\emph{ph}} if every graph in $\mathcal{F}$ has its vertex set partitioned into pairs of vertices (one on each side), and each subgraph induced on a subset of the pairs is also in $\mathcal{F}$.
As our first result, we show that treewidth, clique-width, and mim-width are all asymptotically equivalent to $\mathcal{F}$-branchwidth for suitable choices of $\mathcal{F}$.

Taken on its own, this result is not surprising: partner-hereditarity is a weak restriction that simply allows graph classes to serve as obstructions in cut-functions, there are infinitely many \texttt{\emph{ph}} graph classes, and some happen to characterize these well-studied decompositional parameters. For our second result, we consider an additional useful property of our obstruction classes: $\mathcal{F}$ is \emph{size-identifiable} (\texttt{\emph{si}}) if there exists a unique graph in $\mathcal{F}$ of each order. We show:

\begin{result}
\label{res:si}
There exist precisely six \texttt{si} \texttt{ph} classes. Moreover, for each \texttt{ph} class $\mathcal{F}$, $\mathcal{F}$-branchwidth is asymptotically equivalent to $\mathcal{F}^\downarrow$-branchwidth, where $\mathcal{F}^\downarrow$ is the restriction of $\mathcal{F}$ to the \texttt{si} \texttt{ph} classes contained in $\mathcal{F}$.
\end{result}

The six \texttt{\emph{si}} \texttt{\emph{ph}} classes comprise the classes of \emph{complete bipartite}, \emph{matching} and \emph{chain} graphs along with their bipartite complements.
Essentially, Result~\ref{res:si} means that all decompositional parameters which can be defined using \texttt{\emph{ph}} obstruction classes can be characterized by a combination of these \texttt{\emph{si}} classes, with treewidth, clique-width and mim-width being three of these. 
While some other combinations of obstruction classes do not give rise to interesting parameters (e.g., it is easy to verify that the parameter obtained by excluding the \emph{complete bipartite} class and its complement class is only bounded for bounded-size graphs), other combinations of obstruction classes lead to completely new structural graph parameters which might have unique properties and have, to the best of our knowledge, never been studied (e.g., excluding the \emph{anti-matching} class or the \emph{chain} class).

What is more, we show that this unification result can be strengthened even further when the aim is to compute approximately-optimal decompositions for $\mathcal{F}$-branchwidth instead of characterizing the widths themselves. We identify three ``\texttt{\emph{primal}}'' \texttt{\emph{si ph}} classes---the classes of matching, chain, and bipartite complements of matchings---and show:

\begin{result}
\label{res:primal}
For a \texttt{ph} class $\mathcal{F}$, let $\mathcal{F}^*$ be the restriction of $\mathcal{F}$ to the \texttt{primal} \texttt{ph} classes contained in $\mathcal{F}$. Then every optimal $\mathcal{F}^*$-branch decomposition is also an approximately-optimal $\mathcal{F}$-branch decomposition.
\end{result}

Result~\ref{res:primal} allows us to reduce the computation of approximately-optimal $\mathcal{F}$-branch decompositions for any $\mathcal{F}$ to combinations of just three \texttt{\emph{primal}} \texttt{\emph{ph}} classes. For instance, treewidth and clique-width have different characterizations in terms of \texttt{\emph{si ph}} classes, but computing decompositions for both parameters corresponds to a single case in our framework. Computing decompositions for mim-width then corresponds to another case.

Having reduced the computation of this wide range of decompositional parameters to $\mathcal{F}^*$-\textsc{Branchwidth}, i.e., the problem of computing an optimal $\mathcal{F}^*$-branch decomposition for some union $\mathcal{F}^*$ of \texttt{\emph{primal ph}} classes, the second part of our paper is dedicated to solving this problem. Obtaining an algorithm for $\mathcal{F}^*$-\textsc{Branchwidth} when parameterized by the target width itself is a problem that not only unifies the approximation of treewidth and clique-width, but also the long-standing open problem of computing mim-width (which is known to be $\W[1]$-hard~\cite{SaetherV16} and for which even \XP-tractability is open). But while we cannot hope for fixed-parameter tractability under this parameterization, we make substantial progress by considering other parameters:
our contribution includes three novel fixed-parameter algorithms for $\mathcal{F}^*$-\textsc{Branchwidth} that not only give a unified platform for computing multiple width parameters, but also push the boundaries of tractability specifically for the notoriously difficult problem of computing mim-width~\cite{HogemoTV19,SaetherV16,Yamazaki18}.

\begin{result}
\label{res:algo}
	Let $\cF^*$ be the union of some \texttt{primal} \texttt{ph} families.
$\mathcal{F}^*$-\textsc{Branchwidth:}
\begin{enumerate}
\item is fixed-parameter tractable when parameterized by treewidth and the maximum degree,
\item is fixed-parameter tractable when parameterized by treedepth, and
\item admits a linear kernel when parameterized by the feedback edge set number.
\end{enumerate}
\end{result}

\smallskip
\noindent \textbf{Organization of the Paper.} In Section~\ref{sec:fbw} we introduce $\mathcal{F}$-branchwidth.
Section~\ref{sec:overview} serves as an overview section in which we explain our technical contributions on a high-level.
Firstly, this includes substantiating $\mathcal{F}$-branchwidth as a unifying parameter by comparing it to previously studied parameters and identifying six \emph{size identifiable} graph classes as a foundation that covers the whole breadth of $\mathcal{F}$-branchwidth.
The full proofs and formal details for this are given in Section~\ref{sec:comparison}.
Secondly, Section~\ref{sec:overview} also includes an outline of our three unifying algorithms and their correctness proofs.
The full descriptions and proofs are then provided in Sections~\ref{sec:tw}, \ref{sec:td} and~\ref{sec:fes}.

\ifshort
\smallskip
\noindent
\emph{Statements where proofs or details are provided in the full version are marked with $\star$.}
\fi

	\section{Preliminaries}
	\label{sec:prelims}
	For an integer $i$, we let $[i]=\{1,2,\dots,i\}$.
We refer to the handbook by Diestel~\cite{Diestel12} for
standard graph terminology. We also refer to the standard books for a basic overview of parameterized complexity theory~\cite{CyganFKLMPPS15,DowneyFellows13}. Two graph parameters $a, b$ are \emph{asymptotically equivalent} if for every graph class $\mathcal{G}$, either both $a$ and $b$ are bounded by some constant, or neither is. \(\Delta(G)\) denotes the maximum vertex degree in \(G\). The \emph{length} of a path refers to the number of edges on that path.
\ifshort We denote by \(R(n_1, n_2,\ldots, n_q)\) the (multicolor) Ramsey number for \(q\) colors and clique size \(n_i\) for the color \(i\in[q]\)~\cite{graham1990ramsey}. ($\star$)
\fi

\ifshort
We assume familiarity with the graph parameters \emph{treedepth}~\cite{NesetrilOssonademendez12} and \emph{treewidth}~\cite{RobertsonSeymour86}, including the notion of \emph{nice tree-decompositions}~\cite{DowneyFellows13}. ($\star$)
A set $Q\subseteq E(G)$ of edges in a graph $G$ is called a \emph{feedback edge set} if $G-Q$ is acyclic. The \emph{feedback edge set number} of $G$ is the minimum size of a feedback edge set of $G$; this has been used to obtain fixed-parameter algorithms and/or polynomial kernels in a number of difficult problems~\cite{BetzlerBNU12,GanianOrdyniak21}.

\emph{Typical sequences} will be an important tool used in Section~\ref{sec:tw}; intuitively, these allow us to store a bounded-size snapshot of a sequence of bounded-size non-negative integers that preserves the extremes of the sequence. We refer readers interested in the technical details of that section to Bodlaender and Kloks' original work~\cite{BodlaenderKloks96} for the necessary definitions. ($\star$)
\fi

\iflong
\subparagraph*{Ramsey Theory.} Given an edge-colored graph \(G\) a subset of vertices \(C\subseteq V(G)\) is a \emph{monochromatic clique} of color \(c\) in \(G\) if the induced subgraph \(G[C]\) of \(G\) is complete and every edge in \(G[C]\) is colored by color \(c\). 

\begin{fact}[Ramsey's Theorem, see, e.g.,~\cite{graham1990ramsey}]\label{fact:ramsey}
	Let \(n_1, n_2, \ldots, n_q\in\Nat\). There exists an integer \(R\) such that for every edge-colored complete graph \(G\) on \(R\) vertices with \(q\) colors \(1,\ldots,q\) there exists \(i\in[q]\) such that \(G\) contains a monochromatic clique of color \(i\) and size \(n_i\). 
\end{fact}

We denote by \(R(n_1, n_2,\ldots, n_q)\) the minimum integer that satisfies the above fact and call such a number (multicolor) Ramsey Number. We note the following well known bounds on \(R(n_1, n_2,\ldots, n_q)\).

\begin{fact}[Erd{\H{o}}s and Szekeres~\cite{ErdosSzekeres1935}]
	\(R(n_1, n_2)\le \binom{n_1+n_2-1}{n_1-1}\le 2^{n_1+n_2}\).
\end{fact}

\begin{fact}
	\(R(n_1, n_2,\ldots, n_q)\le R(n_1, n_2,\ldots,n_{q-2}, R(n_{q-1},n_q))\).
\end{fact}

\smallskip
\noindent \textbf{Treewidth.}\quad
A \emph{nice tree-decomposition}~$\mathcal{T}$ of a graph $G=(V,E)$ is a pair 
$(T,\chi)$, where $T$ is a tree (whose vertices we call \emph{nodes}) rooted at a node $r$ and $\chi$ is a function that assigns each node $t$ a set $\chi(t) \subseteq V$ such that the following holds: 
\begin{itemize}[noitemsep]
	\item For every $uv \in E$ there is a node
	$t$ such that $u,v\in \chi(t)$.
	\item For every vertex $v \in V$,
	the set of nodes $t$ satisfying $v\in \chi(t)$ forms a subtree of~$T$.
	\item $|\chi(\ell)|=1$ for every leaf $\ell$ of $T$ and $|\chi(r)|=0$.
	\item There are only three kinds of non-leaf nodes in $T$:
	\begin{itemize}[noitemsep,label=]
        \item \textbf{Introduce node:} a node $t$ with exactly
          one child $t'$ such that $\chi(t)=\chi(t')\cup
          \{v\}$ for some vertex $v\not\in \chi(t')$.
        \item \textbf{Forget node:} a node $t$ with exactly
          one child $t'$ such that $\chi(t)=\chi(t')\setminus
          \{v\}$ for some vertex $v\in \chi(t')$.
        \item \textbf{Join node:} a node $t$ with two children $t_1$,
          $t_2$ such that $\chi(t)=\chi(t_1)=\chi(t_2)$.
	\end{itemize}
\end{itemize}

The \emph{width} of a nice tree-decomposition $(T,\chi)$ is the size of a largest set $\chi(t)$ minus~$1$, and the \emph{treewidth} of the graph $G$,
denoted $\tw(G)$, is the minimum width of a nice tree-decomposition of~$G$.
Efficient fixed-parameter algorithms are known for computing a nice tree-decomposition of near-optimal width~\cite{BodlaenderDDFLP16,Kloks94}.

\begin{proposition}[\cite{BodlaenderDDFLP16}]\label{fact:findtw}%
	There exists an algorithm which, given an $n$-vertex graph $G$ and an integer~$k$, in time $2^{\bigoh(k)}\cdot n$ either outputs a tree-decomposition of $G$ of width at most $5k+4$ and $\bigoh(n)$ nodes, or determines that $\tw(G)>k$.
\end{proposition}

\smallskip
\noindent \textbf{Treedepth and Feedback Edge Set.}\quad
	A \emph{rooted forest} is a forest in which every component has  a specified node called a \emph{root}. The \emph{closure} of a rooted forest $T$ is the graph on $V(T)$ in which two vertices are adjacent if and only if one is a descendant of the other in $T$.
The \emph{height} of a rooted forest is the number of vertices
in  a longest root-to-leaf path in $T$.
The \emph{treedepth} of a graph  $G$, denoted by $\td(G)$, is
the minimum height of a rooted forest whose closure contains $G$ 
as a subgraph.

\begin{fact}[\cite{NesetrilOssonademendez12}]\label{fact:compute-td}
  Given an $n$-vertex graph $G$ and a constant $w$, it is possible to
  decide whether $G$ has treedepth at most $w$, and if so, to compute an
  optimal treedepth decomposition of $G$ in time $\bigoh(n)$.
\end{fact}

A set $Q\subseteq E(G)$ of edges in a graph $G$ is called a \emph{feedback edge set} if $G-Q$ is acyclic. The \emph{feedback edge set number} of $G$ is the minimum size of a feedback edge set of $G$ and can be computed efficiently.

\begin{fact}
\label{fact:fescomp}
A minimum feedback edge set of a graph $G$ can be obtained by deleting the edges of minimum spanning trees of all connected components of $G$, and hence can be computed in time $\bigoh(|E(G)|+|V(G)|)$.
\end{fact}

Both measures have been used to obtain fixed-parameter algorithms and/or polynomial kernels in a number of applications~\cite{BetzlerBNU12,GutinJW16,GanianO18,GanianOrdyniak21}.
\fi

\smallskip
\noindent \textbf{Branchwidth and Branch Decompositions.}\quad
While branchwidth and branch decompositions have also been defined over general ground sets, here we introduce the definition that matches the most common usage scenario on graphs.
A set function $f:2^M\rightarrow \mathbb{Z}$ over a ground set $M$ is called \emph{symmetric}
if $f(X)=f(M\setminus X)$ for all $X\subseteq M$.
A tree is \emph{subcubic} if all its nodes have degree at most~$3$.

A \emph{branch decomposition} of a graph $G$ is a subcubic tree $\decT$ whose leaves are mapped to the vertices of $G$ via a bijective mapping $\decf$. 
For an edge $e$ of $\decT$, the connected components of $\decT - e$ induce a bipartition $(X,Y)$ of the set of leaves of $\decT$, and through $\decf$ also of $V(G)$. For a symmetric function $f:2^{V(G)}\rightarrow \mathbb{Z}$, the \emph{width} of $e$ is $f(X)$, the width of the branch decomposition $\decT$ is the maximum width over all edges of $\decT$, and the branchwidth of $G$ is the minimum width of a branch decomposition of $G$. $f$ is often called the \emph{cut-function}, and we use $G[X,Y]$ to denote the bipartite graph obtained from $G$ by removing all edges that do not have precisely one endpoint in $X$ and precisely one endpoint in \(Y\).

\newcommand{\lcac}[2]{\operatorname{restr}_{#1}(#2)}
\ifshort
	For a tree \(T\) and \(A \subseteq V(T)\) we define the \emph{restricted tree} of \(T\) w.r.t.\ \(A\) as the tree
	\(\lcac{T}{A}\) that arises from the minimal subtree of $T$ containing all vertices in $A$ by recursively contracting all edges $uv$ where $u$ has degree $2$ and is not in $A$.
	Note also that every edge \(xy \in E(\lcac{T}{A})\) can be understood as the contraction of the path \(P\) between \(x\) and \(y\) in \(T\).
	We also say that the edge \(xy\) of the restricted tree with respect to \(X\) of \(T\) \emph{corresponds} to \(P\).
\fi

\iflong
Well-known applications of this definition include mim-width~\cite{vatshelle2012new,JaffkeKT20}, where the cut-function is the size of a maximum induced matching in $G[X,Y]$. Another example is maximum-matching width, which is asymptotically equivalent to treewidth~\cite{JeongST18} and where the cut-function is the size of a maximum (but not necessarily induced) matching in $G[X,Y]$. Rank-width is another well-known example of branchwidth; there, the cut-function is the rank of the bipartite adjacency matrix induced by $(X,Y)$. Rank-width is asymptotically equivalent to clique-width~\cite{OumSeymour06}.

	In Sections~\ref{sec:tw} and \ref{sec:td}, it will be useful to consider the \emph{restrictions} of (the tree structure of) a branch decomposition \((\decT,\decf)\).
	For a tree \(T\) and a set \(A \subseteq V(T)\) we define the \emph{restricted tree} of \(T\) with respect to \(A\) as the tree
	\(\lcac{T}{A}\) that arises from the minimal subtree of 
	\(T\) containing all vertices in $A$ by recursively contracting all edges $uv$ where $u$ has degree $2$ and is not in $A$.
	
	Note that for a branch decomposition \((\decT,\decf)\) the restricted tree with respect to subset of leaves \(A \subseteq \decf(V(G))\) of a branch decomposition of \(G\) can be seen as a branch decomposition of \(G[A]\).
	Conversely a branch decomposition \((\decT_1,\decf_1)\) of \(G\) \emph{extends} a branch decomposition \((\decT_2,\decf_2)\) of \(G[A]\), if \(\lcac{\decT_1}{A} = \decT_2\) and \(\decf_1(a) = \decf_2(a)\) for \(a \in A\).\\
	Note also that every edge \(xy \in E(\lcac{\decT}{A})\) can be understood as the contraction of the path \(P\) between \(x\) and \(y\) in \(\decT\).
	We also say that the edge \(xy\) of the restricted tree with respect to \(X\) of \(\decT\) \emph{corresponds} to \(P\).
	
To streamline the presentation in some sections (such as Section~\ref{sec:tw}), we suppress the function $\decf$ of a branch decomposition \(\decT\) by assuming that the leaves of branch decompositions are actually the leaves of $G$.
	\fi

\iflong
\smallskip
\noindent \textbf{Typical Sequences.}\quad
Typical sequences introduced in \cite{BodlaenderKloks96} will be an important tool in the \FPT\ algorithm for computing \(\mathcal{F}^*\)-branchwidth parameterized by treewidth and vertex degree.
We provide a brief overview over important definitions and results.

\begin{definition}[\cite{BodlaenderKloks96}]
	Let \(s = (s_1, \dotsc, s_\ell)\) be a sequence of natural numbers (including \(0\)) of length \(\ell\).  The \emph{typical sequence} \(\tau(s)\) of \(s\) is obtained from \(s\) by an iterative exhaustive application of the following two operations:
	\iflong 
	\begin{enumerate}
		\item Removing consecutive repetitions:
		If there is an index \(i \in [\ell - 1]\) such that \(s_i = s_{i + 1}\), we change the sequence \(s\) from \((s_1, \dotsc, s_\ell)\) to \((s_1, \dotsc, s_i, s_{i + 2}, \dotsc, s\ell)\).
		\item Typical  operation:
		If there are \(i,j \in [\ell]\)  such  that \(j - i \geq 2\) and for  all \(i \leq k \leq j\), \(s_i \leq s_k \leq s_j\), or for all \(i \leq k \leq j\), \(s_i \geq s_k \geq s_j\) then we change the sequence \(s\) from \(s_1, \dotsc, s_\ell\) to \(s_1, \dotsc, s_i, s_j, \dotsc, s_\ell\), i.e.\ we remove all entries of \(s\) (strictly) between index \(i\) and \(j\).
	\end{enumerate}
	\fi 
	\ifshort 
		\emph{1) Removing consecutive repetitions}:
		If there is an index \(i \in [\ell - 1]\) such that \(s_i = s_{i + 1}\), we change the sequence \(s\) from \((s_1, \dotsc, s_\ell)\) to \((s_1, \dotsc, s_i, s_{i + 2}, \dotsc, s\ell)\). \emph{2) Typical  operation}:
		If there are \(i,j \in [\ell]\)  such  that \(j - i \geq 2\) and for  all \(i \leq k \leq j\), \(s_i \leq s_k \leq s_j\), or for all \(i \leq k \leq j\), \(s_i \geq s_k \geq s_j\) then we change the sequence \(s\) from \(s_1, \dotsc, s_\ell\) to \(s_1, \dotsc, s_i, s_j, \dotsc, s_\ell\), i.e.\ we remove all entries of \(s\) (strictly) between index \(i\) and \(j\).
	\fi
	We say a sequence of natural numbers is \emph{typical} whenever it is identical to its typical sequence.
\end{definition}
While it is maybe not immediately obvious from the definition, the typical sequence of a sequence of natural numbers is always unique.
As we use typical sequences to index table entries or \emph{records} in a dynamic program in Section~\ref{sec:tw} and also branch on ways to modify them, the following combinatorial bounds will be useful.
\ifshort 
\begin{fact}[\cite{BJT20, BodlaenderKloks96, TSB05a}]
	\label{fact:typ-bounds+comp} 
		The length of the typical sequence of a sequence with entries in \([k]\) is at most \(2k + 1\).
		The number of typical sequences with entries in \([k]\) is at most \(\frac{8}{3}2^{2k}\).
		Moreover, the typical sequence of a sequence of natural numbers of length \(\ell\) can be computed in time in \(\mathcal{O}(\ell)\) and the interleaving of two typical sequences with entries in \([k]\) can be computed in time in \(\mathcal{O}(k^42^{4k(2k + 1)})\).
\end{fact}
\fi
\iflong
\begin{fact}[\cite{BodlaenderKloks96}]
	\label{fact:typ-bounds} 
	\begin{itemize}
		\item The length of the typical sequence of a sequence with entries in \(\{0\} \cup [k]\) is at most \(2k + 1\).
		\item The number of typical sequences with entries in \(\{0\} \cup [k]\) is at most \(\frac{8}{3}2^{2k}\).
	\end{itemize}
\end{fact}
It is known that typical sequences of a sequence of natural numbers can be computed in linear time.
\begin{fact}[\cite{BJT20}]
	\label{fact:typ-comp}
	The typical sequence of a sequence of natural numbers of length \(\ell\) can be computed in time in \(\mathcal{O}(\ell)\).
\end{fact}
For the correctness of our dynamic program it is helpful to note some trivial observations:
The typical sequence of a concatenation of two sequences \(s\) and \(s'\) is the typical sequence of the concatenation of the typical sequences of \(s\) and \(s'\).
Moreover, the typical sequence of a sequence \((s_1 + z, \dotsc, s_\ell + z)\) that arises from a typical sequence by adding some constant \(z \in \mathbb{N}\) to every entry of a sequence \(s\) is equal to the sequence that arises by adding \(z\) to each entry of the typical sequence of \(s\).
We also write \(s + z\) for a sequence \(s\) of length \(\ell\) and a natural number \(z\) to denote the sequence \((s_1 + z, \dotsc, s_\ell + z)\).

A definition which will be crucial at the join nodes in our dynamic program is that of the \(\oplus\)-operator for typical sequences.
\begin{definition}[\cite{TSB05a,BodlaenderKloks96}]
	Let \(s\) be a sequence of \(\ell\) natural numbers.  We define the set \(E(s)\) of \emph{extensions} of \(s\) as the set of sequences that are obtained from \(s\) by repeating each of its elements an arbitrary number of times, and at least once.
	Now we define the \emph{interleaving} of two typical sequences \(s\) and \(t\) of arbitrary length to be the set \(s \oplus t = \{\tau((\tilde{s}_1 + \tilde{t}_1, \dotsc, \tilde{s}_\ell + \tilde{t}_\ell)) \mid \tilde{s} \in E(s), \tilde{t} \in E(t), |\tilde{s}| = |\tilde{t}| = \ell\}\), where we use \(|r|\) to denote the length of a sequence \(r\).
\end{definition}

\begin{fact}[\cite{TSB05a}]
	\label{fact:typ-compint}
	The interleaving of two typical sequences with entries in \([k]\) can be computed in time in \(\mathcal{O}(k^42^{4k(2k + 1)})\).
\end{fact}
\fi 
\fi

	\section{$\cF$-Branchwidth}	
	\label{sec:fbw}
	In this paper we consider branchwidth and branch decompositions where the value of the cut-function $f(X)$ is defined by 
	the maximum number of vertices of any
	induced subgraph of $G[X,V(G)\setminus X]$ that is isomorphic to some graph in some fixed infinite family of bipartite graphs. 
	We focus on families of bipartite graphs where each part has the same size satisfying a certain closure property.
	Let \(\cF\) be an infinite family of bipartite graphs where each graph in $\cF$ has bipartition \((A=(a_1,\ldots, a_n),B= (b_1,\ldots, b_n))\). 
	We say that \(\cF\) is \ph\ (\phsymb) if for each \(2n\)-vertex graph in \(\cF\) and each subset \(L\) of \([n]\), the graph induced on \(\{a_i,b_i | i\in L\}\) is isomorphic to a graph in \(\cF\). 
	Note that the vertices in both parts of the bipartite graph in \(\cF\) are ordered.
	Given a bipartite graph $H=(V,E)$ with bipartition \((A=(a_1,\ldots, a_n),B= (b_1,\ldots, b_n))\), for each $i\in [n]$ we call \(a_i\) and \(b_i\) \emph{partners}. Moreover, for simplicity, when referring to a bipartite graph $H$ with bipartition \((A,B)\) and edge set $E$ we will sometimes write $H=(A,B,E)$.
	
	For a \phsymb\ family of bipartite graphs \(\cF\), we define \emph{\(\cF\)-branchwidth} as the branchwidth where the cut-function $\decf$ of a cut \((X,Y)\) is defined as the largest \(n\) such that a \(2n\)-vertex graph in \(\cF\) is isomorphic to an induced subgraph of \(G[X,Y]\). For a graph \(G\), a branch decomposition \(\decT\) of \(G\), and an edge \(e\) of \(\decT\), we denote by \(\FbwSh(\decT,e) \) the width of the cut of \(G\) \emph{induced} by \(e\), by \(\FbwSh(\decT)\) the width of \(\decT\) and we denote by \(\FbwSh(G)\) the \(\cF\)-branchwidth of \(G\).
	\iflong
	We begin by noting the following simple observations that follow directly from the fact that the width of an edge depends only on the existence of some graph in \(\cF\) as induced subgraph across the cut. 
	
	\begin{observation}\label{obs:closedInducedSubgraphs}
		Let \(\cF\) be a \emph{\phsymb} family of graphs, \(G\) a graph, and \(H\) an induced subgraph of \(G\), then \(\FbwSh(H)\le \FbwSh(G) \). 
	\end{observation}
	
	\begin{observation}\label{obs:closedSubFamilies}
		Let \(\cF, \cF'\) be two \emph{\phsymb} families of graphs such that \(\cF'\subseteq \cF\) and let \(G\) be a graph. Then \(\cF'\mbox{-}\bw(G)\le \FbwSh(G) \). 
	\end{observation}

	We say that two  families of graphs \(\cF_1\) and \(\cF_2\)
	are \emph{isomorphic} if for every graph \(H_1\) in \(\cF_1\), the class \(\cF_2\) contains a graph isomorphic to \(H_1\) and for every graph \(H_2\) in \(\cF_2\), the class \(\cF_1\) contains a graph isomorphic to \(H_2\). 
	Since the ordering on the vertices in \(G[X,Y]\) for a partition \((X,Y)\) of the vertices of \(G\) is not important to determine the width of the cut, we get the following observation. 
	
	\begin{observation}\label{obs:isomorphicFamilies}
		Let \(\cF_1, \cF_2\) be two isomorphic \emph{\phsymb} families of graphs and \(G\) a graph, then \(\cF_1\mbox{-}\bw(G)=  \cF_2\mbox{-}\bw(G) \). 
	\end{observation}
	\fi
	
It will be useful to introduce the following six special families of \phsymb\ classes. 
	In the following, we denote a bipartite graph with vertex bipartition $(A, B)$ and edge set $E$ by $(A, B, E)$.
	Let:
	
	\begin{enumerate}
		\item $H^n_\emptyset=((a_1,\ldots,a_n), (b_1,\ldots,b_n), \emptyset)$ and let \(\iemptyFamily = \{H^n_\emptyset \mid n\in \Nat \} \), i.e.,  \iemptyFamily contains all edgeless graphs on even number of vertices. 
		
		\item $H^n_{\imatch} = ((a_1,\ldots,a_n), (b_1,\ldots,b_n), E_{\imatch})$ where \(E_{\imatch} = \{a_ib_i\mid i\in [n] \} \)  
		and let \(\imatchFamily = \{H^n_{\imatch} \mid n\in \Nat \} \), i.e.,  \imatchFamily contains all \(1\)-regular graphs (matchings). 
		\item $H^n_{\ichain}= ((a_1,\ldots,a_n), (b_1,\ldots,b_n), E_{\ichain})$ where \(E_{\ichain} = \{a_ib_j\mid i\le j\in [n] \} \) and let \(\ichainFamily = \{H^n_{\ichain} \mid n\in \Nat \} \), i.e.,  \ichainFamily contains all bipartite chain graphs without twins. 
		
		\item $H^n_{\ichainstrict}= ((a_1,\ldots,a_n), (b_1,\ldots,b_n), E_{\ichainstrict})$ where \(E_{\ichainstrict} = \{a_ib_j\mid i < j\in [n] \} \) and let \(\ichainstrictFamily = \{H^n_{\ichain} \mid n\in \Nat \} \), i.e.,  \ichainstrictFamily contains each graph in \ichainFamily without the matching.
		\item $H^n_{\iantimatch}= ((a_1,\ldots,a_n), (b_1,\ldots,b_n), E_{\iantimatch})$ where \(E_{\iantimatch} = \{a_ib_j\mid i\neq j\in [n] \} \) and let \(\iantimatchFamily = \{H^n_{\iantimatch} \mid n\in \Nat \} \), i.e.,  \iantimatchFamily contains for every \(n\) the \(2n\)-vertex \((n-1)\)-regular bipartite graph (all anti-matchings). 
		\item $H^n_{\icomplete}= ((a_1,\ldots,a_n), (b_1,\ldots,b_n), E_{\icomplete})$ where \(E_{\icomplete} = \{a_ib_j\mid i, j\in [n] \} \) and let \(\icompleteFamily = \{H^n_{\icomplete} \mid n\in \Nat \} \), i.e.,  \icompleteFamily contains all complete bipartite graphs $K_{n,n}$.
	\end{enumerate}
	
 We will later prove that to compute an approximately-optimal branch decomposition for any \phsymb\ family \(\cF\), it suffices to compute an optimal branch decomposition for  \(\cF^*\), where \(\cF^*\) is the union of a subset of classes \imatchFamily, \ichainFamily, and \iantimatchFamily\ that are fully contained in \(\cF\) (see Lemma~\ref{lem:primal}).

\section{Overview of Technical Contributions}\label{sec:overview}

We now give an overview of our technical contributions and sketch the main ideas of our more involved proofs. There is a one to one correspondence between the subsections of this section and the remaining technical sections. The proofs of the statements in this section are postponed to the later sections. For the convenience of the reader, we use the same numbering of statements and definitions as in the corresponding full technical sections.

\subsection{Size-Identifiable Classes and Comparison to Existing Measures}
Here, we focus on the six families introduced at the end of Section~\ref{sec:fbw}. We first show that while there is an infinite number of possible \phsymb\ graph classes, to asymptotically characterize $\Fbw$ for any \phsymb\ family of graphs \(\cF\) we only need to consider these families. Afterwards, we turn our attention on several previously studied graph measures and show how they compare to $\Fbw$. Finally we show that to compute an approximate \(\cF^*\)-branch decomposition, it suffices to focus on only three of these classes, namely \(\imatchFamily\), \(\ichainFamily\), and \(\iantimatchFamily\).
\subsubsection{Size-Identifiable Classes}

We say that \(\cF\) is \si\ (\sisymb) if for each \(n\in \Nat\), there is a single \(2n\)-vertex graph in \(\cF\) (up to isomorphism). We will show \(\cF_\emptyset\), \(\imatchFamily\), \(\ichainFamily\), \(\ichainstrictFamily\), \( \iantimatchFamily\), and \(\icompleteFamily\) are the only \sisymb\ \phsymb\ families of graph. To show that these are the only \sisymb\ \phsymb\ families as well as to show importance of these families, we will make use of the following lemma that basically shows that any bipartite graph with ordered partitions contains a large subset of partners that induce a graph isomorphic to a graph in one of the six above families.

\begin{restatable*}{lemma}{ramseyPartners}
	\label{lem:ramseyPartners}
		Let $q \ge R(2n,2n,2n,2n)$ and let \(H^q = ((a_1,\ldots, a_q), (b_1,\ldots, b_q), E)\) be a bipartite graph. Then there exists $L\subseteq [q]$ such that $|L|= n$ and the graph induced on \(\{a_i,b_i\mid i\in L\}\) is isomorphic to one of $H^n_\emptyset$, $H^n_{\imatch}$, $H^n_{\ichain}$, $H^n_{\ichainstrict}$, $H^n_{\iantimatch}$, or $H^n_{\icomplete}$.
	\end{restatable*}

\begin{restatable*}{lemma}{onlysix}
	\label{lem:onlysix}
		There are precisely six \emph{\sisymb} and \emph{\phsymb} families: \iemptyFamily, \imatchFamily, \ichainFamily, \ichainstrictFamily, \iantimatchFamily, \icompleteFamily.
\end{restatable*}

We now apply Ramsey's Theorem to show that these six \texttt{\emph{si}} classes capture the whole hierarchy of width parameters induced by $\mathcal{F}$-branchwidth.

\begin{restatable*}{lemma}{sigood}
	\label{lem:sigood}
		Let \(\cF\) be a \emph{\phsymb} graph class and let \(\cF'\) be the union of all \emph{\sisymb} families fully contained in \(\cF\). Then \(\cF\)-branchwidth is asymptotically equivalent to \(\cF'\)-branchwidth.
\end{restatable*}

\subsubsection{Comparison to Existing Measures}
Next, we discuss how specific \texttt{\emph{si ph}} classes capture previously studied parameters.
First, mim-width is precisely \imatchFamily-branchwidth. 

\begin{restatable*}{observation}{mimWidth}
	\label{obs:mim-width}
	For every graph \(G\), mim-width of \(G\) is equal to \imatchFamily-branchwidth of \(G\).
\end{restatable*}

The next decomposition parameter that we consider is treewidth. 
To show the equivalence, we actually show that \emph{maximum-matching width}, which is always linearly upper- and lower-bounded by treewidth~\cite{vatshelle2012new,JeongST18}, is also \(\Fbw\) for the \phsymb\ family \(\cF\) of all \texttt{\emph{ph}} graphs that contain a matching on the partners.

\begin{restatable*}{lemma}{twbound}
	\label{lem:twbound}
There exists a \emph{\phsymb} family \(\cF\) such that for all graphs \(G\) it holds \(\FbwSh(G)\le \tw(G) + 1\le 3\cdot \FbwSh(G)\).
\end{restatable*}

By Lemma~\ref{lem:sigood}, treewidth is then asymptotically equivalent to $\mathcal{F}^*$-branchwidth for $\mathcal{F}^*$ being the union of some \texttt{\emph{si ph}} classes, notably \imatchFamily, \ichainFamily, \iantimatchFamily, and \(\icompleteFamily\).
By applying the easy observation that \(\cF'\subseteq \cF\) implies \(\cF'\mbox{-}\bw(G)\le \FbwSh(G)\) for any \phsymb\ families \(\cF\) and \(\cF'\) and any graph \(G\),
we can also establish an explicit upper bound when considering \imatchFamily, \ichainFamily and \iantimatchFamily.
The reason we are interested in specifically this set of families will become clear when we define \emph{\primal} families.

\begin{restatable*}{corollary}{FisTWBounded}
	\label{cor:FisTWBounded}
	Let \(\cF^*\) be a union of any combination of families among \imatchFamily, \ichainFamily, and \iantimatchFamily and let \(G\) be a graph. Then \(\cF^*\mbox{-}\bw(G)\le\tw(G)+1\).
\end{restatable*}
		
Throughout Section~\ref{sec:tw} and in Proposition~\ref{prop:manycomponent} in Section~\ref{sec:td}, the fact that $\Fsbw(G)$ is bounded in the respective parameters is essential, and Corollary~\ref{cor:FisTWBounded} gives this bound.

The next parameter that we will consider is clique-width. Similarly as for treewidth, it will be easier to show equivalence with another parameter, \emph{module-width}, which is known to be asymptotically equivalent to cliquewidth~\cite{Rao08,BuiXuanSTV13}. 
An important tool to establish this equivalence is Corollary 2.4 in \cite{DingOOV96}, which intuitively states that any sufficiently large twin-free bipartition of vertices in a graph must contain a large chain, anti-matching or matching.

\begin{restatable*}{lemma}{compModuleWidth}
	\label{lem:compModuleWidth}
	Let \(\cF= \imatchFamily\cup \iantimatchFamily \cup \ichainFamily\). Then there exists a computable function $f$ such that \(\FbwSh(G)\le \operatorname{modw}(G) \le f(\FbwSh(G))\) for every graph $G$; in particular, module-width and $\cF$-branchwidth with \(\cF= \imatchFamily\cup \iantimatchFamily \cup \ichainFamily\) are asymptotically equivalent. 
\end{restatable*}

Finally, we compare \(\Fbw\) to the so-called \(H\)-join decompositions introduced by Bui-Xuan, Telle, and Vatshelle~\cite{BuiXuanTV10}.

We note that unlike other parameters, $H$-join decompositions do not have a specific ``width''; each graph either admits a decomposition or not, and the complexity measure is captured by the order of the graph $H$. Keeping this in mind, we can show that $\mathcal{F}$-branch captures \(H\)-joins as well.

\begin{restatable*}{lemma}{compHJoin}
	\label{lem:compHJoin}
	Let \(\cF= \imatchFamily\cup \iantimatchFamily \cup \ichainFamily \) and let \(H\) be a fixed graph. Then for all \(H\)-join decomposable graphs \(G\) it holds that \(\cF\mbox{-}\bw(G)\le |V(H)|\).	
\end{restatable*}

\subsubsection{Primal Families and Computing $\mathcal{F}$-Branchwidth}

Let us call the \phsymb\ families \(\imatchFamily\), \(\ichainFamily\), and \(\iantimatchFamily\) \primal. As the last result of this subsection, we can build on Lemma~\ref{lem:sigood} and show that it suffices to focus our efforts to compute good $\mathcal{F}$-branch decompositions merely on unions of \primal\ classes. To this end, one can prove the following simple claim.

\begin{restatable*}{lemma}{antichain}
	\label{lem:antichain}
	Let $\cF$ be a union of \texttt{si ph} classes. If \(\ichainstrictFamily\subseteq \cF\), then for every graph \(G\) it holds that \(|\FbwSh(G) - \cF'\operatorname{-bw}(G)|\le 1\), where \(\cF' = (\cF\setminus\ichainstrictFamily)\cup\ichainFamily \).
\end{restatable*}

Our high-level strategy for dealing with \(\icompleteFamily\) and \(\iemptyFamily\) is to start by considering an arbitrary optimal branch decomposition $\decT$ of a graph $G$ for the class $\cF^*$ obtained after removing these two families from $\cF$. Intuitively, we show that if there were a large difference between the $\cF$-branchwidth of $\decT$ (which is equal to the \(\cF^*\)-branchwidth of \(G\)) and the $\cF$-branchwidth of $G$, then $G$ or the complement of \(G\) would have to contain a sufficiently large complete bipartite subgraph; but then \emph{any} decomposition of the graph would still necessarily have a cut inducing a fairly large complete bipartite graph or edgeless graph. This would contradict our assumption about the $\cF$-branchwidth of $G$ being substantially smaller than the $\cF$-branchwidth of $\decT$. Using this, we obtain:

\begin{restatable*}{lemma}{primalLemma}
	\label{lem:primal}
	Let $\cF$ be a union of a subset of \texttt{si ph} classes other than \(\ichainstrictFamily\). Let \(\cF^*\) be the union of all \emph{\primal} families in \(\cF\). Then an optimal \(\cF^*\)-branch decomposition of a graph $G$ is a \(3\)-approximate \(\cF\)-branch decomposition of $G$.
\end{restatable*}

We use $\cF^*$ to refer to the union of any combination of \texttt{\emph{primal}} families, and $\Fsbw(G)$ to refer to the $\cF^*$-branchwidth of $G$. If follows from Lemmas~\ref{lem:antichain} and \ref{lem:primal} that computing \(\cF^*\)-branchwidth is sufficient to obtain approximately optimal decompositions for all possible $\cF$-branchwidths.
We conclude this section by showing that for \textsc{$\cF^*$-Branchwidth}, i.e., the problem of computing an optimal $\cF^*$-branch decomposition for a fixed choice of $\cF^*$, it is sufficient to deal with connected graphs only.

\begin{restatable*}{lemma}{componentLemma}
	\label{lem:component}
	Let $\mathcal{F}^*$ be a union of a non-empty set of \emph{\primal} families of graphs.
	For a graph $G$, the $\mathcal{F}^*$-branchwidth of $G$ is the maximum $\mathcal{F}^*$-branchwidth over all connected components of $G$.
\end{restatable*}

\subsection{Treewidth and Maximum Vertex Degree}
\label{subsec:tw}
Our aim in Section~\ref{sec:tw} is to prove the following theorem.

\begin{restatable*}{theorem}{twDegTheorem}
	\label{thm:tw_deg}
	Let \(\mathcal{F}^*\) be the union of some \emph{\primal} \emph{\phsymb} families.
	\Fsbwprob is \FPT\ parameterized by the treewidth and the maximum degree of the input graph.
\end{restatable*}

As an immediate corollary, by considering \(\mathcal{F}^* = \mathcal{F}_{\imatch}\) we obtain:
\begin{restatable*}{corollary}{twDegMimWidthCorollary}
	\label{thm:tw_deg}
	\textsc{Mim-Width} is \FPT\ parameterized by the treewidth and the maximum degree of the input graph.
\end{restatable*}

On a high level, our proof of Theorem~\ref{thm:tw_deg} relies on the usual dynamic programming approach 
which computes a set of records for each node \(t \in V(T)\) of a minimum-width nice tree decomposition \((T, \chi)\) of \(G\) in a leaves-to-root manner.
Here each record captures a set of partial \(\mathcal{F}^*\)-branch decompositions not exceeding \(\mathcal{F}^*\)-branchwidth at most \(\tw(G)\).
It is worth noting that the most involved and difficult parts of our algorithm are necessary to handle the case that \(\mathcal{F}_{\imatch} \subseteq \mathcal{F}^*\).

For \(t \in V(T)\), denote by \(T_t\) the subtree of \(T\) rooted at \(t\), and by \(G_t = G[\bigcup_{s \in V(T_t)} \chi(s)]\) the subgraph of \(G\) induced by the vertices in the bags of \(T_t\).
Moreover, for \(v \in V(G)\) we denote its \emph{closed distance-\(i\) neighbourhood} in \(G\) by \(N^i[v]\), its \emph{open distance-\(i\) neighbourhood in \(G\) by \(N^i(v)\)}, and also extend these notations to vertex sets \(V' \subseteq V(G)\) by denoting \(N^i[V'] = \bigcup_{v \in V'} N^i[v]\) and \(N^i(V') = N^i[V'] \setminus V'\).

We define a \emph{record} \((D, \flat, \Lambda, \sigma, \alpha^{\imatch}, \alpha^{\iantimatch}, \alpha^{\ichain})\) of \(t \in V(T)\) to consist of the following.
\begin{itemize}[leftmargin=*]
	\item A binary tree \(D\) on at most \(2^{|N^3[\chi(t)]|}\) vertices, all leaves of which are identified with distinct vertices in \(N^3[\chi(t)]\).
	Intuitively, \(D\) will describe the restricted tree with respect to \(N^3[\chi(t)]\) of all branch decompositions of \(G_t\) that are captured by the record.
	\item A set \(\flat = \{\flat_e\}_{e \in E(D)}\) indexed by the edges of \(D\), where each \(\flat_e\) is a sequence of subsets of \(N^{3\tw(G)}(N^3[\chi(t)])\), such that \(\bigcup_{e \in E(D)} \flat_e\) is a partition of \(N^{3\tw(G)}(N^3[\chi(t)])\).
	\(\flat_e\) describes the order in which subtrees containing vertices in \(N^{3\tw(G)}(N^3[\chi(t)])\) are attached to the path which corresponds to \(e\) in any branch decomposition captured by \((D, \flat, \Lambda, \sigma, \alpha^{\imatch}, \alpha^{\iantimatch}, \alpha^{\ichain})\).
	The entry \(\flat\) can be interpreted to express a ``refinement'' of each edge of \(D\).
	Similarly to edges of \(D\) corresponding to paths in a branch decomposition captured by the record
	we will define refinements of these paths in such branch decompositions into \emph{blocks} which correspond to pairs \((e,i)\) where \(e \in E(D)\) and~\(i \in [|\flat_e|+1]\).
	\item \(\Lambda = (\Lambda_e)_{e \in E(D)}\) indexed by the edges of \(D\), where each \(\Lambda_e\) is a sequence \(\Lambda_e = (\lambda_1, \dotsc, \lambda_{|\flat_e| + 1})\) of subsets of \(N^3(\chi(t))\).
	\(\lambda_i \in \Lambda_e\) should be such that the \imatchwidth value at every edge in the block corresponding to \((e,i)\) in any branch decomposition captured by \((D, \flat, \Lambda, \sigma, \alpha^{\imatch}, \alpha^{\iantimatch}, \alpha^{\ichain})\) can be achieved by a matching such that its intersection with \(G[N^3[\chi(t)]]\) is incident to precisely the vertices in \(\lambda_i\).
	In other words, \(\lambda_i \in \Lambda_e\) describes the interaction we assume any matching to have with \(G[N^3[\chi(t)]]\) when computing the \imatchwidth value at any edge in the block corresponding to \((e,i)\).
	We can argue (Lemma~\ref{lem:blockset}) that such a set \(\lambda_i\) exists.
	\item A collection \(\sigma = (\sigma_{e,i})_{e \in E(D), i \in [|\flat_e| + 1]}\) of typical sequences \(\sigma_{e,i}\) with entries in \(\{0\} \cup [\tw(G) + 1]\).
	\(\sigma_{e,i}\) should describe a further refinement of each path corresponding to \((e,i)\) in a branch decomposition captured by \((D, \flat, \Lambda, \sigma, \alpha^{\imatch}, \alpha^{\iantimatch}, \alpha^{\ichain})\) according to the \imatchwidth values that can be achieved by the restriction of the branch decomposition to \(G_t\) under the condition imposed by the choice of \(\lambda_i\).
	The importance of these typical sequences is that we can argue that it is safe to assume that a partial branch decomposition at node \(t\) can be extended to a minimum-\(\Fsbw\) decomposition of \(G\) if and only if this is possible by attaching vertices only at positions in the typical sequences (Lemma~\ref{lem:typical}).
	\item \(\alpha^{\imatch}, \alpha^{\iantimatch}, \alpha^{\ichain} \in [\tw(G) + 1]\) should describe the \(\imatchwidth\), \(\iantimatchwidth\) and \(\ichainwidth\) of any branch decomposition of \(G_t\) captured by \((D, \flat, \Lambda, \sigma, \alpha^{\imatch}, \alpha^{\iantimatch}, \alpha^{\ichain})\).
	For \(\alpha^{\imatch}\) there is a slight caveat -- to compute this entry in the record we will rely on the correctness of \(\Lambda\) which we cannot ensure at the time of computation, but rather retroactively after the root of \(T\) is processed.
	This is done by constructing a \emph{canonical decomposition} concurrently with each record, which is representative of all branch decompositions captured by \((D, \flat, \Lambda, \sigma, \alpha^{\imatch}, \alpha^{\iantimatch}, \alpha^{\ichain})\), and then after processing the root of \(T\) verifying that its \imatchwidth is at most \(\alpha^{\imatch}\).
\end{itemize}

Note that the number of records is easily seen to be \FPT\ parameterized by \(\tw(G)\) and \(\Delta(G)\) using 
the fact that the number of considered typical sequences is at most \(\frac{8}{3}2^{2(\tw(G) + 1)}\)~\cite{BodlaenderKloks96}.
We remark that the information stored in \(D\) is sufficient to keep track of the \iantimatchwidth and the \ichainwidth of the constructed partial decomposition (this means to update \(\alpha^{\iantimatch}\) and \(\alpha^{\ichain}\)) appropriately.
For \imatchwidth the ``forgotten'' vertices, i.e.\ vertices in \(V(G_t) \setminus \chi(t)\) have a more far-reaching impact.
More explicitly, any graphs in \(\mathcal{F}_{\iantimatch}\) and \(\mathcal{F}_{\ichain}\) induced in a cut of a partial branch decomposition which contains a vertex in \(\chi(t)\), can only contain vertices in \(N^3[\chi(t)]\).

Obviously, this is not true for graphs in \(\mathcal{F}_{\imatch}\) which can have infinite diameter.

Hence, for \imatchwidth instead of keeping track of vertices which can occur in graphs in \(\FFF^*\) together with vertices in \(\chi(t)\),
we store the rough location in the constructed partial branch decomposition of a sufficiently large neighbourhood \(\tilde{N}\) of \(\chi(t)\) in \(G\) to separate the interaction (in terms of independence) of edges of \(G_t[V(G_t) \setminus \tilde{N}]\) and edges of \(G[V(G) \setminus (V(G_t) \cup \tilde{N})]\) that occur in independent induced matchings at cuts in the constructed partial branch decomposition.
After this separation we still need to be able to actually derive the \imatchwidth.
To do this we use typical sequences in a similar way as they are also used to compute e.g.\ treewidth, pathwidth and cutwidth~\cite{BodlaenderKloks96,TSB05a,Hamm19,BJT20}, as indicated in the description of \(\sigma\) above.

In the remainder of this subsection we give an overview of the technical results and formal definitions that justify our approach.
For this fix a node \(t \in V(T)\) of a tree decomposition \((T,\chi)\) of \(G\),
let \(\decT\) be a branch decomposition of \(G\) of minimum \(\mathcal{F}^*\)-branchwidth, \(D = \lcac{\decT}{N^3[\chi(t)]}\), 
for an edge \(e \in E(D)\), \(P_e\) denote the path in \(\decT\) that corresponds to \(e\), and
\(\flat\) be such that for every \(e \in E(D)\),  
(1) the subtree of \(\decT - P_e\) containing \(v \in N^{3\tw(G)}(N^3[\chi(t)])\) is attached to an internal node of \(P_e\) if and only if \(v \in \bigcup_{d \in \flat_e} d\); and (2) for \(v,w \in \bigcup_{d \in \flat_e} d\), \(v\) is contained in a subtree of \(\decT - P_e\) attached at a node \emph{before} (in an arbitrary but fixed traversal of \(P_e\)) the node at which the subtree of \(\decT - P_e\) containing \(w\) is attached if and only if the set in \(\flat_e\) that contains \(v\) is before the set in \(\flat_e\) that contains \(w\) in the sequence \(\flat_e\).

\begin{restatable*}{definition}{BlockDefinition}
	\label{def:block}
	For \(e \in E(D)\) which corresponds to the path \(P_e\) in \(\decT\), each maximal subpath of \(P_e\) which does not have internal nodes at which subtrees of \(\decT - P_e\) containing vertices in \(N^{3\tw(G)}(N^3[\chi(t)])\) are attached is called a \emph{block} of \(e\).
\end{restatable*}

\begin{restatable*}{observation}{NoBlocksObservation}
	\label{obs:No_blocks}
	There are \(|\flat_e| + 1\) blocks of an edge \(e \in E(D)\).
\end{restatable*}
In this way we can immediately relate each block to a pair \((e,i)\) where \(e \in E(D)\) and \(i \in [|\flat_e| + 1]\), where we enumerate the blocks of \(e\) according to the fixed traversal of \(P_e\).
For the remainder of this subsection we fix an edge \(e \in E(D)\) and its corresponding path \(P_e\)~in~\(\decT\).

Next we show that blocks have the desirable property that we can consider induced matchings with uniformly restricted interaction with \(N^3[\chi(t)]\) when computing the \imatchwidth values at edges of a fixed block.
Even stronger this restricted interaction does not depend on the internal structure of a block which is important for the application of typical sequences.

\begin{restatable*}{lemma}{blocksetLemma}
	\label{lem:blockset}
	Let \(P'\) be a block of \(P_e\).
	Then there is some \(\lambda \subseteq N^3[\chi(t)]\) such that the following holds.
	Consider an arbitrary branch decomposition \(\tilde{\decT}\) of \(G\) that arises from \(\decT\) by deleting all subtrees of \(\decT - P_e\) that are attached at internal vertices of \(P'\) and reattaching binary trees with the same cumulative set of leaves at internal nodes of \(P'\) or subdivisions of edges of \(P'\).
	Denote by \(\tilde{P}\) the path in \(\tilde{\decT}\) which corresponds to the path of all (possibly subdivided) edges of \(P'\).
	Then for every edge \(\tilde{e} \in E(\tilde{P})\) there is an induced matching \(\tilde{M}\) in the bipartite subgraph of \(G\) induced by \(\tilde{e}\) with \(2\imatchwidth(\decT,\tilde{e})\) vertices such that the edges in \(E(\tilde{M}[N^3[\chi(t)]])\) are incident to exactly \(\lambda\).
\end{restatable*}

Lemma~\ref{lem:blockset} justifies the following definition.

\begin{restatable*}{definition}{LambdaBlockDefinition}
	\label{def:lambda_block}
	We say that a block \(P'\) of \(P_e\) is a \emph{\(\lambda\)-block} where \(\lambda \subseteq N^3[\chi(t)]\) satisfies the properties described in Lemma~\ref{lem:blockset}.
\end{restatable*}

Note that the entry \(\Lambda\) in our dynamic programming records acts as a guess for which sets \(\lambda_{e,i}\) each block associated to \((e,i)\) is a \(\lambda_{e,i}\)-block for.

Now consider a \(\lambda\)-block \(P_\lambda\) of \(P\).
We show that we can apply typical sequences in the usual way within \(P_\lambda\), i.e.\ assume that whenever some vertices in \(V(G) \setminus V(G_t)\) are attached in \(\decT\) at \(P_\lambda\), that they are attached at certain points of \(P_\lambda\) which can be distinguished even after using typical sequences for the \(\imatchwidth\) to compress \(P_\lambda\).

For a branch decomposition \(\decT'\) of \(G_t\) and an edge \(f \in E(\decT')\) 
we use \(\mimw_\lambda(\decT',f)\) to denote the maximum size of an induced matching \(M\) in the cut of \(\decT'\) at \(f\) in which the set of vertices in \(N^3[\chi(t)]\) adjacent to edges in \(M[N^3[\chi(t)]]\) is equal to \(\lambda\).
In the pathological case that there is no such matching in which the set of vertices in \(N^3[\chi(t)]\) adjacent to edges in \(M[N^3[\chi(t)]]\) is equal to \(\lambda\), we set \(\mimw_\lambda(\decT',f) = \bot\).
Finally for two edges \(p, q \in E(P_\lambda)\), we write \(P_\lambda(p,q)\) to denote the subpath of \(P_\lambda\) between \(p\) and \(q\) and excluding \(p\) and \(q\).
We do not make any distinction between \(P_\lambda(p,q)\) and \(P_\lambda(q,p)\).

\begin{restatable*}{lemma}{typicalLemma}
	\label{lem:typical}
	Let \(\decT' = \lcac{\decT}{N^3[V(G_t)]}\) be the branch decomposition of \(G[N^3[V(G_t)]]\) given by \(\decT\),
	and let \(P'\) be the subpath of \(\decT'\) that corresponds to \(P_\lambda\).
	Assume that there are \(p, q\in E(P')\) such that all of the following hold.
	\begin{enumerate}
		\item \label{typical:cond1}
		\(
		\mimw_\lambda(\decT', p) = \min\{\mimw_\lambda(\decT', r) \mid r \in E(P'(p,q))\}.
		\)
		
		\item \label{typical:cond2}
		\(
		\mimw_\lambda(\decT', q) = \max\{\mimw_\lambda(\decT', r) \mid r \in E(P'(p,q))\}.
		\)
		
		\item \label{typical:cond3}
		For every subtree \(\decT_v\) of \(\decT\) which is attached at an internal node \(v\) of \(P_\lambda\), it holds that \(V(\decT'_v) \cap N^3[\chi(t)] = \emptyset\).
	\end{enumerate}
	Obtain \(\decT^{**}\) from \(\decT\) by removing for each internal node \(v\) of \(P_\lambda(p,q)\) \(V(G) \setminus V(G_t)\) from \(\decT\)
	and attaching \(F_v = \lcac{\decT_v}{V(G) \setminus V(G_t)}\) in the order of the considered vertices \(v\) along \(P_\lambda\) at the iterative subdivision of \(p\).
	See Figure~\ref{fig:typicalOverview} for an illustration.

	Then \(\mimw(\decT^{**}) \leq \mimw(\decT)\), \(\iantimatchwidth(\decT^{**}) \leq \iantimatchwidth(\decT)\), and \(\ichainwidth(\decT^{**}) \leq \ichainwidth(\decT)\).
\end{restatable*}

\begin{figure}
	\begin{minipage}[c]{.5\textwidth}
		\includegraphics[page=3]{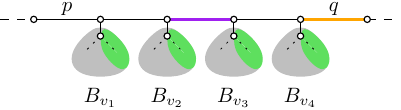}
		\vspace{0.5em}
		
		\includegraphics[page=4]{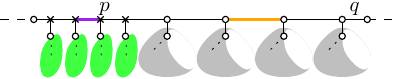}
	\end{minipage}
	\hfill
	\begin{minipage}[c]{.48\textwidth}
		\caption{\label{fig:typicalOverview}
			Illustration of modification of \(\decT\) (top) to \(\decT^{**}\) (bottom) as described in Lemma~\ref{lem:typical}.
			The cross-vertices in the bottom figure are the subdivision vertices used to attach the \(\lcac{\decT_v}{V(G) \setminus V(G_t)}\) to \(p\).
			The cutfunction values at the orange and purple edge in \(\decT^{**}\) are upper-bounded by the cutfunction values at the orange and purple edge in \(\decT^*\) respectively.}
	\end{minipage}
\end{figure}

In the later application of Lemma~\ref{lem:typical}, \(P'(p,q)\) takes the role of a path that is contracted in for obtaining any of the typical sequences in a record, and iteratively applying this lemma ensures the ``safeness'' (w.r.t.\ \(\cF^*\)-branchwidth and the properties of blocks) of using them.

With the machinery now in place, we are able to describe the dynamic programming procedure which we use to prove Theorem~\ref{thm:tw_deg}. (see \emph{Sections~\ref{sec:dp} and~\ref{sec:correct}}).

\subsection{Treedepth}
\label{subsec:td}
In Section~\ref{sec:td} we give a fixed parameter tractable algorithm for $\mathcal{F}^*$-\textsc{Branchwidth} parameterized by the treedepth of the input graph,
when $\mathcal{F}^*$ is the union of a non-empty set of some \primal \phsymb families. We begin by outlining the high-level idea.

If a graph has small treedepth but sufficiently many vertices, 
then one can always find a small vertex set $R$ such that $G-R$ has many components each of them has bounded size.
The main step is to argue that in this case, we can safely remove (``prune'') one of the components without changing its $\mathcal{F}^*$-branchwidth.
Proposition~\ref{prop:manycomponent} captures the idea of this procedure, and we inductively use this result from bottom to top in the treedepth decomposition. 
At the end, we reduce the given graph to a graph whose number of vertices is bounded by a function of the treedepth of the given graph, i.e., a (non-polynomial) kernel. The problem can then be solved on such a kernel using an arbitrary brute-force algorithm.

\begin{restatable*}{proposition}{manycomponentProposition}
	\label{prop:manycomponent}
	There is a function $g\colon \mathbb{N}\times \mathbb{N} \to \mathbb{N}$ satisfying the following. 
	Let $t,p,$ and $m$ be positive integers, and let $\mathcal{F^*}$ be the union of a non-empty set of \texttt{primal ph} classes. 
	Let $G$ be a graph, let $R$ be a vertex set of size at most $t$, let $H$ be a graph on $p$ vertices, and let $\{G_i\mid i\in [m]\}$ be a set of components of $G-R$ such that 
	\begin{itemize}
		\item for each $i\in [m]$, there is a graph isomorphism $\phi_i$ from $H$ to $G_i$, 
		\item for all $i, j\in [m]$ and $h\in V(H)$ and $v\in R$, $\phi_i(h)$ is adjacent to $v$ if and only if $\phi_j(h)$ is adjacent to $v$.
	\end{itemize}
	If $m\ge g(t,p)$ and $\Fsbw(G)\le t$, then $\Fsbw(G)=\Fsbw(G-V(G_m))$.
\end{restatable*}

We explain the idea to prove Proposition~\ref{prop:manycomponent}. Suppose we have $G$, $R$, $H$, and the family $\{G_i\mid i\in [m]\}$ as in the statement.
Let $(\decT, \decf)$ be a branch decomposition of $G-V(G_m)$ of optimal width, and 
we want to find a branch decomposition of $G$ having same width.
Let $V(H)=\{h_1, h_2, \ldots, h_p\}$.
First we take nodes corresponding to $R$ and its least common ancestors. Let $S$ be the union of all least common ancestors in the tree.
We take a large subset $I_1$ of $[m]$ satisfying that 
for all $a\in [p]$ and $i_1, i_2\in I_1$, $\phi_{i_1}(h_a)$ and $\phi_{i_2}(h_a)$ are contained in the same component of $\decT-S$.

As a second step, we show that 
there are a large subset $I_2\subseteq I_1$, a set $F$ of at most $\ell-1$ edges in $\decT$ for some $1\le \ell\le p$, 
a partition $(A_1, \ldots, A_{\ell})$ of $[p]$, and a bijection $\mu$ from $[\ell]$ to the set of connected components of $\decT-F$ such that 
for all $j\in [\ell]$ and $i_1, i_2\in I_2$, the minimal subtree of $\decT$ containing all nodes in $\{\decf^{-1}(\phi_{i_1}(h_b)) \mid  b\in A_j\}$
and the minimal subtree of $\decT$ containing all nodes in $\{\decf^{-1}(\phi_{i_2}(h_b)) \mid  b\in A_j\}$
are vertex-disjoint.
In fact, by the choice of $S$ and $I_2$, there should be a component of $\decT-S-F$ containing all nodes corresponding to $\{h_b\mid b\in A_j\}$.
Then in each component of $\decT-S-F$ corresponding to $A_j$, 
we can find an edge separating  the family of minimal subtrees corresponding to $A_j$ in a balanced way, 
and expand that edge to some subtree corresponding to vertices of $\{\phi_m(h_b) \mid   b\in A_j\}$.
When we expand the edge, we take the restricted tree of $\decT$ with respect to some $\{\decf^{-1}(\phi_\alpha(h_b)) \mid   b\in A_j\}$ which is closest to the edge, and make a copy and identify with the edge in a natural way.
Then we can prove that every new cut has small width as well.

\begin{restatable*}{theorem}{treedepthTheorem}
	\label{thm:treedepth}
	Let $\mathcal{F}^*$ be the union of a non-empty set of \emph{\primal} \texttt{ph} classes. Then $\mathcal{F}^*$-\textsc{Branchwidth} is fixed-parameter tractable when parameterized by treedepth.
\end{restatable*}

\begin{restatable*}{corollary}{treedepthCorollary}
	\label{cor:treedepth}
	\textsc{Mim-Width} is fixed-parameter tractable parameterized by the treedepth of the input graph.
\end{restatable*}

\subsection{Feedback Edge Set}\label{subsec:fes}
In Section~\ref{sec:fes} we give a linear kernel for \Fsbwprob parameterized by the size 
$k$ of a feedback edge set of the input graph,
when $\Fstar$ is the union of some \primal \phsymb families.

If a graph $G$ has many vertices but a small feedback edge set,
then either it has many bridges (in form of `dangling trees') and isolated vertices,
or it has long paths of degree two vertices in $G$,
which we refer to as `unimportant paths'.
This simple observation paves the road to the kernel:
For the former, we can observe that bridges and isolated vertices 
are not important when it comes to computing \Fstar-branchwidth.
For the latter, we show that we can always shrink any unimportant path to a 
constant-length subpath without changing the \Fstar-branchwidth of $G$.
It is not surprising that this gives a safe reduction rule when $\Fstar$\xspace
does not contain induced matchings.
For induced matchings however, proving safeness requires quite an amount of detail.
For this we show Lemma~\ref{lem:fes:unimportant:placement} which states that
if $G$ has a long enough unimportant path $P$, 
then we can modify each branch decomposition of $G$ such that the vertices
of a long enough subpath of $P$ appear `consecutively' in it,
without increasing the \Fstar-branchwidth.
In addition to that, we only have to overcome one minor hurdle 
that concerns induced chains of value $2$.
From there on it is not difficult to argue that 
contracting an edge on $P$ does not change the \Fstar-branchwidth either.

Let us begin. 
We first define unimportant paths,
and what it means for the vertices of a path to appear `consecutively'
in a branch decomposition via the notion of preservation.

\begin{restatable*}{definition}{unimportantPathDefinition}
	\label{def:unimportantPath}
	Let $G$ be a graph. 
	A path $P \subseteq G$ is called \emph{unimportant}
	if all vertices on $P$ have degree~two~in~$G$.
\end{restatable*}

\begin{restatable*}{definition}{preservesDefinition}
	\label{def:preserves}
	Let $G$ be a graph, let $(\decT, \decf)$ a branch decomposition of $G$,
	and let $P = a_1 \ldots a_\ell \subseteq G$ be a path.
	We say that $(\decT, \decf)$ \emph{preserves} $P$
	if there is a caterpillar $\decT''$ in $\decT$ such that the linear order 
	induced by $\decT''$ (up to reversal) is $a_1 \ldots a_\ell$.
\end{restatable*}

We now turn to the main technical lemma needed in the correctness proof of the 
reduction rule alluded to above. 
Note that it only concerns induced matchings.

\begin{restatable*}{lemma}{FesUnimportantPlacementLemma}
	\label{lem:fes:unimportant:placement}
	Let $G$ be a graph with branch decomposition $(\decT, \decf)$.
Suppose that $G$ contains an unimportant path $P$ of length $7$.
Then there is a branch decomposition $(\decT', \decf')$ of $G$ such that
\begin{enumerate}
	\item $\mimw(\decT', \decf') \le \mimw(\decT, \decf)$, 
	\item $(\decT', \decf')$ preserves a subpath $P^\star$ of $P$ on at least five vertices, and
	\item the set of cuts induced by $(\decT, \decf)$ on $G - V(P^\star)$ is equal to the set of cuts induced by $(\decT', \decf')$ on $G - V(P^\star)$.
\end{enumerate}
\end{restatable*}

The construction in the previous lemma does not 
increase the $\Fmatch$-branchwidth of a branch decomposition.
It might however introduce induced chain graphs of value $2$ in some cut, which is a technical complication that can, luckily, be dealt with by considering the case of $\Fstar$-branchwidth $1$ separately. Next, we can develop separate arguments that are designed to handle the case of $\Fchain$-branchwidth when $\Fstar$-branchwidth is at least $2$, and show that these two approaches can be combined together to deal with $\Fstar$-branchwidth for any union $\Fstar$ of \primal \phsymb\ classes. Recall that a graph is \emph{bridgeless} if there is no edge $e \in E(G)$
such that $G - e$ has more connected components than $G$; we obtain:

\begin{restatable*}{lemma}{FesMainLemma}
	\label{lem:fes:main}
	Let $\Fstar$ be a union of \emph{\primal} \emph{\phsymb} classes.
	Let $G$ be a bridgeless graph 
	without isolated vertices 
	with feedback edge set number $k$ and an unimportant path $P$ of length at least $8$. Let $e \in E(P)$. Either $\card{V(G)} \le 8k - 3$, or $\Fsbw(G) = \Fsbw(G/e)$, or both.
\end{restatable*}

As the previous lemma requires that our input graph does not have any bridges,
we need one more simple lemma to show that we can safely remove bridges from the
input graph without changing its \Fstar-branchwidth.

\begin{restatable*}{reduction}{FesBridgeIsolatedReduction}
	\label{reduction:fes:bridge:isolated}
	Let $(G, k)$ be an instance of \Fsbwprob, 
	where $G$ has a feedback edge set of size $k$.
	Then, reduce $(G, k)$ to $(G', k)$ where $G'$ is obtained from $G$ by
	\begin{itemize}
		\item removing all bridges of $G$, and then
		\item removing all isolated vertices of $G$.
	\end{itemize}
\end{restatable*}

\begin{restatable*}{reduction}{FesUnimportantReduction}
	\label{reduction:fes:unimportant}
	Let $(G, k)$ be an instance of \Fsbwprob,
	where $G$ has a feedback edge set of size $k$.
	Let $P \subseteq G$ be any unimportant path in $G$ of length at least $8$ 
	and let $e \in E(P)$.
	Then, reduce $(G, k)$ to $(G/e, k)$.
\end{restatable*}

We now have all the tools necessary to describe the kernelization algorithm:
We first apply Reduction Rule~\ref{reduction:fes:bridge:isolated} to clean up the graph.
Next, while the input graph has more than $18k - 8$ vertices
we find a long enough unimportant path and apply 
Reduction Rule~\ref{reduction:fes:unimportant} to it.
It remains to employ the results of this section to finalize the correctness proof, yielding:

\begin{restatable*}{theorem}{FesMainTheorem}
	\label{thm:fes_main}
		Let $\Fstar$ be the union of some \emph{\primal} \emph{\phsymb} families.
	\Fsbwprob admits a linear kernel when parameterized by the feedback edge set number of the input graph.	
\end{restatable*}

\begin{restatable*}{corollary}{FesMainCorollary}
	\label{cor:fes_main}
	\textsc{Mim-Width} parameterized by the feedback edge set number of the input graph admits a linear kernel.
\end{restatable*}

\section{Size-Identifiable Classes and Comparison to Existing Measures}
\label{sec:comparison}

Here, we focus on the six families introduced at the end of Section~\ref{sec:fbw}. We first show that while there is an infinite number of possible \phsymb\ graph classes, to asymptotically characterize $\Fbw$ for any \phsymb\ family of graphs \(\cF\) we only need to consider these families. Afterwards, we turn our attention on several previously studied graph measures and show how they compare to $\Fbw$. Finally we show that to compute an approximate \(\cF^*\)-branch decomposition, it suffices to focus on only three of these classes, namely \(\imatchFamily\), \(\ichainFamily\), and \(\iantimatchFamily\).

\subsection{Size-Identifiable Classes}
	We say that \(\cF\) is \si\ (\sisymb) if for each \(n\in \Nat\), there is a single \(2n\)-vertex graph in \(\cF\) (up to isomorphism). We will show \(\cF_\emptyset\), \(\imatchFamily\), \(\ichainFamily\), \(\ichainstrictFamily\), \( \iantimatchFamily\), and \(\icompleteFamily\) are the only \sisymb\ \phsymb\ families of graph. To show that these are the only \sisymb\ \phsymb\ families as well as to show importance of these families, we will make use of the following lemma that basically shows that any bipartite graph with ordered partitions contains a large subset of partners that induce a graph isomorphic to a graph in one of the six above families.
\ramseyPartners
\iflong
\begin{proof}
	To prove the lemma, we will construct an auxiliary edge-colored complete graph \(G\) on \(q\) vertices such that there is bijection between the vertices of \(G\) and partners \(a_i, b_i\) in \(H^q\). Moreover, a monochromatic clique in \(G\) will correspond to one of $H^n_\emptyset$, $H^n_{\imatch}$, $H^n_{\ichain}$, $H^n_{\iantimatch}$, or $H^n_{\icomplete}$.
	Let \(G\) be the edge-colored graph such that \(V(G) = (v_1, v_2, \ldots, v_n)\) and the edge \(v_iv_j\), \(1\le i<j\le q\), has color 
	\begin{itemize}
		\item 1 if both edges \(a_ib_j, a_jb_i\) are in \(E\), 	
		\item 2 if the edge \(a_ib_j\) is in \(E\) and the edge \(a_jb_i\) is not in \(E\),
		\item 3 if the edge \(a_jb_i\) is in \(E\) and the edge \(a_ib_j\) is not in \(E\),
		\item 4 if both edges \(a_ib_j, a_jb_i\) are not in \(E\),
	\end{itemize}

	Since $q \ge R(2n,2n,2n,2n)$, by Ramsey's Theorem $G$ contains a monochromatic clique on \(2n\) vertices. Note that either at least \(n\) vertices of the clique correspond to partners that are adjacent to each other, or at least \(n\) vertices of the clique correspond to partners that are not adjacent to each other. Let \(C\) be a monochromatic clique on \(n\) vertices such that either $a_ib_i\in E$ for each $v_i\in C$ (i.e., all partners in $C$ are matched), or $a_ib_i\not \in E$ for each $v_i\in C$ (i.e., all partners in $C$ are not matched).
	Let us now consider each of the following possibilities:
	\begin{description}
		\item[Color 1 and partners are adjacent.] It is easy to verify that the graph induced on \(\{a_i,b_i\mid i\in L\}\) is isomorphic with $H^n_{\icomplete}$.
		\item[Color 1 and partners are not adjacent.] Then the graph induced on \(\{a_i,b_i\mid i\in L\}\) contains an edge \(a_i,b_j\) if and only if $i\neq j$ and it is isomorphic with $H^n_{\iantimatch}$.	
		
		\item[Color 2 and partners are adjacent.] Then the graph induced on \(\{a_i,b_i\mid i\in L\}\) contains an edge \(a_i,b_j\) if and only if $i\le j$ and it is isomorphic with $H^n_{\ichain}$.
		\item[Color 2 and partners are not adjacent.] Then the graph induced on \(\{a_i,b_i\mid i\in L\}\) contains an edge \(a_i,b_j\) if and only if $i < j$ and it is isomorphic with $H^n_{\ichainstrict}$.
		
		\item[Color 3 and partners are adjacent.] Then the graph induced on \(\{a_i,b_i\mid i\in L\}\) contains an edge \(a_i,b_j\) if and only if $i\ge j$ and it is isomorphic with $H^n_{\ichain}$. Note that here we have to change the order of the pairs to obtain \(H^n_{\ichain}\).
		\item[Color 3 and partners are not adjacent.] Then the graph induced on \(\{a_i,b_i\mid i\in L\}\) contains an edge \(a_i,b_j\) if and only if $i > j$ and it is isomorphic with $H^n_{\ichainstrict}$.
		
		\item[Color 4 and partners are adjacent.] Then the graph induced on \(\{a_i,b_i\mid i\in L\}\) contains an edge \(a_i,b_j\) if and only if $i=j$ and it is isomorphic with $H^n_{\imatch}$.
		\item[Color 4 and partners are not adjacent.] Then it is easy to verify that the graph induced on \(\{a_i,b_i\mid i\in L\}\) is isomorphic with $H^n_{\emptyset}$. \qedhere
	\end{description} 
\end{proof}
\fi
	
\onlysix
\iflong
	\begin{proof}
		It is rather straightforward to verify that each of the families \iemptyFamily, \imatchFamily, \ichainFamily, \ichainstrictFamily, \iantimatchFamily, \icompleteFamily\ is \phsymb\ and \sisymb. 
		Let \(\cF\) be an \sisymb\ and \phsymb\ graph family and let $H^n$ denote the unique \(2n\)-vertex graph in \(\cF\). First note that if \(H^n\) is in one of the six families described above (i.e., for example \(H^n = H^n_{\imatch}\)) then for all $m<n$ the graph \(H^m\) is in the same family of the graphs. Indeed for each of the six families the bipartite subgraph of \(H^n\) induced on the first \(m\) partners is in the same family. Since \(\cF\) is \phsymb, this subgraph is in \(\cF\) and since \(\cF\) is also \sisymb, this subgraph is isomorphic to \(H^m\). To finish the proof, it suffices to show that \(H^n\) is always one of the the following six graphs: $H^n_\emptyset$, $H^n_{\imatch}$, $H^n_{\ichain}$, $H^n_{\ichainstrict}$, $H^n_{\iantimatch}$, or $H^n_{\icomplete}$. Indeed, let $q = R(2n,2n,2n,2n)$ and let us consider the graph \(H^q\). By Lemma~\ref{lem:ramseyPartners} there exists \(L\subseteq [q]\) such that \(|L| = n\) and the subgraph of \(H^q\) induced on \(\{a_i,b_i\mid i\in L\}\) is isomorphic to one of $H^n_\emptyset$, $H^n_{\imatch}$, $H^n_{\ichain}$, $H^n_{\ichainstrict}$, $H^n_{\iantimatch}$, or $H^n_{\icomplete}$. Since \(\cF\) is \phsymb, \(\cF\) contains a graph isomorphic to one of $H^n_\emptyset$, $H^n_{\imatch}$, $H^n_{\ichain}$, $H^n_{\ichainstrict}$, $H^n_{\iantimatch}$, or $H^n_{\icomplete}$ and since \(\cF\) is \sisymb, this graph is indeed isomorphic to \(H^n\).
	\end{proof}
	\fi
	
We now apply Ramsey's Theorem to show that these six \texttt{\emph{si}} classes capture the whole hierarchy of width parameters induced by $\mathcal{F}$-branchwidth.

\sigood
\iflong
\begin{proof}
	Since \(\cF'\subseteq \cF\), it follows from Observation~\ref{obs:closedSubFamilies} that for every graph \(G\) it holds that \(\cF'\operatorname{-bw}(G)\le \operatorname{\cF-bw}(G)\). It hence suffices to show that $\operatorname{\cF-bw}(G)$ is upper-bounded by a function of $\cF'\operatorname{-bw}(G)$.
	
	Let \(q\) be the size of the largest graph \(H\) such that \(H\in \cF\setminus \cF'\) and \(H\) is isomorphic to one of $H^q_\emptyset$, $H^q_{\imatch}$, $H^q_{\ichain}$, $H^q_{\ichainstrict}$, $H^q_{\iantimatch}$, or $H^q_{\icomplete}$; if no such graph $H$ exists, we set $q$ to $0$. Note that, in any case, \(\cF\) does not fully contain the \sisymb\ class containing \(H\) and \(q\) is a constant. 
	Let us consider a branch decomposition \(\decT\) of \(G\) such that \(\cF'\mbox{-}\bw(\decT)=k\) and let \(n = 1+\max(q,k)\). We will show that  \(\cF\mbox{-}\bw(\decT)< R(2n,2n, 2n, 2n)\), which will complete the proof. 
	
	For the sake of contradiction, let us assume that this is not the case and that there is an edge \(e\in\decT\) such that \(\cF\mbox{-}\bw(\decT,e)\ge R(2n,2n, 2n, 2n)\). Let \(X,Y\) be the partition of \(V(G)\) induced by \(e\). Then \(G[X,Y]\) has an induced subgraph on \(2m\) vertices, for some \(m \ge R(2n,2n, 2n, 2n)\), that is isomorphic to some graph \(H^m\) in \(\cF\). However, \(m\ge R(2n,2n, 2n, 2n)\) and by Lemma~\ref{lem:ramseyPartners}, there is a subset \(L\subseteq [m] \) such that \(|L|=n\) and the graph \(H^n\) induced on \(\{a_i,b_i\mid i\in L\}\) is isomorphic to one of $H^n_\emptyset$, $H^n_{\imatch}$, $H^n_{\ichain}$, $H^n_{\ichainstrict}$, $H^n_{\iantimatch}$, or $H^n_{\icomplete}$. Since \(n>q\), \(H^n\) is isomorphic to a graph in an \sisymb\ family fully contained in \(\cF\). But then \(\cF'\mbox{-}\bw(\decT,e)\ge n > k\), contradicting \(\cF'\mbox{-}\bw(\decT)=k\).
\end{proof}
\fi

\subsection{Comparison to Existing Measures} 
 Next, we discuss how specific \texttt{\emph{si ph}} classes capture previously studied parameters.
 First, mim-width is precisely \imatchFamily-branchwidth. 
 \iflong Indeed, mim-width is defined as branchwidth where the cut-function $f(X)$ is defined as the size of the largest induced matching in \(G[X,V(G)\setminus X]\), which is precisely the largest $n$ such that $H^n_{\imatch}$ is an induced subgraph of \(G[X,V(G)\setminus X]\). 
 \fi
 
\mimWidth

	The next decomposition parameter that we consider is treewidth. 
	\ifshort To show the equivalence, we actually show that \emph{maximum-matching width}, which is always linearly upper- and lower-bounded by treewidth~\cite{vatshelle2012new,JeongST18}, is also \(\Fbw\) for the \phsymb\ family \(\cF\) of all \texttt{\emph{ph}} graphs that contain a matching on the partners.
	\fi
	
	\twbound
	\iflong
	\begin{proof}
		It is known that \(\operatorname{mmw}(G)\le \tw(G) + 1\le 3\cdot \operatorname{mmw}(G)\), where \(\operatorname{mmw}(G)\) is the maximum matching-width of \(G\)~\cite{vatshelle2012new,JeongST18}. The maximum matching-width is the branchwidth where the cut-function \(\decf(X)\) is the size of maximum matching in \(G[X, V(G)\setminus X]\). Let \(\cF\) be the family of graphs that for each \(n\) contains all bipartite graphs \(H = ((a_1,\ldots, a_n),(b_1,\ldots, b_n), E)\) such that \(a_ib_i\in E\) for all \(i\in [n]\). It is straightforward to verify that for each graph \(H\) in \(\cF\) on \(n\) vertices and each subset \(L\) of \([n]\), the graph induced on \(\{a_i,b_i | i\in L\}\) is in \(\cF\). Hence \(\cF\) is a \phsymb\ family. Moreover, the size of a maximum matching in \(G[X, V(G)\setminus X]\) is precisely the largest \(n\) such that a \(2n\)-vertex graph in \(\cF\) occurs as an induced subgraph of \(G[X,V(G)\setminus X]\). Indeed, if \(G[X, V(G)\setminus X]\) contains a matching \(M = \{a_1b_1, a_2b_2, \ldots, a_nb_n\}\), where for each \(i\in [n]\) the vertex \(a_i\) is in \(X\) and the vertex \(b_i\) is in \(V(G)\setminus X\), then \(G[\{a_1,\ldots,a_n\}, \{b_1,\ldots,b_n\}]\) is in \(\cF\). On the other hand, the size of a maximum matching of a \(2n\)-vertex graph \(H\) in \(\cF\) with bipartition \((A=(a_1,\ldots, a_n),B= (b_1,\ldots, b_n))\) is \(n\).
	\end{proof}
	
		Note that the \phsymb\ family \(\cF\) in the above lemma contains \(\imatchFamily\), \(\ichainFamily\), \(\icompleteFamily\) and a family isomorphic to \((\iantimatchFamily\setminus \{H^1_{\iantimatch}\}) \cup \{H^1_{\imatch}\} \) as subfamilies, because anti-matchings of size at least 2 contain a perfect matching. Hence in combination with Observations~\ref{obs:closedSubFamilies}~and~\ref{obs:isomorphicFamilies} we get the following corollary whose importance will be apparent later and which gives a better bound for one inequality than a direct application of Lemma~\ref{lem:sigood}. 
	\fi

\ifshort
By Lemma~\ref{lem:sigood}, treewidth is then asymptotically equivalent to $\mathcal{F}^*$-branchwidth for $\mathcal{F}^*$ being the union of some \texttt{\emph{si ph}} classes, notably \imatchFamily, \ichainFamily, \iantimatchFamily, and \(\icompleteFamily\).
By applying the easy observation that \(\cF'\subseteq \cF\) implies \(\cF'\mbox{-}\bw(G)\le \FbwSh(G)\) for any \phsymb\ families \(\cF\) and \(\cF'\) and any graph \(G\),
we can also establish an explicit upper bound when considering \imatchFamily, \ichainFamily and \iantimatchFamily.
The reason we are interested in specifically this set of families will become clear when we define \emph{\primal} families.
\fi

	\FisTWBounded
	\iflong
	\begin{proof}
	Let \(\cF\) be the family defined in the proof of Lemma~\ref{lem:twbound}. If $\cF^*$ does not contain \iantimatchFamily, then 
	 \(\cF^*\) is isomorphic to a subfamily of $\cF$. 
	 By Observations~\ref{obs:closedSubFamilies}~and~\ref{obs:isomorphicFamilies}, we have  
	 \(\cF^*\mbox{-}\bw(G)\le \cF\mbox{-}\bw(G)\le\tw(G)+1.\)

	 Now we assume that $\cF^*$ contains \iantimatchFamily. Let $\cF^{**}$ be the family of all graphs in $\cF^*$ except the graph $H^1_{\iantimatch}$. Then $\cF$ contains a family isomorphic to $\cF^{**}$. We observe that 
	 \(\cF^*\mbox{-}\bw(G)\le \max (1, \cF^{**}\mbox{-}\bw(G)).\)
	 Therefore, we have  \(\cF^*\mbox{-}\bw(G) \le\tw(G)+1.\)
	\end{proof}
	\fi
	
	Throughout Section~\ref{sec:tw} and in Proposition~\ref{prop:manycomponent} in Section~\ref{sec:td}, the fact that $\Fsbw(G)$ is bounded in the respective parameters is essential, and Corollary~\ref{cor:FisTWBounded} gives this bound.
	
	The next parameter that we will consider is clique-width. Similarly as for treewidth, it will be easier to show equivalence with another parameter, \emph{module-width}, which is known to be asymptotically equivalent to cliquewidth~\cite{Rao08,BuiXuanSTV13}. 
	\ifshort
	An important tool to establish this equivalence is Corollary 2.4 in \cite{DingOOV96}, which intuitively states that any sufficiently large twin-free bipartition of vertices in a graph must contain a large chain, anti-matching or matching.
	\fi
	
	\iflong
		Let \(G\) be a graph and \(X\subseteq V(G) \) be a subset of its vertices. The twin class partition of \(X\)
		is a partition of \(X\) such that for all \(x,y\in X\), $x$ and $y$ are in the same partition class if and only if \(N(x)\setminus X = N(y)\setminus X \). We let \(\operatorname{ntc}(X)\) be the number of classes in the twin class partition of \(X\). 
	
	Note that the function \(\operatorname{ntc}: 2^{V(G)}\rightarrow \Nat\) is not symmetric. 
	However, to deal with such functions it is natural to consider a \emph{rooted} adaptation of branchwidth. Let a rooted branch decomposition \(\decT\) to be a branch decomposition that is \emph{rooted} at a node \(r\in \decT\) of degree two.
	For an edge \(e\) in \(\decT\), let the width of \(e\) be \(f(X)\) such that the connected components of $\decT\setminus e$ induce a bipartition $(X,Y)$ of the set of leaves of $\decT$ and where \(X\) is in the component of $\decT\setminus e$ not containing \(r\). The width of a rooted decomposition and rooted branchwidth is then defined in the same way as for branchwidth. The graph parameter module-width is then the rooted branchwidth with the cut-function \(\operatorname{ntc}\). For a graph \(G\) we denote by \(\operatorname{modw}(G)\) the module-width of \(G\) and for a rooted branch decomposition \(\decT\) and edge \(e\) of \(\decT\) we denote by \(\operatorname{modw}(\decT)\) and \(\operatorname{modw}(\decT,e)\) the module-width of the decomposition \(\decT\) and of the edge \(e\), respectively.
	
	\begin{fact}[\cite{Rao08}]
		Let \(G\) be a graph, then \(\operatorname{modw}(G)\le \operatorname{cw}(G)\le 2 \operatorname{modw}(G)\). 
	\end{fact}
		
	To show that we can capture module-width as \(\Fbw\), we will make use of the following lemma. 
\begin{fact}[Corollary 2.4 in \cite{DingOOV96}]\label{prop:ramseyCliquewidth}
	For every positive integer \(n\), there is an integer \(d\) with the following property: Suppose \(G\)  is  a bipartite graph with a bipartition \(U\) and \(V\) such that \(|U|= m\ge d\) and \(N(u_1)\neq N(u_2) \) for \(u_1\neq u_2\in U\). Then \(U\) and \(V\) have \(n\)-element subsets \(U'\) and \(V'\), respectively, such that the subgraph \(G'\) induced by \(U'\cup V'\)
	is isomorphic to \(H^n_{\imatch}\), \(H^n_{\iantimatch}\), or \(H^n_{\ichain}\).
\end{fact}

\fi

	\compModuleWidth

\iflong
\begin{proof}
	First we will show that \(\FbwSh(G)\le \operatorname{modw}(G) \). 
	 Let \(\decT\) be rooted branch decomposition of \(G\) and let \(e\) be and edge in \(\decT\) that induces bipartition $(X,Y)$, where \(X\) is in the component of \(\decT\setminus e\) not containing the root. Observe that all graphs in \(\cF\) are twin-free, i.e., for every graph \(H\in\cF\) and every pair of vertices \(x,y\in V(H)\) such that \(x\neq y\) we have \(N_H(x)\neq N_H(y)\). Hence, if \(G[X,V(G)\setminus X]\) contains a \(2n\)-vertex induced subgraph isomorphic to a graph in \(\cF\), then each of the \(n\) vertices in this subgraph are in different twin classes of \(X\). It follows that \(ntc(X)\) is at least the largest \(n\) such that a \(2n\)-vertex graph in \(\cF\) is isomorphic to an induced subgraph of \(G[X,V(G)\setminus X]\). Hence, \(\cF\mbox{-}\bw(\decT) \le \operatorname{modw}(\decT)\). 
	 
	 On the other hand, let \(\decT\) be a branch decomposition of \(G\) and let \(\decT'\) be a rooted branch decomposition we get from \(\decT\) by subdividing an arbitrary edge of \(\decT\) and rooting \(\decT'\) in the new vertex \(r\). Trivially, \(\cF\mbox{-}\bw(\decT') = \cF\mbox{-}\bw(\decT)\). 
	 We show that \( \operatorname{modw}(\decT')\le f(\cF\mbox{-}\bw(\decT'))\) for some computable function \(f\). Let \(e\) be an edge in \(\decT\) that induces the bipartition \((X,Y)\), where \(X\) is in the component of \(\decT\setminus e\) not containing the root. Let \(n\) be the largest integer such that \(G[X,V(G)\setminus X]\) contains a \(2n\)-vertex induced subgraph isomorphic to a graph in \(\cF\) and let \(d\) be the integer given by Proposition~\ref{prop:ramseyCliquewidth} for \(n+1\). If \(\operatorname{ntc}(X) \ge d\), then \(X\) contains a subset \(U\) of at least \(d\) vertices, each in a differenct twin partition class, such that any two distinct vertices in \(U\) have different neighborhoods in \(V(G)\setminus X\). But then by Proposition~\ref{prop:ramseyCliquewidth} \(G[X,V(G)\setminus X]\) contains as an induced subgraph a graph isomorphic to a \((2n+2)\)-vertex graph in \(\cF\). Hence, \( \operatorname{modw}(\decT')\le f(\cF\mbox{-}\bw(\decT'))\), where \(f(n) = d\), such that \(d\) is the integer guaranteed by Proposition~\ref{prop:ramseyCliquewidth} for \(n+1\). 
\end{proof}
\fi

Finally, we compare \(\Fbw\) to the so-called \(H\)-join decompositions introduced by Bui-Xuan, Telle, and Vatshelle~\cite{BuiXuanTV10}.
\iflong
\begin{definition}
	Let \(H = (V_1, V_2, E)\) be a bipartite graph. Let \(G\) be a graph and let \(S\subseteq V(G)\). We say that \(G\) is an \(H\)-join across the ordered cut \((S,V(G)\setminus S)\) if there exists a partition \(\cP = (P_1, \ldots, P_r) \) of \(S\) and a partition \(\mathcal{Q} = (Q_1, \ldots, Q_s)\), and injective functions \(f_1: \cP\rightarrow V_1 \) and \(f_2: \cP\rightarrow V_2 \) such that for any \(x\in P_i\in \cP\) and \(y\in Q_j\in \mathcal{Q}\) we have \(x\) adjacent to \(y\) in \(G\) if and only if \(f_1(P_i)\) is adjacent to \(f(Q_j)\) in \(H\). We say that \(G\) is an \(H\)-join across the non-ordered cut \(\{S, V(G)\setminus S\}\) if \(G\) is an \(H\)-join across either \((S,V(G)\setminus S)\) or \((V(G)\setminus S,S)\).
\end{definition}
Let \(\decT\) be a branch decomposition of a graph \(G\). We say that \(\decT\) is \(H\)-join decomposition of \(G\) if for every edge \(e\) of \(\decT\) that induces the bipartition \((X,Y)\) of \(V(G)\) the graph \(G\) is \(H\)-join across the cut \(\{X,Y\}\). A graph that admits an \(H\)-join decomposition is called an \(H\)-join decomposable graph. \fi
We note that unlike other parameters, $H$-join decompositions do not have a specific ``width''; each graph either admits a decomposition or not, and the complexity measure is captured by the order of the graph $H$. Keeping this in mind, we can show that $\mathcal{F}$-branch captures \(H\)-joins as well.

	\compHJoin

\iflong
\begin{proof}
	Let \(\decT\) be an \(H\)-join decomposition of \(G\).
	We claim that \(\cF\mbox{-}\bw(G)\le|V(H)|\). For the sake of contradiction let us assume otherwise and let \(e\) be an edge in \(\decT\) such that \(\cF\mbox{-}\bw(\decT,e)\ge |V(H)|+1\) and let \(X,Y\) be the partition of \(V(G)\) induced by \(e\). Then the graph
	\(G[X,Y]\) contains a \((2|V(H)|+2)\)-vertex induced subgraph \(F\) isomorphic to a graph in \(\cF\). Because for every pair of vertices \(u,v\) in a graph in \(\cF\) it holds that \(N(u)\neq N(v)\), it follows that no two vertices of \(F\) can map to the same vertex of \(H\) by the functions \(f_1\) and \(f_2\) from the definition of \(H\)-join. But \(|V(F)|>|V(H)|\), which is impossible.
\end{proof}

We remark that if \(H\) is twin-free (i.e., there are no two vertices \(u,v\in V(H)\) such that \(N(u)=N(v)\)) and there is an edge of an \(H\)-join decomposition of a graph \(G\) that induces partition \((X,Y)\) of \(V(G)\) and that the functions \(f_1\) and \(f_2\) in the definition of \(H\)-join are surjective, then we can find a subgraph of \(G[X,Y]\) isomorphic to \(H\) by picking a single vertex from each class in the partitions of \(X\) and \(Y\). Proposition~\ref{prop:ramseyCliquewidth} then implies that \(\cF\mbox{-}\bw(G)\ge f(|V(H)|)\) for some computable function \(f\). However, in general the class of \(H\)-decomposable graphs contains also very simple graph. So unless we say for example that \(H\) is a graph of minimum size such that \(G\) is \(H\)-join decomposable, we cannot give lower bound \(\cF\mbox{-}\bw(G)\) as a function of \(|V(H)|\) only upper bound. 

\fi

\subsection{Primal Families and Computing $\mathcal{F}$-Branchwidth}	
	Let us call the \phsymb\ families \(\imatchFamily\), \(\ichainFamily\), and \(\iantimatchFamily\) \primal. As the last result of this section, we can build on Lemma~\ref{lem:sigood} and show that it suffices to focus our efforts to compute good $\mathcal{F}$-branch decompositions merely on unions of \primal classes. \iflong To this end, let us first prove the following simple claim.\fi\ifshort To this end, one can prove the following simple claim.\fi
		
	\antichain
	\iflong
	\begin{proof}
		Let \(X, Y\) be a partition of \(V(G)\) and let \(n\in\Nat\) be the largest integer such that \(G[X,Y]\) contains an induced subgraph isomorphic to \(H^n_{\ichain}\). If \(H^n_{\ichain} = ((a_1, \ldots, a_n), (b_1, \ldots, b_n), E_{\ichain})\), then it is easy to see that \(((a_2, \ldots, a_n), (b_1, \ldots, b_{n-1}), E_{\ichain})\) is isomorphic to \(H^{n-1}_{\ichainstrict}\). On the other hand if \(m\in\Nat\) is the largest integer such that \(G[X,Y]\) contains an induced subgraph isomorphic to \(H^m_{\ichainstrict} = ((a_1, \ldots, a_m), (b_1, \ldots, b_m), E_{\ichainstrict})\), then \(((a_1, \ldots, a_{m-1}), (b_2, \ldots, b_{m}), E_{\ichainstrict})\) is isomorphic to \(H^{m-1}_{\ichain}\).
	\end{proof}
	\fi
	
	\ifshort
	Our high-level strategy for dealing with \(\icompleteFamily\) and \(\iemptyFamily\) is to start by considering an arbitrary optimal branch decomposition $\decT$ of a graph $G$ for the class $\cF^*$ obtained after removing these two families from $\cF$. Intuitively, we show that if there were a large difference between the $\cF$-branchwidth of $\decT$ (which is equal to the \(\cF^*\)-branchwidth of \(G\)) and the $\cF$-branchwidth of $G$, then $G$ or the complement of \(G\) would have to contain a sufficiently large complete bipartite subgraph; but then \emph{any} decomposition of the graph would still necessarily have a cut inducing a fairly large complete bipartite graph or edgeless graph. This would contradict our assumption about the $\cF$-branchwidth of $G$ being substantially smaller than the $\cF$-branchwidth of $\decT$. Using this, we obtain:
	\fi
	
	\iflong
	To prove that \(\iemptyFamily\) and \(\icompleteFamily\) are redundant we will use the following well-known fact that for any subset of leaves of branch decomposition \(\decT\) there is an edge of \(\decT\) that separates the subset in balanced way. Let \(\decT = (V,E)\) be a a subcubic tree and consider a nonnegative weight function \(\mathbf{w}: V\rightarrow \mathbb{R}_{\ge 0}\). Let \(e\) be an edge in \(\decT\) and \(\decT_1, \decT_2\) be the two connected components of \(\decT - e\). We say that \(e\) is \(\alpha\)-balanced, for \(0< \alpha\le \frac{1}{2}\), if \(\alpha\cdot\mathbf{w}(\decT) \le \mathbf{w}(\decT_1) \le (1-\alpha)\cdot\mathbf{w}(\decT)\), where \(\mathbf{w}(\decT) = \sum_{v\in V(\decT)}\mathbf{w}(v)\) (note that \(\mathbf{w}(\decT) = \mathbf{w}(\decT_1)+ \mathbf{w}(\decT_2)\) so also  \(\alpha\cdot\mathbf{w}(\decT) \le \mathbf{w}(\decT_2) \le (1-\alpha)\cdot\mathbf{w}(\decT)\)). The following lemma is folklore and we include the proof only for completeness.
	
	\begin{lemma}
	\label{lem:balancedtrees}
		Let \(\decT = (V,E)\) be a a subcubic tree and consider a nonnegative weight function \(\mathbf{w}: V\rightarrow \mathbb{R}_{\ge 0}\) on vertices of \(\decT\). Then there exists a \(\frac{1}{3}\)-balanced edge in \(\decT\).
	\end{lemma}

\begin{proof}
	Let us orient the edges of \(\decT\) as follows. Let \(e\) be an edge in \(\decT\) and \(\decT_1, \decT_2\) be the two connected components of \(\decT - e\). We orient the edge from \(\decT_1\) to \(\decT_2\) if \(\mathbf{w}(\decT_1) < \frac{1}{2}\cdot \mathbf{w}(\decT) \) and from \(\decT_2\) to \(\decT_1\) otherwise. The directed graph obtained in this way is acyclic and it has a sink, i.e., vertex of out-degree zero. It is straightforward to verify that there is an edge incident to the sink such that \(\frac{1}{3}\cdot\mathbf{w}(\decT) \le \mathbf{w}(\decT_1) \le \frac{1}{2}\cdot\mathbf{w}(\decT)\).
\end{proof}

	Using the above lemma we get the following corollaries. 
	
	\begin{corollary}\label{cor:emptyLarge}
		Let \(n\in\Nat\) and let \(G\) be a graph such that there exists a partition \((A,B)\) of \(V(G)\) with \(H^{3n}_\emptyset\) isomorphic to an induced subgraph \(H\) of \(G[A,B]\). Then each branch decomposition \(\decT\) of \(G\) has an edge \(e\) that induces a partition \((X,Y)\) of \(G\) such that \(H^{n}_\emptyset\) is isomorphic to an induced subgraph of \(G[X,Y]\).
	\end{corollary}

\begin{proof}
	Let \(H= (A_H, B_H, E_H)\) such that \(A_H = A\cap V(H) \) and \(B_H = B\cap V(H)\). Let \(\mathbf{w}: V(\decT)\rightarrow \mathbb{R}_{\ge 0}\) such that \(\mathbf{w}(v)=1\) if \(v=\mathcal{L}(u)\) for some vertex \(u\in V(H)\) and \(\mathbf{w}(v)=0\), otherwise. Let \(e\) be a \(\frac{1}{3}\)-balanced edge of \(\decT\) and let \((X,Y)\) be the partition of \(G\) induced by \(e\) such that \(\mathbf{w}(X)\le \mathbf{w}(Y)\). Then \(\frac{|V(H)|}{3}\le |V(H)\cap X|\le \frac{|V(H)|}{2}\). Furthermore, let us assume, without loss of generality, that \(|A_H\cap X|\ge |B_H\cap X|\). 
	Clearly, \(|A_H\cap X|\ge n\) and \(|B_H\cap X|\le 2n\), hence \(|B_H\cap Y|\ge n\). It is easy to verify that \( (A_H\cap X)\cup (B_H\cap Y) \) induces a graph isomorphic to \(H^{n}_\emptyset\). 
\end{proof}
	Following an analogous argument for \(\icompleteFamily\) we get. 
\begin{corollary}\label{cor:completeLarge}
	Let \(n\in\Nat\) and let \(G\) be a graph such that there exists a partition \((A,B)\) of \(V(G)\) with \(H^{3n}_{\icomplete}\) isomorphic to an induced subgraph \(H\) of \(G[A,B]\). Then each branch decomposition \(\decT\) of \(G\) has an edge \(e\) that induces a partition \((X,Y)\) of \(G\) such that \(H^{n}_{\icomplete}\) is isomorphic to an induced subgraph of \(G[X,Y]\).
\end{corollary}
	
Given the above two corollaries it is straightforward to show the following lemma.
\fi
	
	\primalLemma

\iflong
	\begin{proof}
		Let \(\decT\) be an optimal \(\cF^*\)-branch decomposition and let \(e\) be an edge in \(\decT\) that induces partition \((X,Y)\) of \(G\). If \(\cF^*\)-\(\bw(\decT, e)< \FbwSh(\decT, e)=n\), then either 
		\begin{itemize}
		\item \(G[X,Y]\) contains an induced subgraph isomorphic to \(H^{n}_\emptyset\) and \(\iemptyFamily\subseteq  \cF\setminus \cF^* \), or
		\item it contains an induced subgraph isomorphic to \(H^{n}_{\icomplete}\) and \(\icompleteFamily\subseteq  \cF\setminus \cF^* \). 
		\end{itemize} In both cases, we have  \(\FbwSh(G)\ge \frac{n}{3}\) by Corollaries~\ref{cor:emptyLarge} and~\ref{cor:completeLarge}.  
	\end{proof}
	\fi
	
	We use $\cF^*$ to refer to the union of any combination of \texttt{\emph{primal}} families, and $\Fsbw(G)$ to refer to the $\cF^*$-branchwidth of $G$. If follows from Lemmas~\ref{lem:antichain} and \ref{lem:primal} that computing \(\cF^*\)-branchwidth is sufficient to obtain approximately optimal decompositions for all possible $\cF$-branchwidths.
We conclude this section by showing that for \textsc{$\cF^*$-Branchwidth}, i.e., the problem of computing an optimal $\cF^*$-branch decomposition for a fixed choice of $\cF^*$, it is sufficient to deal with connected graphs only.
		
		\componentLemma

\iflong
\begin{proof}
	We will prove the lemma by induction on the number of connected components. Clearly, if \(G\) is connected, then the statement hold.
	Now, if \(G\) be a disjoint union of graphs \(G_1\) and \(G_2\). Given a \(\cF^*\)-branch decomposition \(\decT\) of \(G\) of width \(k\), it is easy to see that the restriction of \(\decT\) to \(G_1\) (resp. \(G_2\)), that is the branch decomposition obtained from \(\decT\) by deleting the leaves that do not correspond to vertices of \(G_1\) (resp. \(G_2\)), is a  decomposition \(\decT\) of \(G_1\) (resp. \(G_2\)) of width \(k\). Note that this shows that the \(\cF^*\)-branchwidth of \(G\) is at least maximum of \(\cF^*\)-branchwidth of \(G_1\) and \(\cF^*\)-branchwidth of \(G_2\) which is in turn, by induction hypothesis, at least the maximum \(\cF^*\)-branchwidth over all connected components of \(G\).
	
	On other hand, let \(\decT_1\) and \(\decT_2\) be \(\cF^*\)-branch decompositions of \(G_1\) and \(G_2\) respectively. 
	We construct a \(\cF^*\)-branch decomposition of \(G\) of width \(\max(\cF^*\mbox{-}\bw(\decT_1), \cF^*\mbox{-}\bw(\decT_2))\). This shows that \(\cF^*\)-branchwidth of \(G\) is at most maximum of \(\cF^*\)-branchwidth of \(G_1\) and \(\cF^*\)-branchwidth of \(G_2\) which is by induction hypothesis the maximum \(\cF^*\)-branchwidth over all connected components of \(G\), finishing the proof. 
	
	Let \(e_1\) be an arbitrary edge of \(\decT_1\) and \(e_2\) an arbitrary edge of \(\decT_2\). Let us subdivide \(e_1\) and \(e_2\) and let the resulting vertices be \(w_1\) and \(w_2\), respectively. Note that subdivision of an edge in a branch decomposition does not change the width of the decomposition. Now let \(\decT\) be the \(\cF^*\)-branchwidth decomposition obtained by taking the disjoint union of  \(\decT_1\) and  \(\decT_2\) and adding the edge \(w_1w_2\). Let us show that the width of every edge is bounded by \(\max(\cF^*\mbox{-}\bw(\decT_1), \cF^*\mbox{-}\bw(\decT_2))\). First, consider the edge \(w_1w_2\). The partition of \(V(G)\) induced by this edge is \((V(G_1), V(G_2))\). There is no edge between \(V(G_1)\) and \(V(G_2)\). But \(\mathcal{F}^*\) is a union of a non-empty set of \primal\ families of graphs and every vertex in a bipartite graph in \(\mathcal{F}^*\) has degree at least one. So no induced subgraph of \(G[(V(G_1), V(G_2))]\) is isomorphic to a graph in \(\cF^*\) and the width of the edge \(w_1w_2\) is \(0\).
	Now let \(e\) be an edge in \( \decT_1\) and let \((X,Y)\) be the partition of \(V(G)\) induced by \(e\). It is easy to see that, depending of which connected component of \(\decT_1- e\) contains \(w_1\), we have either \(V(G_2)\subseteq X\) or \(V(G_2)\subseteq Y\). Let \(H\) be an induced subgraph of \(G[X,Y]\) isomorphic to a graph in \(\cF^*\). It follows that every vertex in \(V(H)\cap X\) has a neighbor in \(Y\) and  every vertex in \(V(H)\cap Y\) has a neighbor in \(X\). Since either \(V(G_2)\subseteq X\) or \(V(G_2)\subseteq Y\) and there is no edge between \(G_1\) and \(G_2\), we have \(V(H)\subseteq V(G_1)\) and \(\cF^*\mbox{-}\bw(\decT,e)=\cF^*\mbox{-}\bw(\decT_1,e)\). An analogous argument for an edge in \(\decT_2\) finishes the proof.
\end{proof}
\fi

	\section{Treewidth and Maximum Vertex Degree}
	\label{sec:tw}
Our aim in this section is to prove the following theorem.
	\iflong
	To streamline the presentation of this result, in this section we suppress the function $\decf$ of a branch decomposition \(\decT\) by assuming that the leaves of branch decompositions are actually the vertices of $G$.
	\fi
	
	\twDegTheorem

As an immediate corollary, by considering \(\mathcal{F}^* = \mathcal{F}_{\imatch}\) we obtain:
	\twDegMimWidthCorollary
	
On a high level, our proof of Theorem~\ref{thm:tw_deg} relies on the usual dynamic programming approach
	which computes a set of records for each node \(t \in V(T)\) of a minimum-width nice tree decomposition \((T, \chi)\) of \(G\) in a leaves-to-root manner.
	Here each record captures a set of partial \(\mathcal{F}^*\)-branch decompositions not exceeding \(\mathcal{F}^*\)-branchwidth at most \(\tw(G)\).
	It is worth noting that the most involved and difficult parts of our algorithm are necessary to handle the case that \(\mathcal{F}_{\imatch} \subseteq \mathcal{F}^*\).
	\iflong
	We first describe the records and give some informal intuition on their role in the dynamic program.
	\fi

	For \(t \in V(T)\), denote by \(T_t\) the subtree of \(T\) rooted at \(t\), and by \(G_t = G[\bigcup_{s \in V(T_t)} \chi(s)]\) the subgraph of \(G\) induced by the vertices in the bags of \(T_t\).
	Moreover, for \(v \in V(G)\) we denote its \emph{closed distance-\(i\) neighbourhood} in \(G\) by \(N^i[v]\), its \emph{open distance-\(i\) neighbourhood in \(G\) by \(N^i(v)\)}, and also extend these notations to vertex sets \(V' \subseteq V(G)\) by denoting \(N^i[V'] = \bigcup_{v \in V'} N^i[v]\) and \(N^i(V') = N^i[V'] \setminus V'\).
	
	We define a \emph{record} \((D, \flat, \Lambda, \sigma, \alpha^{\imatch}, \alpha^{\iantimatch}, \alpha^{\ichain})\) of \(t \in V(T)\) to consist of the following.
	\begin{itemize}[leftmargin=*]
		\item A binary tree \(D\) on at most \(2^{|N^3[\chi(t)]|}\) vertices, all leaves of which are identified with distinct vertices in \(N^3[\chi(t)]\).
		Intuitively, \(D\) will describe the restricted tree with respect to \(N^3[\chi(t)]\) of all branch decompositions of \(G_t\) that are captured by the record\ifshort .\fi
		\iflong
		 \((D, \flat, \Lambda, \sigma, \alpha^{\imatch}, \alpha^{\iantimatch}, \alpha^{\ichain})\).\fi
		\item A set \(\flat = \{\flat_e\}_{e \in E(D)}\) indexed by the edges of \(D\), where each \(\flat_e\) is a sequence of subsets of \(N^{3\tw(G)}(N^3[\chi(t)])\), such that \(\bigcup_{e \in E(D)} \flat_e\) is a partition of \(N^{3\tw(G)}(N^3[\chi(t)])\).
		\(\flat_e\) describes the order in which subtrees containing vertices in \(N^{3\tw(G)}(N^3[\chi(t)])\) are attached to the path which corresponds to \(e\) in any branch decomposition captured by \((D, \flat, \Lambda, \sigma, \alpha^{\imatch}, \alpha^{\iantimatch}, \alpha^{\ichain})\).
		The entry \(\flat\) can be interpreted to express a ``refinement'' of each edge of \(D\).
		Similarly to edges of \(D\) corresponding to paths in a branch decomposition captured by the record 
		we will define refinements of these paths in such branch decompositions into \emph{blocks} which correspond to pairs \((e,i)\) where \(e \in E(D)\) and~\(i \in [|\flat_e|+1]\).
		\item \(\Lambda = (\Lambda_e)_{e \in E(D)}\) indexed by the edges of \(D\), where each \(\Lambda_e\) is a sequence \(\Lambda_e = (\lambda_1, \dotsc, \lambda_{|\flat_e| + 1})\) of subsets of \(N^3(\chi(t))\).
		\(\lambda_i \in \Lambda_e\) should be such that the \imatchwidth value at every edge in the block corresponding to \((e,i)\) in any branch decomposition captured by \((D, \flat, \Lambda, \sigma, \alpha^{\imatch}, \alpha^{\iantimatch}, \alpha^{\ichain})\) can be achieved by a matching such that its intersection with \(G[N^3[\chi(t)]]\) is incident to precisely the vertices in \(\lambda_i\).
		In other words, \(\lambda_i \in \Lambda_e\) describes the interaction we assume any matching to have with \(G[N^3[\chi(t)]]\) when computing the \imatchwidth value at any edge in the block corresponding to \((e,i)\).
		\iflong
		We will show that such a set \(\lambda_i\) always exists in Corollary~\ref{cor:blockset}.
		\fi
		\ifshort
		We can argue (Lemma~\ref{lem:blockset}) that such a set \(\lambda_i\) exists.
		\fi

		\item A collection \(\sigma = (\sigma_{e,i})_{e \in E(D), i \in [|\flat_e| + 1]}\) of typical sequences \(\sigma_{e,i}\) with entries in \(\{0\} \cup [\tw(G) + 1]\).
		\(\sigma_{e,i}\) should describe a further refinement of each path corresponding to \((e,i)\) in a branch decomposition captured by \((D, \flat, \Lambda, \sigma, \alpha^{\imatch}, \alpha^{\iantimatch}, \alpha^{\ichain})\) according to the \imatchwidth values that can be achieved by the restriction of the branch decomposition to \(G_t\) under the condition imposed by the choice of \(\lambda_i\).

		The importance of these typical sequences is that we can argue that it is safe to assume that a partial branch decomposition at node \(t\) can be extended to a minimum-\(\Fsbw\) decomposition of \(G\) if and only if this is possible by attaching vertices only at positions in the typical sequences (Lemma~\ref{lem:typical}).

		\item \(\alpha^{\imatch}, \alpha^{\iantimatch}, \alpha^{\ichain} \in [\tw(G) + 1]\) should describe the \(\imatchwidth\), \(\iantimatchwidth\) and \(\ichainwidth\) of any branch decomposition of \(G_t\) captured by \iflong the record\fi \((D, \flat, \Lambda, \sigma, \alpha^{\imatch}, \alpha^{\iantimatch}, \alpha^{\ichain})\).
		For \(\alpha^{\imatch}\) there is a slight caveat -- to compute this entry in the record we will rely on the correctness of \(\Lambda\) which we cannot ensure at the time of computation, but rather retroactively after the root of \(T\) is processed.

		This is done by constructing a \emph{canonical decomposition} concurrently with each record, which is representative of all branch decompositions captured by \((D, \flat, \Lambda, \sigma, \alpha^{\imatch}, \alpha^{\iantimatch}, \alpha^{\ichain})\), and then after processing the root of \(T\) verifying that its \imatchwidth is at most \(\alpha^{\imatch}\).
	\end{itemize}

	Note that the number of records is easily seen to be \FPT\ parameterized by \(\tw(G)\) and \(\Delta(G)\) using 
	\iflong 
	Fact~\ref{fact:typ-bounds}.
	\fi
	\ifshort 
	the fact that the number of considered typical sequences is at most \(\frac{8}{3}2^{2(\tw(G) + 1)}\)~\cite{BodlaenderKloks96}.
	\fi
	We remark that the information stored in \(D\) is sufficient to keep track of the \iantimatchwidth and the \ichainwidth of the constructed partial decomposition (this means to update \(\alpha^{\iantimatch}\) and \(\alpha^{\ichain}\)) appropriately.
	For \imatchwidth the ``forgotten'' vertices, i.e.\ vertices in \(V(G_t) \setminus \chi(t)\) have a more far-reaching impact.
	\iflong
	More explicitly, all graphs in \(\mathcal{F}_{\iantimatch}\) and \(\mathcal{F}_{\ichain}\) have diameter at most \(3\).
	Hence any such graph induced in a cut of a partial branch decomposition which contains a vertex in \(\chi(t)\), can only contain vertices in \(N^3[\chi(t)]\).
	\fi
	\ifshort
	More explicitly, any graphs in \(\mathcal{F}_{\iantimatch}\) and \(\mathcal{F}_{\ichain}\) induced in a cut of a partial branch decomposition which contains a vertex in \(\chi(t)\), can only contain vertices in \(N^3[\chi(t)]\).
	\fi
	Obviously, this is not true for graphs in \(\mathcal{F}_{\imatch}\) which can have infinite diameter.
	
	Hence, for \imatchwidth instead of keeping track of vertices which can occur in graphs in \(\FFF^*\) together with vertices in \(\chi(t)\),
	we store the rough location in the constructed partial branch decomposition of a sufficiently large neighbourhood \(\tilde{N}\) of \(\chi(t)\) in \(G\) to separate the interaction (in terms of independence) of edges of \(G_t[V(G_t) \setminus \tilde{N}]\) and edges of \(G[V(G) \setminus (V(G_t) \cup \tilde{N})]\) that occur in independent induced matchings at cuts in the constructed partial branch decomposition.
	After this separation we still need to be able to actually derive the \imatchwidth.
	To do this we use typical sequences in a similar way as they are also used to compute e.g.\ treewidth, pathwidth and cutwidth~\cite{BodlaenderKloks96,TSB05a,Hamm19,BJT20}, as indicated in the description of \(\sigma\) above.
	
\iflong 
	The use of blocks is crucial for the applicability of typical sequences, as will become clear when we establish the correctness of using typical sequences in Lemma~\ref{lem:typical}. 
	Hence, we first describe blocks,
	then show how to correctly refine blocks using typical sequences in Section~\ref{sec:typical},
	and finally assimilate these components to show how to compute the records of our dynamic program at each type of node in the tree decomposition in Section~\ref{sec:dp} and finally prove Theorem~\ref{thm:tw_deg} in Section~\ref{sec:correct}.
\fi	

\iflong
\subsection{Blocks}
	\label{sec:blocks}
\fi

	\iflong
	Fix for this entire subsection some node \(t \in V(T)\) of our tree decomposition \((T,\chi)\) of \(G\),
	\fi
	\ifshort
	In the remainder of this subsection we give an overview of the technical results and formal definitions that justify our approach.
	For this fix a node \(t \in V(T)\) of a tree decomposition \((T,\chi)\) of \(G\),
	\fi
	let \(\decT\) be a branch decomposition of \(G\) of minimum \(\mathcal{F}^*\)-branchwidth, \(D = \lcac{\decT}{N^3[\chi(t)]}\), 
	for an edge \(e \in E(D)\), \(P_e\) denote the path in \(\decT\) that corresponds to \(e\), and
	\(\flat\) be such that for every \(e \in E(D)\), 
	\ifshort 
	 (1) the subtree of \(\decT - P_e\) containing \(v \in N^{3\tw(G)}(N^3[\chi(t)])\) is attached to an internal node of \(P_e\) if and only if \(v \in \bigcup_{d \in \flat_e} d\); and (2) for \(v,w \in \bigcup_{d \in \flat_e} d\), \(v\) is contained in a subtree of \(\decT - P_e\) \iflong which is \fi attached at a node \emph{before} (in an arbitrary but fixed traversal of \(P_e\)) the node at which the subtree of \(\decT - P_e\) containing \(w\) is attached if and only if the set in \(\flat_e\) that contains \(v\) is before the set in \(\flat_e\) that contains \(w\) in the sequence \(\flat_e\).
	\fi\iflong  
	\begin{itemize}
		\item the subtree of \(\decT - P_e\) containing \(v \in N^{3\tw(G)}(N^3[\chi(t)])\) is attached to an internal node of \(P_e\) if and only if \(v \in \bigcup_{d \in \flat_e} d\); and
		\item for \(v,w \in \bigcup_{d \in \flat_e} d\), \(v\) is contained in a subtree of \(\decT - P_e\) which is attached at a node \emph{before}\footnote{in an arbitrary but fixed traversal of \(P_e\)} the node at which the subtree of \(\decT - P_e\) containing \(w\) is attached if and only if the set in \(\flat_e\) that contains \(v\) is before the set in \(\flat_e\) that contains \(w\) in the sequence \(\flat_e\);
	\end{itemize}
	\fi 
	
	\BlockDefinition
\iflong 	
By the definition of \(\flat_e\), we immediately can observe the following.
\fi
	\NoBlocksObservation
	In this way we can immediately relate each block to a pair \((e,i)\) where \(e \in E(D)\) and \(i \in [|\flat_e| + 1]\), where we enumerate the blocks of \(e\) according to the fixed traversal of \(P_e\).
	For the remainder of this subsection we fix an edge \(e \in E(D)\) and its corresponding path \(P_e\)~in~\(\decT\).
	
	Next we show that blocks have the desirable property that we can consider induced matchings with uniformly restricted interaction with \(N^3[\chi(t)]\) when computing the \imatchwidth values at edges of a fixed block.
	\ifshort
	Even stronger this restricted interaction does not depend on the internal structure of a block which is important for the application of typical sequences.
	\fi
	
	\iflong 
	As an intermediate step we show that vertices that are contained in subtrees of \(\decT - P_e\) that are attached at an internal node of some block of \(P_e\) are not connected to \(N^3[\chi(t)]\) in any bipartite subgraph of \(G\) considered for \(\mathcal{F}^*\)-branchwidth values.
	\begin{proposition}
		\label{prop:intsep}
		Let \(P'\) be a block of \(P_e\).
		Any vertex of \(G\) that is contained in a subtree of \(\decT - P_e\) that is attached to an internal node of \(P'\) is not connected via a path to \(N^3[\chi(t)]\) in any bipartite subgraph of \(G\) induced by an edge of \(P'\).
	\end{proposition}
	\begin{proof}
		The fact that \(\decT\) achieves minimum \(\mathcal{F}^*\)-branchwidth for \(G\) and \(\Fsbw(G) \leq \tw(G) + 1\) (by Corollary~\ref{cor:FisTWBounded}) implies that the diameter of any bipartite subgraph of \(G\) induced by any edge of \(P'\) is at most \(3\tw(G) + 3\);
		otherwise the longest shortest path in such a bipartite subgraph is an induced path of length more than \(3\tw(G) + 3\) and hence includes an induced matching with more than \(\tw(G) + 1\) edges just by taking every third edge on the induced path.
		Hence the set of vertices that are connected via a path to \(N^3[\chi(t)]\) in any bipartite subgraph of \(G\) induced by any edge of \(P'\) necessarily lies in \(N^{3\tw(G)}(N^3[\chi(t)])\).
	\end{proof}

	Using this result we are able to show a slightly stronger statement than the desirable property described above.
	This stronger property will become critical when we show the correctness of our dynamic programming procedure for introduce nodes.
	It allows us to shift all vertices of \(G - N^3[\chi(G_t)]\) that are attached at internal nodes of a block in a hypothetical decomposition to positions distinguished by the used typical sequences, which rely on the properties within a block \emph{while maintaining these properties}.
	\fi
	\blocksetLemma
\iflong
	\begin{proof}
		Fix \(\tilde{\decT}\) and \(\tilde{P}\) as described in the statement of the lemma.
		Consider a set of pairwise independent edges \(E_L\) each of which is connected to some vertex in \(N^3[\chi(t)]\) in all bipartite subgraphs of \(G\) induced by any edge of \(\tilde{P}\), such that \(E_L\) is cardinality-maximum with these properties.
		Let \(E_\lambda = E_L \cap E(G[N^3(\chi(t))])\) and \(\lambda \subseteq N^3[\chi(t)]\) be the set of endpoints of all edges in \(E_\lambda\).
		We claim that for every edge \(\tilde{e} \in E(\tilde{P})\) there is an induced matching \(\tilde{M}\) in the bipartite subgraph of \(G\) induced by \(\tilde{e}\) with \(2\imatchwidth(\decT,\tilde{e})\) vertices such that the edges in \(E(\tilde{M}[N^3[\chi(t)]])\) are exactly \(E_\lambda\), and hence the edges in \(E(\tilde{M}[N^3[\chi(t)]])\) are incident to exactly \(\lambda\).
		
		To show this, fix an edge \(\tilde{e} \in E(\tilde{P})\) and an induced matching \(M^*\) in the bipartite subgraph of \(G\) induced by \(\tilde{e}\) with \(2\imatchwidth(\decT,\tilde{e})\) vertices.
		(Such a matching exists by definition of \(\imatchwidth\).)
		Obtain \(\tilde{M}\) from \(M^*\) by replacing the set of all edges in \(E(M^*)\), at least one of whose endpoints is connected to some vertex in \(N^3[\chi(t)]\) in all bipartite subgraphs of \(G\) induced by any edge of \(\tilde{P}\), with \(E_L\).
		By the maximality of \(E_L\), \(\tilde{M}\) has at least as many vertices as \(M^*\).
		Next we show that \(\tilde{M}\) is indeed an induced matching in the bipartite subgraph of \(G\) induced by \(\tilde{e}\).
		The edges in \(E_L\) are pairwise independent in the bipartite subgraph of \(G\) induced by \(\tilde{e}\) by choice of \(E_L\) and the edges in \(E(\tilde{M}) \setminus E_L\) are pairwise independent in the bipartite subgraph of \(G\) induced by \(\tilde{e}\) because they are contained in an induced matching \(M^*\).
		It remains to show that any edge in \(e_1 \in E_L\) is independent of any edge \(e_2 \in E(\tilde{M}) \setminus E_L\) in the bipartite subgraph of \(G\) induced by \(\tilde{e}\).
		By the choice of \(E_L\), \(e_1\) is connected to a vertex in \(N^3[\chi(t)]\) in the bipartite subgraph of \(G\) induced by any edge of \(\tilde{P}\), in particular the subgraph induced by \(\tilde{e}\).
		On the other hand, by the construction of \(\tilde{M}\) from \(M^*\), neither of the endpoints of \(e_2\) are connected to any vertex in \(N^3[\chi(t)]\) in all bipartite subgraphs of \(G\) induced by any edge of \(\tilde{P}\).
		In other words, there is some edge \(e' \in E(\tilde{P})\) such that in the bipartite subgraph of \(G\) induced by \(e'\), neither endpoint of \(e_2\) is connected to any vertex in \(N^3[\chi(t)]\).
		Proposition~\ref{prop:intsep} together with the fact that, apart from the subtrees of \(\decT - P_e\) which are attached at internal nodes of \(P'\), \(\decT\) and \(\tilde{\decT}\) are the same, implies that \(e_2\) not being connected to any vertex in \(N^3[\chi(t)]\) in the bipartite subgraph of \(G\) induced by \(e'\), means the same is true for all bipartite subgraphs of \(G\) induced any edge of \(\tilde{P}\), in particular the subgraph induced by \(\tilde{e}\).
		All together this shows that the endpoints of \(e_1\) and \(e_2\) are not even in the same connected component of the bipartite subgraph of \(G\) induced by \(\tilde{e}\), and thus \(e_1\) and \(e_2\) are independent in this subgraph of \(G\).
		
		Thus for each edge \(\tilde{e} \in E(\tilde{P})\) we are able to construct an induced matching \(\tilde{M}\) in the bipartite subgraph of \(G\) given by the cut of \(G\) induced by \(\tilde{e}\) with \(2\imatchwidth(\decT,\tilde{e})\) vertices, such that the set of all edges in \(E(\tilde{M})\), at least one of whose endpoints is connected to some vertex in \(N^3[\chi(t)]\) in all bipartite subgraphs of \(G\) induced by any edge of \(\tilde{P}\), is equal to \(E_L\).
		In particular, \(E(\tilde{M}[N^3[\chi(t)]]) = E_\lambda\).
	\end{proof}

	Applying this lemma for \(\tilde{\decT} = \decT\), we obviously get:
	\begin{corollary}
		\label{cor:blockset}
		Let \(P'\) be a block of \(P_e\).
		Then there is some \(\lambda \subseteq N^3[\chi(t)]\) such that for every edge \(e' \in E(P')\) there is an induced matching \(M'\) in the bipartite subgraph of \(G\) induced by \(e'\) with \(2\imatchwidth(\decT,e')\) vertices such that the edges in \(E(M'[N^3[\chi(t)]])\) are incident to exactly~\(\lambda\).
	\end{corollary}

	Moreover Lemma~\ref{lem:blockset} justifies the following definition.
\fi
\ifshort 
Lemma~\ref{lem:blockset} justifies the following definition.
\fi
	\LambdaBlockDefinition

	\iflong
	We point out already that the entry \(\Lambda\) in our dynamic programming records acts as a guess for which sets \(\lambda_{e,i}\) each block associated to \((e,i)\) is a \(\lambda_{e,i}\)-block for.
	\fi
	\ifshort
	Note that the entry \(\Lambda\) in our dynamic programming records acts as a guess for which sets \(\lambda_{e,i}\) each block associated to \((e,i)\) is a \(\lambda_{e,i}\)-block for.
	\fi
\iflong	At first glance this might seem unnecessary considering the proof of Lemma~\ref{lem:blockset}, as in the proof we characterised an appropriate \(\lambda_{e,i}\) for any arbitrary block explicitly.
	However this characterisation assumes full knowledge of \(\decT\) which we will not have during the dynamic programming procedure while constructing \(\decT\).
\fi	
\iflong 
	\subsection{Typical Sequences Within Blocks}
	\label{sec:typical}
 	
	Fix for this entire subsection some node \(t \in V(T)\) of our tree decomposition \((T,\chi)\) of \(G\), and
	let \(\decT\) be a branch decomposition of \(G\) of minimum \(\mathcal{F}^*\)-branchwidth, \(D\) be the restricted tree with respect to \(N^3[\chi(t)]\) of \(\decT\), and \(e \in E(D)\) correspond to the path \(P\) in \(\decT\).
\fi
	Now consider a \(\lambda\)-block \(P_\lambda\) of \(P\).
	We show that we can apply typical sequences in the usual way within \(P_\lambda\), i.e.\ assume that whenever some vertices in \(V(G) \setminus V(G_t)\) are attached in \(\decT\) at \(P_\lambda\), that they are attached at certain points of \(P_\lambda\) which can be distinguished even after using typical sequences for the \(\imatchwidth\) to compress \(P_\lambda\).
	
	\iflong 
	The following lemma serves as an intermediate step to cleanly separate vertices in \(V(G) \setminus V(G_t)\) from vertices in \(V_t\) within each subtree of \(\decT - P_\lambda\) attached at an internal vertex of \(P_\lambda\), without shifting these points of attachment.
	Note, that for this we do not yet need to make use of the properties of \(P_\lambda\) as a \(\lambda\)-block.
	
	\begin{lemma}
		\label{lem:attachtree}
		Let \(S\) be a subtree of \(\decT - P_\lambda\) that attaches at an internal node of \(P_\lambda\),
		let \(F = \lcac{S}{V(G) \setminus V(G_t)}\).
		Let \(\decT^*\) arise from \(\decT\) by replacing \(S\) by \(\lcac{S}{V(G_t)}\) and then and attaching \(F\) at the subdivided edge between the root of \(\lcac{S}{V(G_t)}\) and its neighbour on \(P_\lambda\).
		We call the tree that is now attached at \(P_\lambda\) in place of \(S\) \(S^*\).
		See Figure~\ref{fig:attachtree} for an example.
		Then \(\mimw_\lambda(\decT^*) \leq \mimw_\lambda(\decT)\), \(\iantimatchwidth(\decT^*) \leq \iantimatchwidth(\decT)\), and \(\ichainwidth(\decT^*) \leq \ichainwidth(\decT)\).
	\end{lemma}
	\begin{figure}
		\centering
		\begin{minipage}{.5\textwidth}
			\centering
			\includegraphics[page=1]{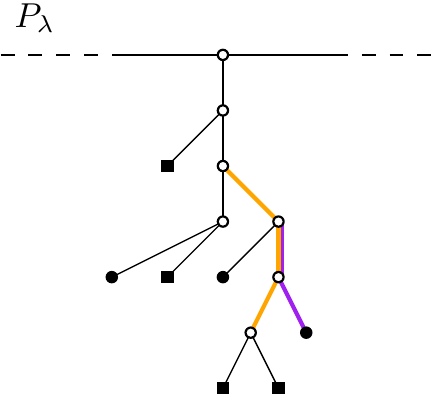}
		\end{minipage}
		\begin{minipage}{.5\textwidth}
			\centering
			\includegraphics[page=2]{fig/typ-section-figs}
		\end{minipage}
		\caption{\label{fig:attachtree}
			Modification of \(S\) (left) to \(S^*\) (right) as described in Lemma~\ref{lem:attachtree}.
			Filled vertices correspond to vertices in \(V(G)\), and among these the squares correspond to vertices in \(V(G) \setminus V(G_t)\).
			The cross-vertex on the right is the subdivision vertex used to attach \(F\) to \(S[V(G) \setminus V(G_t)]\).
			The cutfunction values at the orange edge in \(S^*\) are upper-bounded by the cutfunction values at any of the three orange edges in \(S\),
			and the cutfunction values at the purple edge in \(S^*\) are upper-bounded by the cutfunction values at any of the two purple edges in \(S\).
		}
	\end{figure}
	\begin{proof}
		All cuts induced by edges in \(\decT^*[V(\decT) \setminus V(S^*)]\) correspond immediately to identical cuts in \(\decT[V(\decT) \setminus V(S)]\).
		Hence it is sufficient to consider edges in \(S^*\).
		For more convenient phrasing, we also say a cut of \(G\) induced by an edge in \(S^*\) (or \(S\)) is a \emph{\(G\)-cut} in \(S^*\) (or \(S\)).
		
		By construction of \(S^*\), every \(G\)-cut in \(S^*\) which is not identical to some \(G\)-cut in \(S\), has either only vertices in \(V(S^*) \cap (V(G) \setminus V(G_t))\), or only vertices in \(V(S^*) \cap V(G_t)\) on one side.
		
		\(\lcac{S}{V(G) \setminus V(G_t)}\) does not contain any vertices from \(N^3[\chi(t)]\).
		By properties of tree decompositions, it also does not contain any vertices from \(N^3[V(G) \setminus V(G_t)]\).
		Hence when considering the cutfunctions \iantimatchwidth and \ichainwidth at some \(G\)-cut in \(S^*\) which is not identical to some \(G\)-cut in \(S\), they are bounded by the maximum of the respective cutfunction attained at an edge that separates the same vertex subset of \(N^3[V(G) \setminus V(G_t)]\), or the same subset of \(N^3[V(G_t) \setminus N^3[\chi(t)]]\) from \(P_\lambda\) respectively.
		Such an edge exists by construction of \(S^*\) from \(S\).
		
		It remains to consider \(\mimw\).
		Here the situation is similar as for \iantimatchwidth, and \ichainwidth.
		While a single independent matching across a \(G\)-cut in \(S^*\) can involve edges adjacent to vertices in \(V(G) \setminus V(G_t)\), \emph{and} edges adjacent to vertices in \(V(G_t) \setminus \chi(t)\), there are no edges between vertices in \(V(G) \setminus V(G_t)\) and vertices in \(V(G_t) \setminus \chi(t)\).
		Hence, once again such a \(G\)-cut can be compared to a \(G\)-cut in \(S\) which separates the same set of vertices vertex subset of \(N^3[V(G) \setminus V(G_t)]\), or the same subset of \(N^3[V(G_t) \setminus N^3[\chi(t)]]\) from \(P_\lambda\) respectively.
	\end{proof}
		
	For the next lemma, we introduce the following notation: \fi
	For a branch decomposition \(\decT'\) of \(G_t\) and an edge \(f \in E(\decT')\) 
	\iflong 
	we use \(\mimw(\decT',f)\) to denote the \(\imatchwidth\) achieved in the cut of \(\decT'\) at \(f\).
	Moreover 
	\fi 
	we use \(\mimw_\lambda(\decT',f)\) to denote the maximum size of an induced matching \(M\) in the cut of \(\decT'\) at \(f\) in which the set of vertices in \(N^3[\chi(t)]\) adjacent to edges in \(M[N^3[\chi(t)]]\) is equal to \(\lambda\).
	In the pathological case that there is no such matching in which the set of vertices in \(N^3[\chi(t)]\) adjacent to edges in \(M[N^3[\chi(t)]]\) is equal to \(\lambda\), we set \(\mimw_\lambda(\decT',f) = \bot\).
	\iflong 
	For all later considerations the sum of any number with \(\bot\) equals \(\bot\), and \(\bot\) is treated as larger than any number.
	\fi 
	Finally for two edges \(p, q \in E(P_\lambda)\), we write \(P_\lambda(p,q)\) to denote the subpath of \(P_\lambda\) between \(p\) and \(q\) and excluding \(p\) and \(q\).
	We do not make any distinction between \(P_\lambda(p,q)\) and \(P_\lambda(q,p)\).
	\typicalLemma

	\begin{figure}
		\begin{minipage}[c]{.5\textwidth}
			\includegraphics[page=3]{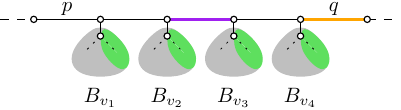}
			\vspace{0.5em}
		
		\includegraphics[page=4]{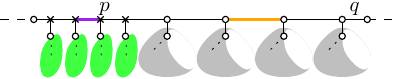}
		\end{minipage}
		\hfill
		\begin{minipage}[c]{.48\textwidth}
		\caption{\label{fig:typical}
			Illustration of modification of \(\decT\) (top) to \(\decT^{**}\) (bottom) as described in Lemma~\ref{lem:typical}.
			\iflong 
			(For simplicity of depiction, we assume \(\decT\) to have the form we can guarantee by the application of Lemma~\ref{lem:attachtree} to all \(\decT_v\); i.e.\ each \(F_v\) is a subtree of \(\decT_v\).)\fi
			The cross-vertices in the bottom figure are the subdivision vertices used to attach the \(\lcac{\decT_v}{V(G) \setminus V(G_t)}\) to \(p\).
			The cutfunction values at the orange and purple edge in \(\decT^{**}\) are upper-bounded by the cutfunction values at the orange and purple edge in \(\decT^*\) respectively.}
		\end{minipage}
	\end{figure}

	\iflong
	\begin{proof}
		As a first step, we make the following observation:
		Let \(\decT^*\) arise from \(\decT\) by applying the construction described in the statement of Lemma~\ref{lem:attachtree} to each \(\decT_v\).
		Correspondingly we refer to each tree that replaces \(\decT_v\) in \(\decT^*\) as \(\decT^*_v\).
		In this way, each \(F_v\) is actually a subtree of \(\decT^*_v\),
		while the cuts at edges in \(P_\lambda\) remain completely unchanged from \(\decT\) to \(\decT^*\).
		Moreover, because the considered cutfunction values in \(\decT\) upper-bound those in \(\decT^*\), it suffices to compare the cutfunction vaules in \(\decT^{**}\) to those in \(\decT^*\).
		
		First consider some internal node \(v\) of \(P_\lambda(p,q)\), and the cut in \(\decT^{**}\) induced by an edge of \(\decT_v[V(\decT_v) \setminus (V(G) \setminus V(G_t))]\).
		This cut is exactly the same as the cut in \(\decT^*\) induced by the same edge in \(\decT^*_v - F_v\).
		Because of the structure of \(\decT^*\) an analogous argument can be made for each cut in \(\decT^{**}\) induced by an edge of \(F_v\);
		this cut is identical to the one in \(\decT^*\) induced by the same edge in \(F_v\).
		
		It remains to consider the cuts in \(\decT^{**}\) which are induced by edges that arose by subdividing \(p\) and edges of \(P_\lambda\).
		We treat both these types of edges separately, and in each case first consider \iantimatchwidth and \ichainwidth, and then turn our attention to \(\mimw_\lambda\).
		
		Let \(x \in E(\decT^{**})\) be an edge that arose by subdividing \(p\),
		and let \(F_{v_i}\) and \(F_{v_{i + 1}}\) be the subtrees of the form \(\lcac{\decT_v}{V(G) \setminus V(G_t)}\) attached at each endpoint of \(x\) (not necessarily both, but at least one of them, exist).
		Then consider the edge \(y \in E(P_\lambda)\) the left endpoint of which \(\decT^*_{v_i}\) is attached at (if \(v_i\) is not defined, then let \(y \in E(P_\lambda)\) be the edge the right endpoint of which \(\decT^*_{v_{i + 1}}\)) (see Figure~\ref{fig:typical}, purple).
		We claim that the cutfunction values in \(\decT^{**}\) at \(x\) are at most as large as the respective cutfunction values in \(\decT\) at \(y\), or in one situation the respective cutfunction values in \(\decT\) at \(p\).
		For \iantimatchwidth and \ichainwidth, the subgraph of \(G\) that determines the cutfunction values at any edge of a branch decomposition either
		(i) only has vertices in \(N^3[\chi(t)]\),
		(ii) only has vertices in \(V(G) \setminus V(G_t)\), or
		(iii) only has vertices in \(V(G_t) \setminus \chi(t)\),
		because of the bounded diameter of antimatchings and chaingraphs, as well as the properties of treedecompositions.
		Obviously all edges in \(P_\lambda\) and \(P_\lambda\) with subdivided \(p\) separate the same vertex sets in \(N^3[\chi(t)]\).
		Hence cutfunction values attained by graphs of type (i) are the same in \(\decT^{**}\) at \(x\), and in \(\decT\) at \(y\).
		Note that by construction the edges between vertices in \(V(G) \setminus V(G_t)\) that are contained in the cut in \(\decT\) at edge \(y\) are the same as the ones in the cut in \(\decT^{**}\) at edge \(x\).
		We point out that in particular this is also true for all edges between vertices in \(V(G) \setminus V(G_t)\) that can be included in an independent matching in combination with \(\lambda\); this will be used when we consider \(\mimw_\lambda(\decT^{**},x)\).
		Hence also the cutfunction values attained by graphs of type (ii) are the same in \(\decT^{**}\) at \(x\), and in \(\decT\) at \(y\).
		For (iii) we actually do not compare the cutfunction values in \(\decT^{**}\) at \(x\) with the ones in \(\decT\) at \(y\), but rather with the ones in \(\decT\) at \(p\), because the edges between vertices in \(V(G_t) \setminus \chi(t)\) that are contained in the cut in \(\decT\) at edge \(p\) are the same as the ones in the cut in \(\decT^{**}\) at edge \(x\).
		
		Now we consider \(\mimw\).
		Our definition of \(\lambda\)-blocks, the fact that \(P_\lambda\) is a \(\lambda\)-block, and the way \(\decT^{**}\) is constructed from \(\decT\) allow us to apply Lemma~\ref{lem:blockset} to obtain that \(\mimw(\decT^{**},x) = \mimw_\lambda(\decT^{**},x)\).
		As \(N^3[\chi(t)]\) acts as a separator between \(N^3[V(G_t)]\) and \(N^3[V(G) \setminus V(G_i)]\) and for \(\mimw_\lambda\) the behaviour of any considered induced matching is determined on \(G[N^3[\chi(t)]]\) by \(\lambda\), any independent matching in which each matching edge involves at least one vertex in \(V(G) \setminus V(G_t)\) that is compatible with (i.e.\ independent of) \(\lambda\) is also compatible with any independent matching in which each matching edge involves at least one vertex in \(V(G_t) \setminus N^3[\chi(t)]\) that is compatible with \(\lambda\).
		Note that by assumption~\ref{typical:cond3}, the vertex subsets of \(\lambda\) separated by the cut in \(\decT^{**}\) at \(x\) are the same as the vertex subsets of \(\lambda\) separated by the cut in \(\decT\) at \(y\).
		Recall that the same is true for the edges between vertices in \(V(G) \setminus V(G_t)\).
		Hence, it merely remains to compare how many edges which have at least one endpoint in \(V(G_t) \setminus N^3[\chi(t)]\) can contribute to \(\mimw_\lambda(\decT^{**},x)\) and how many such edges can contribute to \(\mimw_\lambda(\decT,y)\).
		More specifically, it remains to consider edges which have at least one endpoint in some \(\decT_v\) before or in \(\decT_{v_i}\), since all other edges with at least one endpoint in \(V(G_t) \setminus N^3[\chi(t)]\) are either both separated in \(\decT^{**}\) at edge \(x\) and in \(\decT\) at position \(y\), or both not separated in \(\decT^{**}\) at edge \(x\) and in \(\decT\) at position \(y\).
		Moreover, any edge with an endpoint in some \(\decT_v\) before or in \(\decT_{v_i}\) has both endpoints in \(G_t\) by assumption~\ref{typical:cond3}, and the properties of tree decompositions and the edges we are concerned with are counted as \(\mimw_\lambda(\decT', y) - \ell\) for the cut in \(\decT\) at \(y\) and as \(\mimw_\lambda(\decT', p) - \ell\) for the cut in \(\decT^{**}\) at \(x\), where \(\ell\) denotes the number of edges which are separated by both cuts.
		By assumption~\ref{typical:cond1}, it holds that \(\mimw_\lambda(\decT', p) \leq \mimw_\lambda(\decT', y)\), all together showing that \(\mimw(\decT^{**}, x) = \mimw_\lambda(\decT^{**}, x) \leq \mimw_\lambda(\decT,y) \leq \mimw(\decT)\) in this case.
		
		We turn to edges in \(P_\lambda\), i.e.\ edges of \(\decT^{**}\) on the path corresponding to \(P_\lambda\) which did not arise from subdividing \(p\).
		Let \(x\) by such an edge and we identify \(x\) with the corresponding edge of \(\decT\) and claim that the cutfunction values in \(\decT^{**}\) at \(x\) are at most as large as the respective cutfunction values in \(\decT\) at \(q\) (see Figure~\ref{fig:typical}, orange), or in one situation the respective cutfunction values in \(\decT\) at \(x\).
		Once again, for \iantimatchwidth and \ichainwidth, the subgraph of \(G\) that determines the cutfunction values at any edge of a branch decomposition either
		(i) only has vertices in \(N^3[\chi(t)]\),
		(ii) only has vertices in \(V(G) \setminus V(G_t)\), or
		(iii) only has vertices in \(V(G_t) \setminus \chi(t)\),
		because of the bounded diameter of antimatchings and chaingraphs, as well as the properties of treedecompositions.
		All edges in \(P_\lambda\) and \(P_\lambda\) with subdivided \(p\) separate the same vertex sets in \(N^3[\chi(t)]\), hence cutfunction values attained by graphs of type (i) are the same in \(\decT^{**}\) at \(x\), and in \(\decT\) at \(y\).
		By construction the edges between vertices in \(V(G) \setminus V(G_t)\) that are contained in the cut in \(\decT\) at edge \(q\) are the same as the ones in the cut in \(\decT^{**}\) at edge \(x\).
		In particular this is also true for all edges between vertices in \(V(G) \setminus V(G_t)\) that can be included in an independent matching in combination with \(\lambda\); this will be used when we consider \(\mimw_\lambda(\decT^{**},x)\).
		Hence also the cutfunction values attained by graphs of type (ii) are the same in \(\decT^{**}\) at \(x\), and in \(\decT\) at \(y\).
		For (iii) we actually do not compare the cutfunction values in \(\decT^{**}\) at \(x\) with the ones in \(\decT\) at \(q\), but rather with the ones in \(\decT\) at \(x\), because the edges between vertices in \(V(G_t) \setminus \chi(t)\) that are contained in the cut in \(\decT\) at edge \(x\) are the same as the ones in the cut in \(\decT^{**}\) at edge \(x\).
		
		Now we consider \(\mimw_\lambda\).
		Again, our definition of \(\lambda\)-blocks, the fact that \(P_\lambda\) is a \(\lambda\)-block, and the way \(\decT^{**}\) is constructed from \(\decT\) allow us to apply Lemma~\ref{lem:blockset} to obtain that \(\mimw(\decT^{**},x) = \mimw_\lambda(\decT^{**},x)\).		
		As in the previous case (when \(x\) was an edge that arose from subdividing \(p\)), any independent matching in which each matching edge involves at least one vertex in \(V(G) \setminus V(G_t)\) that is compatible with (i.e.\ independent of) \(\lambda\) is also compatible with any independent matching in which each matching edge involves at least one vertex in \(V(G_t) \setminus N^3[\chi(t)]\) that is compatible with \(\lambda\).
		By assumption~\ref{typical:cond3}, the vertex subsets of \(\lambda\) separated by the cut in \(\decT^{**}\) at \(x\) are the same as the vertex subsets of \(\lambda\) separated by the cut in \(\decT\) at \(q\).
		Recall that the same is true for the edges between vertices in \(V(G) \setminus V(G_t)\).
		Hence, it merely remains to compare how many edges which have at least one endpoint in \(V(G_t) \setminus N^3[\chi(t)]\) can contribute to \(\mimw_\lambda(\decT^{**},x)\) and how many such edges can contribute to \(\mimw_\lambda(\decT,q)\).
		More specifically, it remains to consider edges which have exactly one endpoint in some \(\decT_v\) after \(x\), since all other edges with at least one endpoint in \(V(G_t) \setminus N^3[\chi(t)]\) are either both separated in \(\decT^{**}\) at edge \(x\) and in \(\decT\) at position \(q\), or both not separated in \(\decT^{**}\) at edge \(x\) and in \(\decT\) at position \(q\).
		Moreover, any edge with an endpoint in some \(\decT_v\) before or in \(\decT_{v_i}\) has both endpoints in \(G_t\) by assumption~\ref{typical:cond3}, and the properties of tree decompositions and the edges we are concerned with are counted as \(\mimw_\lambda(\decT', q) - \ell\) for the cut in \(\decT\) at \(q\) and as \(\mimw_\lambda(\decT', x) - \ell\) for the cut in \(\decT^{**}\) at \(x\), where \(\ell\) denotes the number of edges which are separated by both cuts.
		By assumption~\ref{typical:cond2}, it holds that \(\mimw_\lambda(\decT', x) \leq \mimw_\lambda(\decT', q)\), all together showing that \(\mimw(\decT^{**}, x) = \mimw_\lambda(\decT^{**}, x) \leq \mimw_\lambda(\decT,q) \leq \mimw(\decT)\).
	\end{proof}
	\fi
	
	In the later application of Lemma~\ref{lem:typical}, \(P'(p,q)\) takes the role of a path that is contracted in for obtaining any of the typical sequences in a record, and iteratively applying this lemma ensures the ``safeness'' (w.r.t.\ \(\cF^*\)-branchwidth and the properties of blocks) of using them.
	
	\ifshort
	With the machinery now in place, we are able to describe the dynamic programming procedure which we use to prove Theorem~\ref{thm:tw_deg}. (\(\star\), \emph{Sections~4.3 and~4.4}).
	\fi
	
	\iflong
	\subsection{Dynamic Programming Procedure}
	\label{sec:dp}
	With the machinery from the previous two subsections in place we are able to describe the dynamic programming procedure which we use to prove Theorem~\ref{thm:tw_deg}.
	We show correctness of this procedure in the following subsection (Section~\ref{sec:correct}).
	
	We traverse \(T\) in leaves-to-root order and compute a set \(\mathcal{R}(t)\) of records together with corresponding \emph{canonical} branch decompositions \(\decT\) of \(G_t\) for each node \(t \in V(T)\).
	Intuitively the information in the records allows us to iteratively extend the associated branch decompositions while bounding \(\iantimatchwidth\) and \(\ichainwidth\) by \(\alpha^{\iantimatch}\) and \(\alpha^{\ichain}\) respectively.
	For \imatchwidth we face a slight caveat.
	During the dynamic programming we are not able to keep track of the actual \imatchwidth without keeping a significant amount of information about forgotten vertices, which we cannot allow ourselves without exceeding any bound in our parameters.
	Instead we store a bound \(\alpha^{\imatch}\) on the \imatchwidth of the constructed branch decompositions \emph{assuming that in the records by which a certain record is reached each \(\sigma_{e,i}\) corresponds to a \(\lambda_i \in \Lambda_e\)-block which is consistent with the choices made for blocks at lower levels of \(T\)}.
	Checking whether this assumption is actually consistent with the construction in the dynamic programming procedure is the reason we include the canonical branch decompositions, some edges of which have `\emph{pointers}' to elements of sequences in \(\sigma\) for the respective record.
	
	Keeping this in mind, we first describe the four procedures by which we compute \(\mathcal{R}(t)\) depending on the type of \(t\) assuming \(\mathcal{F}\) contains all of \(\mathcal{F}_{\imatch}\), \(\mathcal{F}_{\iantimatch}\) and \(\mathcal{F}_{\ichain}\).
	It can easily be seen that whenever \(\mathcal{F}_{\imatch}\) is not contained in \(\mathcal{F}\), the entries \(\flat\), \(\Lambda\) and \(\alpha^{\imatch}\) as well as the canonical branch decompositions for records can be safely omitted in each step, as the computation of the other entries do not depend on them.
	Analogously whenever \(\mathcal{F}_{\iantimatch}\) is not contained in \(\mathcal{F}\), \(\alpha^{\iantimatch}\) can safely be omitted from all computation steps, and whenever \(\mathcal{F}_{\ichain}\) is not contained in \(\mathcal{F}\), \(\alpha^{\ichain}\) can safely be omitted from all computation steps.
	All of the following constructions are easily seen to be executable in \FPT\ time parameterized by \(\tw(G)\) and \(\Delta(G)\) using
	Facts~\ref{fact:typ-bounds}, \ref{fact:typ-comp} and \ref{fact:typ-compint}.
	\paragraph*{Leaf nodes.}
		Consider a leaf node \(t \in V(T)\), and let \(\chi(t) = \{v\}\).
		We define \(\mathcal{R}(t)\) to consist of the following pairs of records and canonical branch decompositions of \(G_t\).\\
		For each binary tree \(\decT = D\) with leaves \(N^3[v]\), i.e.\ each branch decomposition of \(G_t\), each set \(\flat = (\flat_e)_{e \in E(D)}\) indexed by the edges of \(D\), where each \(\flat_e\) is a sequence of subsets of \(N^{3\tw(G)}(N^3[\chi(t)])\), such that \(\bigcup_{e \in E(D)} \flat_e\) is a partition of \(N^{3\tw(G)}(N^3[\chi(t)])\),
		and each assignment of pairs \((e,i)\) with \(e \in E(D)\) and \(i \in [|\flat_e| + 1]\) to a subset \(\lambda_{e,i}\) of \(N^3[v]\), we construct a record \((D,\flat,(\Lambda_e = (\lambda_{e,1}, \dotsc, \lambda_{|\flat_e| + 1}))_{e \in E(D)}, \sigma,\alpha^{\imatch},\alpha^{\iantimatch},\alpha^{\ichain})\).
		Because the size of \(N^3[v]\) is bounded in \(\tw(G)\) and \(\Delta(G)\), \(\alpha^{\iantimatch} = \iantimatchwidth(\decT)\) and \(\alpha^{\ichain} = \ichainwidth(\decT)\) can be computed by brute force.
		The same is true for \(\sigma = (\sigma_{e,i})_{e \in E(D), i \in [|\flat_e| + 1]}\) for which we set \(\sigma_{e,i} = (\mimw_{\lambda_i}(\decT,e))\).
		Finally we set \(\alpha^{\imatch}\) to the maximum entry in any of the \(\sigma_{e,i}\).\\
		If any of entry of a sequence in \(\sigma\) or any of \(\alpha^{\imatch},\alpha^{\iantimatch},\alpha^{\ichain}\) exceeds \(\tw(G) + 1\) we discard the record and branch decomposition currently being constructed and move on to the next choice for a combination of \(D\), \(\flat\) and \(\Lambda\).
		
	\paragraph*{Introduce nodes.}
		Consider an introduce node \(t \in V(T)\) with child \(t' \in V(T)\) for which \(\mathcal{R}(t')\) was already computed, and let \(\chi(t) \setminus \chi(t') = \{v\}\).
		We define \(\mathcal{R}(t)\) to consist of the following pairs of records and branch decompositions of \(G_t\).
		For each record \((D', \flat', \Lambda', \sigma', \alpha^{{\imatch}'}, \alpha^{{\iantimatch}'}, \alpha^{{\ichain}'})\) in \(\mathcal{R}(t')\) with canonical branch decomposition \(\decT'\) of \(G_{t'}\), we proceed as follows:\\		
		We branch on \(\ell \leq |N^3[\chi(t)] \setminus N^3[\chi(t')]|\) binary trees \(\decT_1, \dotsc, \decT_n\) the union of whose leaves is exactly \(N^3[\chi(t)] \setminus N^3[\chi(t')]\).
		For each \(i \in [\ell]\) branch on an entry \(s_i\) of a sequence in \(\sigma'\) at the subdivision of whose pointing edge of \(\decT'\), \(T_i\) should be attached.
		For \(I \subseteq [\ell]\) such that \(\forall i, j \in I, \ s_i = s_j\) we additionally branch on a permutation \(\pi\) of \(I\).
		Then let \(D\) arise from \(D'\) by attaching each \(\decT_i\) at a subdivision of the edge \(e \in E(D')\) such that there is some \(j\) such that \(s_i\) is an entry of \(\sigma_{e,j}\).
		For \(I \subseteq [\ell]\) such that \(\forall i, j \in I, \ s_i = s_j\) we attach the \(\decT_i\) with \(i \in I\) in root to leaf order according to \(\pi\).
		We can concurrently perform the corresponding insertions on \(\decT'\) to obtain \(\decT\), where an edge of \(\decT\).\\
		Now branch on a set \(\flat = \{\flat_e\}_{e \in E(D)}\) indexed by the edges of \(D\), where each \(\flat_e\) is a sequence of subsets of \(N^{3\tw(G)}(N^3[\chi(t)])\), such that \(\bigcup_{e \in E(D)} \flat_e\) is a partition of \(N^{3\tw(G)}(N^3[\chi(t)])\).
		Before proceeding, we check whether the information branched on (specifically \(D\) and \(\flat\)) up till now is consistent with \(\flat'\).
		Now for each edge \(e \in E(\decT_i)\) and \(j \in [|\flat_e| + 1]\) for one of the attached \(\decT_i\), just as when considering leaf nodes, we branch on a subset \(\lambda_{e,i}\) of \(N^3[\chi(t)]\).
		Each other edge \(e\) of \(D\), is also an edge \(e' = e\) of \(D'\), or a subdivision of an edge \(e'\) of \(D'\).
		We distinguish these two cases when describing how to construct \(\lambda_{e,i}\).\\
		\begin{ourcase}[\(e\) is also an edge \(e'\) of \(D'\)]
			For each \(i \in [|\flat'_{e'}| + 1]\) we branch on a partition of \(\sigma'_{e',i}\) into subsequences \(\rho^i_1, \dotsc, \rho^i_{r_i}\) such that \(\sum_{i \in [|\flat'_{e'}| + 1]} r_i = |\flat_e|\).
			For each subsequence \(\rho^i_j\) we branch on a set of pairwise independent edges \(A^i_j\) of \(G[N^3[\chi(t)]]\) each of which has at least one endpoint in \(N^3[\chi(t)] \setminus N^3[\chi(t')]\).
			Then we set \(\lambda_{e,1}, \dotsc, \lambda_{e,r_1}, \dotsc, \lambda_{e,|\flat_e| + 1}\) to be \(\lambda'_{e',1} \cup V(A^1_1), \dotsc, \lambda'_{e',1} \cup V(A^1_{r_1}), \dotsc, \lambda'_{e,|\flat'_{e'}| + 1} \cup V(A^{|\flat'_{e'}| + 1}_{r_{|\flat'_{e'}| + 1}})\) respectively,
			and \(\sigma_{e,i}\) for \(i \in [|\flat_e| + 1]\) to be \((\bot)\) if not all edges in \(A^{i'}_j\) are in the \(G\)-cut induced by \(e\) in \(D\), and to be the sequence \((z + |A^{i'}_j|)_{z \in \rho^{i'}_j}\) otherwise,
			where \(i'\) and \(j\) are such that \(\lambda_{e,i} = \lambda'_{e',i'} \cup V(A^{i'}_j)\).
		\end{ourcase}
		\begin{ourcase}[\(e\) is a subdivision of an edge \(e'\) of \(D'\)]
			We use similar constructions as in Case 1 with the main difference being that for subdivisions of \(e' \in E(D)\) arising from attaching \(\decT_i\) at \(e'\) we subdivide all \(\sigma'_{e',j}\) at the entries which correspond to any insertion point \(s_i\).
			By this we mean that we duplicate each entry \(s_i\) in any \(\sigma'_{e',j}\) (note that not all \(\sigma'_{e',j}\) contain such insertion point; we only modify those that do), and partition \(\sigma'_{e',j}\) into subsequences at the points between duplicates.
			Each of the arising subsequences can be associated to a unique subdivision of \(e'\) which is an edge \(e\) of \(D\) in a natural way.
			We artificially set \(\flat'_e = \emptyset\), and \(\lambda'_{e,1}\) to be \((\lambda_{e',j})\) and \(\sigma'_{e,1}\) to be the subsequence associated to \(e\).
			With this setup we proceed as in Case 1.\\
		\end{ourcase}
		Finally we set \(\alpha^{\imatch}\) to be the maximum over all entries of the constructed sequences \(\sigma_{e,i}\),
		\(\alpha^{\iantimatch} = \max\{\alpha^{{\iantimatch}'}, \iantimatchwidth(D)\}\),
		and \(\alpha^{\ichain} = \max\{\alpha^{{\ichain}'}, \ichainwidth(D)\}\) where we note that \(\iantimatchwidth\) and \(\ichainwidth\) of the branch decomposition \(D\) of \(G[N^3[\chi(t)]]\) can be computed by brute force.\\
		If any of entry of a sequence in \(\sigma\) or any of \(\alpha^{\imatch},\alpha^{\iantimatch},\alpha^{\ichain}\) exceeds \(\tw(G) + 1\) we discard the record and branch decomposition currently being constructed and move on to the next branch.
		
		\paragraph*{Forget nodes.}
		Consider a forget node \(t \in V(T)\) with child \(t' \in V(T)\) for which \(\mathcal{R}(t')\) was already computed, and let \(\chi(t') \setminus \chi(t) = \{v\}\).
		We define \(\mathcal{R}(t)\) to consist of the following pairs of records and canonical branch decompositions of \(G_t\).
		For each record \((D', \flat', \Lambda', \sigma', \alpha^{{\imatch}'}, \alpha^{{\iantimatch}'}, \alpha^{{\ichain}'})\) in \(\mathcal{R}(t')\) with canonical branch decomposition \(\decT'\) of \(G_{t'}\), we proceed as follows:\\
		We let \(\decT = \decT'\), \(D = \lcac{D'}{N^3[\chi(t)]}\), \(\alpha^{\imatch} = \alpha^{{\imatch}'}\), \(\alpha^{\iantimatch} = \alpha^{{\iantimatch}'}\) and \(\alpha^{\ichain} = \alpha^{{\ichain}'}\).\\
		Moreover, because \(D\) is a restricted tree of \(D'\) with respect to \(N^3[\chi(t)]\) each \(e \in E(D)\) corresponds to a path \(P_e\) in \(D'\).
		To construct \(\flat\) consider \(\tilde{\flat_e} = (\flat'_{e'})_{e' \in E(P_e)}\).
		First we remove all vertices in \(N^{3\tw}(N^3[\chi(t')]) \setminus N^{3\tw}(N^3[\chi(t)])\) from the sets in \(\tilde{\flat_e}\), i.e.\ we consider \(\bar{\flat_e} = (z \cap N^{3\tw(G)}(N^3[\chi(t)]))_{z \in \tilde{\flat_e}}\).
		To finally obtain \(\flat_e\) from \(\bar{\flat_e}\), we remove occurrences of the empty set.\\
		Similarly to obtain \(\lambda_{e,i}\) for each \(e \in E(D)\) and \(i \in [|\flat_e| + 1]\), consider the sequence \(\tilde{\lambda_e} = (\lambda'_{e',1}, \dotsc, \lambda'{e',|\flat'_{e'}| + 1})_{e' \in E(P_e)}\).
		First we remove \(N^3[\chi(t')] \setminus N^3[\chi(t)]\) from each entry in \(\tilde{\lambda_e}\), i.e.\ we consider \(\bar{\lambda_e} = (z \cap N^3[\chi(t)])_{z \in \tilde{\lambda_e}}\).
		Then we obtain \(\lambda_{e}\) from \(\bar{\lambda_e}\) by replacing multiple consecutive repetitions of the same vertex set by a single entry of that vertex set.
		If \(|\lambda_{e}| \neq |\flat_e| + 1\) we discard the record and canonical branch decomposition currently being constructed and move on to the next branch.
		Otherwise we set \(\lambda_{e,i}\) to be the \(i\)-th entry of \(\lambda_{e}\).\\
		For every \(e \in E(D)\) and \(i \in [|\flat_e| + 1]\) consider the concatenation \(\tilde{\sigma}\) of all \(\sigma'_{e',j}\) for which \(\lambda'_j\) were replaced by \(\lambda_i\) when obtaining \(\Lambda_e\) from \(\bar{\Lambda}\) (possibly \(\lambda_i\) corresponds to a single \(\lambda'_j\)).
		To obtain \(\sigma_{e,i}\) from \(\tilde{\sigma}\), we then apply the typical sequence operator to \(\tilde{\sigma}\).
		
		\paragraph*{Join nodes.}
		Consider a join node \(t \in V(T)\) with children \(t_1,t_2 \in V(T)\) for which \(\mathcal{R}(t_1)\) and \(\mathcal{R}(t_2)\) were already computed.
		We define \(\mathcal{R}(t)\) to consist of the following pairs of records and branch decompositions of \(G_t\).
		For each record \((D^1, \flat^1, \Lambda^1, \sigma^1, \alpha^{{\imatch}^1}, \alpha^{{\iantimatch}^1}, \alpha^{{\ichain}^1})\) in \(\mathcal{R}(t^1)\) with corresponding branch decomposition \(\decT^1\) of \(G_{t_1}\), and \((D^2, \flat^2, \Lambda^2, \sigma^2, \alpha^{{\imatch}^2}, \alpha^{{\iantimatch}^2}, \alpha^{{\ichain}^2})\) in \(\mathcal{R}(t^2)\) with corresponding branch decomposition \(\decT^2\) of \(G_{t_2}\), where \(D^1 = D^2\), \(\flat^1 = \flat^2\) and \(\Lambda^1 = \Lambda^2\) we proceed as follows:\\
		We let \(D = D^1 = D^2\), \(\flat = \flat^1 = \flat^2\) and \(\Lambda = \Lambda^1 = \Lambda^2\), \(\alpha^{\iantimatch} = \max\{\alpha^{{\iantimatch}^1}, \alpha^{{\iantimatch}^2}\}\) and \(\alpha^{\ichain} = \max\{\alpha^{{\ichain}^1}, \alpha^{{\ichain}^2}\}\).
 		For each choice of an interleaving \({\tilde{\sigma}}_{e,i} \in \sigma^1_{e,i} \oplus \sigma^1_{e,i}\), we define \(\sigma_{e,i} = \tilde{\sigma}_{e,i} - \frac{|\lambda_{e,i}|}{2}\).
		With this combination of \(D, \flat, \Lambda, \sigma, \alpha^{\iantimatch}\) and \(\alpha^{\ichain}\) we add \((D, \flat, \Lambda, \sigma, \alpha^{\imatch}, \alpha^{\iantimatch}, \alpha^{\ichain})\) to \(\mathcal{R}(t)\) together with the branch decomposition \(\decT\) of \(G_t\) corresponding to the chosen interleaving, where
		\(\alpha^{\imatch}\) is set to the maximum over all entries of the constructed sequences \(\sigma_{e,i}\), \(\alpha^{{\imatch}^1}\) and \(\alpha^{{\imatch}^2}\).\\
		If any of entry of a sequence in \(\sigma\) or any of \(\alpha^{\imatch},\alpha^{\iantimatch},\alpha^{\ichain}\) exceeds \(\tw(G) + 1\) we discard the record and branch decomposition currently being constructed and move on to the next branch.
		
	To complete the description of the algorithm we detail the final step of the dynamic programming procedure which uses \(\mathcal{R}(r)\), where \(r\) is the root of \(T\), to output \(\Fsbw(G)\).
	In particular we deal with the implicit assumption underlying the dynamic program, that that in the records by which a certain record is reached each \(\sigma_{e,i}\) corresponds to a \(\lambda_i \in \Lambda_e\)-block which is consistent with the choices made for blocks at lower levels of \(T\).
	We do so by explicitly verifying the expected \(\Fsbw\) bound for the computed branch decompositions in \(\mathcal{R}(r)\).
	While it might be the case that during the dynamic program the implicit assumption is actually violated, this final step guarantees that this violation does not lead to a wrong output.
	
	\paragraph*{At the root of \(T\) -- the final step.}
	Consider the root \(r\) of \(T\), and \(\mathcal{R}(r)\) computed according to the procedures described above.
	Go through the elements of \(\mathcal{R}(r)\) in ascending order of \(\max\{\alpha^{\imatch},\alpha^{\iantimatch},\alpha^{\ichain}\}\) for the respective records.
	For each element, consider its branch decomposition and compute its \imatchwidth.
	If it is equal to \(\alpha^{\imatch}\) then we output \(\max\{\alpha^{\imatch},\alpha^{\iantimatch},\alpha^{\ichain}\}\) and the associated branch decomposition witnesses that \(\Fsbw(G) = \max\{\alpha^{\imatch},\alpha^{\iantimatch},\alpha^{\ichain}\}\).
	\subsection{Correctness}
	\label{sec:correct}
	
	In this section we argue correctness of the algorithm given in the previous subsection.
	
	First we show that the claimed \(\FFF^*\)-branchwidth given as the output of the algorithm actually can be achieved by a branch decomposition of \(G\), in particular this is true for the canonical branch decomposition \(\decT\) associated to the record \((D, \flat, \Lambda, \sigma, \alpha^{\imatch}, \alpha^{\iantimatch}, \alpha^{\ichain})\) that leads to the output.
	Because we ensure that \(\mimw(\decT) = \alpha^{\imatch}\), we obtain immediately that \(\mimw(\decT) \leq \max\{\alpha^{\imatch},\alpha^{\iantimatch},\alpha^{\ichain}\}\).
	Assume for contradiction that \(\Fsbw(\decT) > \max\{\alpha^{\imatch},\alpha^{\iantimatch},\alpha^{\ichain}\}\).
	Then \(\iantimatchwidth(\decT) > \max\{\alpha^{\imatch},\alpha^{\iantimatch},\alpha^{\ichain}\}\) or \(\ichainwidth(\decT) > \max\{\alpha^{\imatch},\alpha^{\iantimatch},\alpha^{\ichain}\}\).
	Let \(v \in V(G)\) be in a subgraph \(H\) of \(G\) which witnesses \(\iantimatchwidth(\decT) > \max\{\alpha^{\imatch},\alpha^{\iantimatch},\alpha^{\ichain}\}\) or \(\ichainwidth(\decT) > \max\{\alpha^{\imatch},\alpha^{\iantimatch},\alpha^{\ichain}\}\).
	Then because of the small diameter of antimatchings and chain graphs we have that \(V(H) \subseteq N^3[v]\).
	Consider \(t \in V(T)\) in which \(v\) is introduced, or the leaf node with \(\chi(t) = \{v\}\).
	By construction there must be a record \((D', \flat', \Lambda', \sigma', \alpha^{{\imatch}'}, \alpha^{{\iantimatch}'}, \alpha^{{\ichain}'})\) together with an associated branch decomposition \(\decT'\) such that \(\lcac{\decT}{V(G_t)} = \decT'\), and in particular \(\lcac{\decT}{N^3[v]}\) is a subtree of \(D'\).
	This means that \(\alpha^{{\iantimatch}'} \geq \iantimatchwidth(D') > \alpha^{\iantimatch}\) or \(\alpha^{{\ichain}'} \geq \ichainwidth(D') > \alpha^{\ichain}\).
	However, the iterative construction ensures that the entries for \(\alpha^{{\iantimatch}''}\) and  \(\alpha^{{\ichain}''}\) can not decrease for records \((D'', \Lambda'', \sigma'', \alpha^{{\imatch}''}, \alpha^{{\iantimatch}''}, \alpha^{{\ichain}''})\) encountered when constructing \((D, \flat, \Lambda, \sigma, \alpha^{\imatch}, \alpha^{\iantimatch}, \alpha^{\ichain})\) from \((D', \Lambda', \sigma', \alpha^{{\imatch}'}, \alpha^{{\iantimatch}'}, \alpha^{{\ichain}'})\).
	Thus \(\alpha^{\iantimatch} \geq \alpha^{{\iantimatch}'}\) and \(\alpha^{\ichain} \geq \alpha^{{\ichain}'}\), yielding a contradiction.
	
	Conversely, now we show that the \(\FFF^*\)-branchwidth returned by the algorithm is at most as large as the \(\FFF^*\)-branchwidth of any branch decomposition of \(G\), and in particular this implies that the algorithm always returns a solution (which is not immediately clear from the description of the root step).
	Assume for contradiction that there is a node \(t \in V(T)\) such that there is no record \((D, \flat, \Lambda, \sigma, \alpha^{\imatch}, \alpha^{\iantimatch}, \alpha^{\ichain})\) with canonical branch decomposition \(\decT\) of \(G_t\) in \(\mathcal{R}(t)\) such that the following two conditions hold.
	\begin{itemize}
		\item \(\decT\) can be extended to a branch decomposition \(\decT^*\) of \(G\) realizing the \(\FFF^*\)-branchwidth~of~\(G\).
		\item Every path in \(\decT^*\) that corresponds to an edge \(e \in E(D)\) consists of a sequence of blocks \(P_1, \dotsc, P_{|\flat_e| + 1}\), specifically \(P_i\) is a \(\lambda_{e,i}\)-block.
		Further each typical sequence of \((\mimw_{\lambda_{e,i}}(\decT,e))_{e \in E(P'_i)}\) where \(P'_i\) is the path in \(\decT\) which \(P_i\) corresponds to is equal~to~\(\sigma_{e,i}\).
	\end{itemize}
	Let \(t\) be such a node such that for all nodes \(t' \in V(T_t)\) there is a record \((D', \Lambda', \sigma', \alpha^{{\imatch}'}, \alpha^{{\iantimatch}'}, \alpha^{{\ichain}'})\) with corresponding branch decomposition \(\decT'\) of \(G_{t'}\) in \(\mathcal{R}(t')\) such that the following two conditions hold:
	\begin{itemize}
		\item \(\decT'\) can be extended to the same branch decomposition \(\decT^*\) of \(G\) realizing the \(\FFF^*\)-branchwidth of \(G\).
		\item Every path in \(\decT^*\) that corresponds to an edge \(e \in E(D')\) consists of a sequence of blocks \(P_1, \dotsc, P_{|\flat'_e| + 1}\), specifically \(P_i\) is a \(\lambda'_{e,i}\)-block.
		Further each typical sequence of \((\mimw_{\lambda'_{e,i}}(\decT',e))_{e \in E(P'_i)}\) where \(P'_i\) is the path in \(\decT'\) which \(P_i\) corresponds to is equal~to~\(\sigma'_{e,i}\).
	\end{itemize}
	We distinguish the type of \(t\) and arrive at a contradiction in every case.
	\paragraph*{Leaf nodes.}
	Assume that \(t\) is a leaf node and let \(\chi(t) = \{v\}\).
	Let \(\decT^*\) be a branch decomposition that realizes optimum \(\FFF^*\)-branchwidth.
	Then there is a record constructed from \(\lcac{\decT^*}{N^3[v]}\) (then \(\decT = \lcac{\decT^*}{N^3[v]}\) and not discarded as \(\Fsbw(\lcac{\decT^*}{N^3[v]}) \leq \Fsbw(\decT^*) = \Fsbw(G) \leq \tw(G) + 1\) by Corollary~\ref{cor:FisTWBounded}.
	Moreover it is extensively branched for each edge of \(D\) on the sequence of its blocks.
	As each block \(P_i\) corresponds to a single edge \(e\) of \(\decT\), the typical sequence of \((\mimw_{\lambda_i}(\decT,e))_{e \in E(P'_i)}\) where \(P'_i\) is the path in \(\decT\) which \(P_i\) corresponds to, is simply \((\mimw_{\lambda_{e,i}}(\decT,e))\).
	By the construction for leaf nodes this is equal to \(\sigma_{e,i}\).\\
	All together this is a contradiction to the choice of \(t\).
	
	\paragraph*{Introduce nodes.}
	Assume that \(t\) is an introduce node with child \(t'\) and let \(\chi(t) \setminus \chi(t') = \{v\}\).
	By choice of \(t\) there is a record \((D', \flat', \Lambda', \sigma', \alpha^{{\imatch}'}, \alpha^{{\iantimatch}'}, \alpha^{{\ichain}'})\) with canonical branch decomposition \(\decT'\) of \(G_{t'}\) in \(\mathcal{R}(t')\) such that the following two conditions hold:
	\begin{itemize}
		\item \(\decT'\) can be extended to a branch decomposition \(\decT^*\) of \(G\) realizing the \(\FFF^*\)-branchwidth of \(G\).
		\item Every path in \(\decT^*\) that corresponds to an edge \(e \in E(D')\) consists of a sequence of blocks \(P'_1, \dotsc, P'_{|\flat'_e| + 1}\), specifically \(P'_i\) is a \(\lambda'_{e,i}\)-block.
		Further each typical sequence of \((\mimw_{\lambda'_{e,i}}(\decT',e))_{e \in E(P'_i)}\) where \(P'_i\) is the path in \(\decT'\) which \(P'_i\) corresponds to is equal~to~\(\sigma'_{e,i}\).
	\end{itemize}
	Consider instead of \(\decT^*\) the branch decomposition \(\decT^{**}\) that arises by iteratively applying the construction described in the statement of Lemma~\ref{lem:typical} to a subsequence of \[(\mimw_{\lambda'_{e,i}}(\lcac{\decT^*}{N^3[V(G_{t'})]},e))_{e \in E(P'_i)}\] which is contracted to obtain the typical sequence of \((\mimw_{\lambda'_{e,i}}(\lcac{\decT^*}{N^3[V(G_{t'})]},e))_{e \in E(P'_i)}\) for each \(e \in E(D')\), \(i \in [|\flat'_e| + 1]\).
	Then, by the construction described in the statement of Lemma~\ref{lem:typical}, all vertices in \(V(G) \setminus N^3[V(G_{t'})]\) are contained in subtrees of \(\decT^{**}\) which are attached at subdivisions of edges of \(\lcac{\decT^*}{V(G_{t'})} = \lcac{\decT^{**}}{V(G_{t'})}\) with pointers to the typical sequences of one of \((\mimw_{\lambda'_{e,i}}(\lcac{\decT^*}{N^3[V(G_{t'})]},e))_{e \in E(P'_i)}\).
	In particular this is true for vertices in \(N^3[\chi(t)] \setminus N^3[\chi(t')]\) and vertices in \(N^{3\tw(G)}(N^3[\chi(t)]) \setminus N^{3\tw(G)}(N^3[\chi(t')])\).
	Moreover the construction maintains the fact that \(\lcac{\decT^{**}}{N^3[\chi(t')]} = \lcac{\decT^*}{N^3[\chi(t')]} = D'\), and by the definition of \(\lambda\)-blocks that every path in \(\decT^{**}\) that corresponds to an edge \(e \in E(D')\) consists of a sequence of blocks \(P'_1, \dotsc, P'_{|\flat_e| + 1}\), where \(P'_i\) is a \(\lambda'_{e,i}\)-block in \(\decT^{**}\).\\
	Because \(\decT^{**}\) is an extension of \(\lcac{\decT}{V(G_{t'})}\) which is in turn an extension of \(D'\), in the procedure for introduce nodes at some point \(D = \lcac{\decT^{**}}{N^3[\chi(t)]}\) is considered.\\
	Similarly a set of blocks for node \(t'\) can be refined to blocks for \(t\) in accordance with insertion points for all vertices in \(N^{3\tw(G)}(N^3[\chi(t)])\).
	By construction the dynamic program enumerates all refinements of blocks for \(t'\) to blocks for \(t\) that are consistent with any such combinations of insertion points at subdivisions of edges of \(\lcac{\decT^*}{V(G_{t'})} = \lcac{\decT^{**}}{V(G_{t'})}\) with pointers to the typical sequences of one of \((\mimw_{\lambda'_{e,i}}(\lcac{\decT^*}{N^3[V(G_{t'})]},e))_{e \in E(P'_i)}\).
	Hence at some point, additionally to \(D\), \(\flat\), \(\Lambda\), and \(\decT\) are considered such that \(\decT\) can be extended to \(\decT^{**}\) and every path in \(\decT^{**}\) that corresponds to an edge \(e \in E(D)\) consists of a sequence of blocks \(P_1, \dotsc, P_{|\flat_e| + 1}\) for which each \(P_i\) is a \(\lambda_{e,i}\)-block.\\
	Finally, in each of these blocks \(P_i\), the  typical sequence of \((\mimw_{\lambda_{i}}(\decT,e))_{e \in E(P_i)}\) is easily seen to be given by the typical sequence of \((\mimw_{\lambda'_{j}}(\decT',e') + \frac{|\lambda_{i} \setminus \lambda'_{j}|}{2})_{e' \in E(P'_j)}\) where \(P_i\) corresponds to a subpath of the block \(P'_j\) in \(\decT^{**}\) at node \(t'\) (unless \(\lambda_{e,i}\) is not separated at the cuts in the block in which case the typical sequence is simply \((\bot)\)).
	This is because by construction \(\lambda_{i} \setminus \lambda'_{j}\) consists only of endpoints of pairwise independent edges \(A\) that each have at least one endpoint in \(N^3[\chi(t)] \setminus N^3[\chi(t')]\) and are all independent of \(\lambda'_{j}\).
	By the properties of tree decompositions all such edges are also independent of edges of \(G\) which have at least one endpoint in \(V(G_{t'}) \setminus N^3[\chi(t')]\) making all edges of any matching considered to achieve \(\mimw_{\lambda'_{j}}(\decT',e')\) independent of \(A\).\\
	All together we have shown that at some point in the branching procedure we consider a record at node \(t\) corresponding to \(\decT^{**}\) in the desired way.
	This record is not discarded, as none of its entries surpass \(\tw(G) + 1\), because \(\Fsbw(\decT^{**})\leq \tw(G)\), contradicting the choice of \(t\).
	
	\paragraph*{Forget nodes.}
	Assume that \(t\) is a forget node with child \(t'\) and let \(\chi(t') \setminus \chi(t) = \{v\}\).
	By choice of \(t\) there is a record \((D', \flat', \Lambda', \sigma', \alpha^{{\imatch}'}, \alpha^{{\iantimatch}'}, \alpha^{{\ichain}'})\) with canonical branch decomposition \(\decT'\) of \(G_{t'}\) in \(\mathcal{R}(t')\) such that the following two conditions hold:
	\begin{itemize}
		\item \(\decT'\) can be extended to a branch decomposition \(\decT^*\) of \(G\) realizing the \(\FFF^*\)-branchwidth of \(G\).
		\item Every path in \(\decT^*\) that corresponds to an edge \(e \in E(D')\) consists of a sequence of blocks \(P'_1, \dotsc, P'_{|\flat'_e| + 1}\), specifically \(P'_i\) is a \(\lambda'_{e,i}\)-block.
		Further each typical sequence of \((\mimw_{\lambda'_{e,i}}(\decT',e))_{e \in E(P'_i)}\) where \(P'_i\) is the path in \(\decT'\) which \(P'_i\) corresponds to is equal to \(\sigma'_{e,i}\).
	\end{itemize}
	Then at some point we consider \(\decT = \decT'\) which can trivially be extended to \(\decT^*\) and set \(D = \lcac{D'}{N^3[\chi(t)]} = \lcac{\decT^*}{N^3[\chi(t)]}\).
	Moreover a \(\lambda\)-block \(P_i\) of a path \(P_e\) corresponding to \(e \in E(D)\) consists of a concatenation of \(\lambda'_{1}, \dotsc, \lambda'_{\ell}\)-blocks \(P'_1, \dotsc, P'_\ell\) of paths in \(\decT^*\) corresponding to edges \(e'_1, \dotsc, e'_r \in E(P')\) where \(P'\) is the path in \(D'\) corresponding to \(e\), and \(\lambda'_{1} \cap N^3[\chi(t)] = \dotsc, \lambda'_{\ell} \cap N^3[\chi(t)] = \lambda\).
	These blocks are consistent with the construction of \(\flat\) and \(\Lambda\) from \((D', \flat', \Lambda', \sigma', \alpha^{{\imatch}'}, \alpha^{{\iantimatch}'}, \alpha^{{\ichain}'})\) in the dynamic programming procedure for forget nodes.
	That is every path in \(\decT^*\) that corresponds to an edge \(e \in E(D)\) consists of a sequence of blocks \(P_1, \dotsc, P_{|\flat_e| + 1}\) with each \(P_i\) being a \(\lambda_{e,i}\)-block.\\
	Finally, in each of these \(\lambda\)-blocks \(P_i\), the  typical sequence of \((\mimw_{\lambda}(\decT,f))_{f \in E(P_i)}\) is easily seen to be given by the typical sequence of the concatenation of the typical sequences of \((\mimw_{\lambda'_{j}}(\decT',f))_{f \in E(P'_j)}\) where \(j \in [\ell]\), and \(\ell\), \(P'_1, \dotsc, P'_\ell\) and \(\lambda'_1, \dotsc, \lambda'_\ell\) are as above.
	To see this one can argue that for \(f \in E(P'_j) \subseteq E(P_i)\), \(\mimw_{\lambda}(\decT,f) = \mimw_{\lambda'_{j}}(\decT',f)\):\\
	Recall that \(\decT = \decT'\), \(P'_j\) is a subpath of a \(\lambda\)-block of \(P_e\) and at the same time a \(\lambda'_j\)-block of the path in \(\decT^*\), and \(\lambda = \lambda'_j \cap N^3[\chi(t)]\).
	We immediately get \(\mimw_{\lambda}(\decT,f) \geq \mimw_{\lambda'_{j}}(\decT',f)\).
	Assume for contradiction that there is some (possibly empty) \(\lambda'' \neq \lambda'_{j}\) such that \(\lambda'' \cap N^3[\chi(t)] = \lambda\) and \(\mimw_{\lambda''}(\decT,f) \geq \mimw_{\lambda'_{j}}(\decT,f)\).
	Let \(M''\) be a matching in the bipartite subgraph of \(G_t = G_{t'}\) induced by \(f\) in \(\decT = \decT'\) that witnesses this.
	At the same time, because \(P'_j\) is a \(\lambda'_j\)-block there is some matching \(M\) in the bipartite subgraph of \(G\) induced by \(f\) in \(\decT^*\) that witnesses that a matching whose edges in \(E(G[N^3[\chi(t')]]))\) are incident to exactly \(\lambda'_j\) achieves \(\mathcal{F}_{\imatch}\)-branchwidth at least as large as any matching whose edges in \(E(G[N^3[\chi(t')]]))\) are incident to exactly \(\lambda''\) at the edge \(f\) of \(\decT^*\).
	By the properties of tree decompositions, and the fact that the behavior of \(M''\) and \(M\) on \(E(G[N^3[\chi(t)]])\) is determined by \(\lambda\), we can construct the induced matching with edges \(\tilde{M} = E(M'') \cup (E(M) \setminus E(G_t))\) with cardinality greater than \(|E(M)|\), and \(E(\tilde{M}) \cap E(G[N^3[\chi(t)]]) = \lambda'' \neq \lambda'_j\).
	This is a contradiction to the fact that \(P'_j\) is a \(\lambda'_j\)-block.\\
	All together we have shown that at some point in the branching procedure we consider a record at node \(t\) corresponding to \(\decT^{*}\) in the desired way, contradicting the choice of \(t\).
	
	\paragraph*{Join nodes.}
	Assume that \(t\) is a join node with children \(t_1\) and \(t_2\).
	By choice of \(t\) for \(j \in \{1,2\}\) there is a record \((D^j, \flat^j, \Lambda^j, \sigma^j, \alpha^{{\imatch}^j}, \alpha^{{\iantimatch}^j}, \alpha^{{\ichain}^j})\) with canonical branch decomposition \(\decT^j\) of \(G_{t_j}\) in \(\mathcal{R}(t_j)\) such that the following two conditions hold:
	\begin{itemize}
		\item \(\decT^j\) can be extended to a branch decomposition \({\decT^*}\) of \(G\) realizing the \(\FFF^*\)-branchwidth of \(G\).
		\item Every path in \({\decT^*}\) that corresponds to an edge \(e \in E(D^j)\) consists of a sequence of blocks \(P^j_1, \dotsc, P^j_{|\flat^j_e| + 1}\), specifically \(P^j_i\) is a \(\lambda^j_{e,i}\)-block.
		Further each typical sequence of \((\mimw_{\lambda^j_{e,i}}(\decT^j,e))_{e \in E(P^j_i)}\) where \(P^j_i\) is the path in \(\decT^j\) which \(P^j_i\) corresponds to is equal to \(\sigma^j_{e,i}\).
	\end{itemize}
	At some point in the dynamic programming procedure at join nodes we consider combinations of records at \(t_1\) and \(t_2\) such that \(D^1 = \lcac{B^*}{N^3[\chi(t_2)]} \lcac{B^*}{N^3[\chi(t_1)]} = D^2\) for which then also the blocks are the same by the definition of blocks and because \(N^{3\tw(G)}(N^3[\chi(t_1)]) = N^{3\tw(G)}(N^3[\chi(t_2)])\).
	In particular their number is the same for each path in \(\decT^*\) that corresponds to an edge in \(D\), and there is a combination of elements in \(\mathcal{R}(t_1)\) and \(\mathcal{R}(t_2)\) for which not only \(D^1 = D^2\) but also \(\flat^1 = \flat^2\) and \(\Lambda^1 = \Lambda^2\).
	As both the branch decomposition \(\decT^1\) of \(G_{t_1}\) as well as the branch decomposition \(\decT^2\) of \(G_{t_2}\) can be extended to the same branch decomposition \(\decT^*\) there must be a way to combine them to obtain a branch decomposition \(\decT\) of \(G_t\) that can be extended to \(\decT^*\).
	As the structure of \(\decT^*\) is uniquely determined on \(N^3[\chi(t)]\) by \(D\), and along each path of \(\decT^*\) that corresponds to an edge of \(D\) the rough structure on \(N^{3\tw(G)}(N^3[\chi(t)])\) is given by the blocks determined by \(\flat\) and \(\lambda\).\\
	Hence what is left open for the construction of \(\decT\) is to attach the subtrees attached at any path of \(\decT^*\) that corresponds to an edge of \(D\) whose internal structure is already fixed in \(\decT^1\) and \(\decT^2\) respectively.
	Interleavings of the typical sequences correspond to possibilities to do this;
	e.g.\ repeating an entry in the extension of a typical sequence \(\sigma^1_{e,i}\) corresponds to postponing the occurrence of the edge of \(\decT^1\) pointing to that entry while moving on to the next entry of \(\sigma^1_{e,i}\) corresponds to fixing the relative position of the edge of \(\decT^1\) pointing to it also with respect to the edges of \(\decT^2\).
	(This use of \(\oplus\) is similar in other applications of typical sequences.)
	Note that in this way edges that contribute to the \(\imatchwidth_\lambda\) at edges of \(\decT^1\) and edges of \(\decT\) are in some sense overlayed.
	In this way, at least one interleaving has to lead to the construction of \(\decT = \lcac{\decT^*}{V(G_t)}\) and the branch in which this happens is not discarded, as \(\Fsbw(\lcac{\decT^*}{V(G_t)}) \leq \tw + 1\).
	Whenever they are not edges in \(E(G[N^3(\chi(t))])\) these edges are pairwise independent and they are also completely fixed on \(G[N^3[\chi(t)]]\) by \(\lambda\) which they are also independent of.
	Hence the \(\imatchwidth_\lambda\) values behave additively according to the interleaving, apart from the fact that \(\frac{|\lambda|}{2}\) edges are double-counted.
	This is taken into account in the dynamic programming step for join nodes.\\
	All together we have shown that at some point in the branching procedure we consider a record at node \(t\) corresponding to \(\decT^{*}\) in the desired way, contradicting the choice of \(t\).
	
\fi	
	\section{Treedepth}
	\label{sec:td}
	In this section we give a fixed parameter tractable algorithm for $\mathcal{F}^*$-\textsc{Branchwidth} parameterized by the treedepth of the input graph,
	when $\mathcal{F}^*$ is the union of a non-empty set of some \primal \phsymb families. We begin by outlining the high-level idea.

	If a graph has small treedepth but sufficiently many vertices, 
	then one can always find a small vertex set $R$ such that $G-R$ has many components each of them has bounded size.
	The main step is to argue that in this case, we can safely remove (``prune'') one of the components without changing its $\mathcal{F}^*$-branchwidth.
	Proposition~\ref{prop:manycomponent} captures the idea of this procedure, and we inductively use this result from bottom to top in the treedepth decomposition. 
	At the end, we reduce the given graph to a graph whose number of vertices is bounded by a function of the treedepth of the given graph, i.e., a (non-polynomial) kernel. The problem can then be solved on such a kernel using an arbitrary brute-force algorithm.

	\manycomponentProposition

	We explain the idea to prove Proposition~\ref{prop:manycomponent}. Suppose we have $G$, $R$, $H$, and the family $\{G_i\mid i\in [m]\}$ as in the statement.
	Let $(\decT, \decf)$ be a branch decomposition of $G-V(G_m)$ of optimal width, and 
	we want to find a branch decomposition of $G$ having same width.
	Let $V(H)=\{h_1, h_2, \ldots, h_p\}$.
	First we take nodes corresponding to $R$ and its least common ancestors. Let $S$ be the union of all least common ancestors in the tree.
	We take a large subset $I_1$ of $[m]$ satisfying that 
	for all $a\in [p]$ and $i_1, i_2\in I_1$, $\phi_{i_1}(h_a)$ and $\phi_{i_2}(h_a)$ are contained in the same component of $\decT-S$.

	As a second step, we show that 
	there are a large subset $I_2\subseteq I_1$, a set $F$ of at most $\ell-1$ edges in $\decT$ for some $1\le \ell\le p$, 
	 a partition $(A_1, \ldots, A_{\ell})$ of $[p]$, and a bijection $\mu$ from $[\ell]$ to the set of connected components of $\decT-F$ such that 
	for all $j\in [\ell]$ and $i_1, i_2\in I_2$, the minimal subtree of $\decT$ containing all nodes in $\{\decf^{-1}(\phi_{i_1}(h_b)) \mid  b\in A_j\}$
 and the minimal subtree of $\decT$ containing all nodes in $\{\decf^{-1}(\phi_{i_2}(h_b)) \mid  b\in A_j\}$
 are vertex-disjoint.
 	In fact, by the choice of $S$ and $I_2$, there should be a component of $\decT-S-F$ containing all nodes corresponding to $\{h_b\mid b\in A_j\}$.
 	Then in each component of $\decT-S-F$ corresponding to $A_j$, 
	we can find an edge separating  the family of minimal subtrees corresponding to $A_j$ in a balanced way, 
	and expand that edge to some subtree corresponding to vertices of $\{\phi_m(h_b) \mid   b\in A_j\}$.
	When we expand the edge, we take the restricted tree of $\decT$ with respect to some $\{\decf^{-1}(\phi_\alpha(h_b)) \mid   b\in A_j\}$ which is closest to the edge, and make a copy and identify with the edge in a natural way.
 	Then we can prove that every new cut has small width as well.

	\iflong
	We will use the following lemma.
	\fi
	
	\iflong
	\begin{lemma}
	\label{lem:treepartition}
	There is a function~${f \colon \mathbb{N}\times \mathbb{N} \to \mathbb{N}}$ satisfying the following property.
	Let $k$, $m$, $p\in \Nat$.
	Let $T$ be a subcubic tree and for each $i\in [m]$, let $V_i=\{v_{i,j}\mid j\in [k]\}$ be a set of $k$ distinct vertices in $T$ such that 
	$V_a\cap V_b=\emptyset$ for all distinct $a, b\in [m]$.
	
	If $m\ge f(k, p)$, then 
	there are a subset $I$ of $[m]$ of size at least $p$, a set $F$ of $\ell-1$ edges in $T$ for some $1\le \ell \le k$, a partition $(A_1, \ldots, A_{\ell})$ of $[k]$, and a bijection $\mu$ from $[\ell]$ to the set of connected components of $T-F$ such that 
	for all $j\in [\ell]$ and $i_1, i_2\in I$, the minimal subtree of $T$ containing vertices of $\{v_{i_1,b}\mid b\in A_j\}$ and 
	the minimal subtree of $T$ containing vertices of $\{v_{i_2,b}\mid b\in A_j\}$ are vertex-disjoint.
	\end{lemma}
	\fi
	\iflong
	\begin{proof}
    	For each $i\in [m]$, let $Q_i$ be the minimal subtree of $T$ containing all vertices of $V_i$, 
	and for $J\subseteq [k]$, we denote by $Q_i (J)$ the minimal subtree of $T$ containing all vertices of $\{v_{i,j}\mid j\in J\}$.
	
	For all positive integers~$k$ and $p$, we define $f^*(k,p, k):= p$, and 
for~${1\le \ell< k}$, we recursively define
       $f^*(k, p, \ell)=(3\cdot 2^k \cdot f^*(k, p, \ell+1)^4)^{\ell} p$.
       We define $f(k,p):=f^*(k, p, \ell)^2$.
       
    	We claim the following statement.
	\begin{itemize}
		\item Let $F$ be a set of $\ell-1$ edges in $T$, let $(A_1, \ldots, A_{\ell})$ be a partition of $[k]$, let $I_0\subseteq [m]$, and let $\mu$ be a bijection from $[\ell]$ to the set of connected components of $T-F$ satisfying that 
		\begin{itemize}
			\item for all $j\in [\ell]$ and $i\in I_0$, $Q_i(A_j)$ is fully contained in $\mu(j)$,
			\item for each $j\in [\ell]$, either 
			\begin{itemize}
				\item the subtrees in $\{Q_i(A_j)\mid  i\in I_0\}$ are pairwise vertex-disjoint, or
				\item the subtrees in $\{Q_i(A_j)\mid  i\in I_0\}$ share a node.
			\end{itemize}
		\end{itemize}
	If $|I_0|\ge f^*(k,p, \ell)$, then 
	there are
	a subset $I$ of $I_0$ of size at least $p$, 
	a set $F$ of at most $\ell^*-1$ edges in $T$ for some $0\le \ell^* \le k-1$, a partition $(B_1, \ldots, B_{\ell^*})$ of $[k]$, and a bijection $\mu^*$ from $[\ell^*]$ to the set of connected components of $T-F$ such that 
	 for each $j\in [\ell^*]$, the subtrees in $\{Q_i(B_j)\mid  i\in I\}$ are pairwise vertex-disjoint.
	\end{itemize}
	We prove it by induction on $k-\ell\in \{0, 1, \ldots, k-1\}$. If $k-\ell=0$, then each $A_j$ consists of a single element, and the statement is trivial. 
	We assume that $k-\ell>0$.

	For each $j\in [\ell]$, let $H_j$ be the intersection graph of the subtrees in $\{Q_i(A_j)\mid  i\in I_0\}$.
	By the construction, each $H_j$ is a chordal graph on at least $f^*(k,p,\ell)$ vertices, and it is perfect.
	A well-known theorem of Lov\'asz~\cite{Lovasz1972} says that for every perfect graph $G$ on $n$ vertices, 
	we have $n\le \alpha(G)\omega(G)$, where $\alpha(G)$ is the maximum size of an independent set and $\omega(G)$ is the maximum size of a clique.  
	We divide into two cases:
	\begin{itemize}
		\item For some $j\in [\ell]$, $H_j$ has a clique of size at least $3\cdot 2^k \cdot f^*(k,p,\ell+1)^4$.
		\item For all $j\in [\ell]$ $H_j$ has no clique of size $3\cdot 2^k \cdot f^*(k,p,\ell+1)^4$.
	\end{itemize}

	First assume that $H_j$ has a clique of size at least $3\cdot 2^k \cdot f^*(k, p, \ell+1)^4$.
	By the Helly property of subtrees in a tree, the subtrees corresponding to a clique share a common node.
	Let $x$ be a node of $T$ where all subtrees in $\{Q_i(A_j)\mid i\in I_0\}$ contain $x$. 
	As $T$ is subcubic, 
	there exist a neighbor $y$ of $x$ in $T$ and a subset $I_1$ of $I_0$ with $| I_1 |\ge 2^k \cdot f^*(k, p, \ell+1)^4$ such that 
	each subtree in  $\{Q_i(A_j)\mid i\in I_1\}$ contains both $x$ and $y$.
	Let $T^x$ and $T^y$ be the connected components of $T-xy$ where $x\in V(T^x)$ and $y\in V(T^y)$.
	Next, we choose a subset $I_2$ of $I_1$ with $| I_2 |\ge | I_1 |/2^k\ge f^*(k, p, \ell+1)^4$ such that for all $i_1, i_2\in I_2$ and $b\in A_j$, 
	\begin{itemize}
		\item $v_{i_1, b}\in V(T^x)$ if and only if $v_{i_2, b}\in V(T^x)$.
	\end{itemize}
	Let $(A_{j,1}, A_{j,2})$ be the partition of $A_j$ where for $i\in I_2$ and $b\in A_j$, 
	$v_{i,b}\in V(T^x)$ if and only if $b\in A_{j, 1}$.
	This means that for each $i\in I_2$, $Q_i(A_{j,1})$ is contained in $T^x$ 
	and $Q_i(A_{j,2})$ is contained in $T^y$.
	Lastly, by applying Lov\'asz's result twice, 
	we obtain a subset $I_3$ of $I_2$ with $ | I_3 | \ge f^*(k, p, \ell+1)$ such that for each $b\in [2]$,
	\begin{itemize}
		\item the subtrees in the family $\{Q_i(A_{j, b})\mid  i\in I_3\}$ are pairwise vertex-disjoint, or
		\item the subtrees in the family $\{Q_i(A_{j, b})\mid  i\in I_3\}$ share a common node.
	\end{itemize}
	
	Observe that $F\cup \{xy\}$ is a set of $\ell$ edges in $T$ and the partition obtained from $(A_1, \ldots, A_{\ell})$ by replacing $A_j$ with $A_{j,1}, A_{j,2}$  and the family $\{Q_i\mid  i\in I_3\}$ satisfy the precondition with $\ell+1$. 
	Thus, by the induction hypothesis, we obtain the result.

	Now, we assume that each $H_j$ has no clique of size $3\cdot 2^k \cdot f^*(k, p, \ell+1)^4$.
	Because $f^*(k, p, \ell)=(3\cdot 2^k \cdot f^*(k, p, \ell+1)^4)^{\ell} p$, 
	by recursively applying the  Lov\'asz's result~\cite{Lovasz1972} to subtrees in $\{Q_i(A_j)\mid  i\in I_0\}$ for each $j\in [\ell]$, 
	we obtain that there exists $I_4\subseteq I_0$ of size at least $p$ such that 
	 for each $j\in [\ell]$, the subtrees in $\{Q_i(A_j)\mid  i\in I_4\}$ are pairwise vertex-disjoint, and we are done.

	Recall that $f(k,p):=f^*(k,p,1)^2$. By applying the Lov\'asz's result to the given family $\{Q_i\mid i\in [m]\}$, 
	we obtain a subset $I_0\subseteq [m]$ of size at least $f^*(k,p,1)$ such that 
	\begin{itemize}
		\item the subtrees in the family $\{Q_i\mid  i\in I_0\}$ are pairwise vertex-disjoint, or
		\item the subtrees in the family $\{Q_i\mid  i\in I_0\}$ share a common node.
	\end{itemize}
	As the family $\{Q_i\mid  i\in I_0\}$ satisfies the preconditions of the claim with $(F, \ell, (A_1, \ldots, A_{\ell})) =(\emptyset, 1, ([k]))$ and $\mu$ the function from $[1]$ to $\{T\}$, we obtain the result.
	\end{proof}
	\fi

\iflong

The following is well known.
	\begin{lemma}
	\label{lem:balanced}
	Let $m$ be a positive integer.
	Let $T$ be a subcubic tree and let $\{T_i\mid i\in [m]\}$ be a family of pairwise vertex-disjoint subtrees in $T$.
	Then $T$ contains an edge $e$ such that each component of $T-e$ fully contains at least $m/3-1$ trees in $\{T_i\mid i\in [m]\}$.
	\end{lemma}
\fi

\iflong
	Next, we say that a pair $(T, A=(a_1, a_2, \ldots, a_k))$ of a tree $T$ and an ordered set $A$ of $k$ vertices is a \emph{$k$-rooted tree}.
	Two $k$-rooted trees $(T_1, (a_1, a_2, \ldots, a_k))$ and $(T_2, (b_1, b_2, \ldots, b_k))$ are \emph{isomorphic}
	if there is a graph isomorphism $\phi$ from $T_1$ to $T_2$ where $\phi(a_i)=b_i$ for all $i\in [k]$.	
	For a tree $T$ and an ordered set $A=(a_1, a_2, \ldots, a_k)$ of distinct vertices in $T$, 
	we denote by $rt(T; a_1, a_2, \ldots, a_k)$ or $rt(T; A)$ the $k$-rooted tree $(F, A)$ where $F = \lcac{T}{A}$.
	We call it the \emph{reduced tree of $T$ with respect to $A$}.
	\fi
	\iflong
	\begin{proof}[Proof of Proposition~\ref{prop:manycomponent}]
	Let $f$ be the function defined in Lemma~\ref{lem:treepartition}.
	For all positive integers $t$ and $p$, we define that 
	\begin{itemize}
		\item $g_4(t,p):=4t+4p+1$,
		\item $g_3(t,p):=2g_4(t,p)+2$,
		\item $g_2(t, p):=\max (g_3(t,p)(6t-1)+6t , 3g_4(t,p)+3)$,
		\item $g_1(t, p):=f( p, g_2(t,p) )$,
		\item $g(t,p):= (6t+1)^p g_1(t,p)+1$.
	\end{itemize}
	Suppose that $m\ge g(t,p)$ and $\tw(G)\le p$.
	
	It is sufficient to show that $\Fsbw(G)\le \Fsbw(G-V(G_m))$. 
	Let $(\decT, \decf)$ be a branch decomposition of $G-V(G_m)$ that has optimal width.
	From $(\decT, \decf)$, we construct a branch decomposition $(\decT^*, \decf^*)$ of $G$ having width at most $\Fsbw(G-V(G_m))$.
	
	We partition $\decT$ into edge-disjoint subtrees divided by nodes corresponding to vertices of $R$ and their least common ancestors.
	Let $S_1:=\decf^{-1}(R)$, $S_2:=N_{\decT}(S_1)$, and let $S_3$ be the set of least common ancestors of vertices in $S_2$.
	Finally, let $S:=S_1\cup S_2\cup S_3$. 
	Note that $|S_2|\le |S_1|=t$, and $|S_3|\le 2|S_2|\le 2t$.
	Let $\mathcal{C}$ be the set of maximal subtrees in $T$ that do not contain a vertex of $S$ as an internal vertex and do not contain a vertex of $S_1$.
	By the property of least common ancestors, 
	each subtree in $\mathcal{C}$ contains at most two vertices of $S_2\cup S_3$.
	Because a leaf is never taken as a common ancestor of vertices in $S_2$, 
	for every vertex $v\in V(G)\setminus R$, $\decf^{-1}(v)$ is contained in $\decT-S$.
	As each node of $S_2\cup S_3$ is contained in at most two subtrees in $\mathcal{C}$ that appear below, we obtain that $ |\mathcal{C} | \le 2|S_2\cup S_3|+1 \le 6t+1$ (one for the subtree containing the root, when the root is not in $S_2\cup S_3$).

	Since $m\ge g(t,p) = (6t+1)^{p}g_1(t,p)+1$,
	there exists $I_1\subseteq [m-1]$ of size $g_1(t,p)$ such that 
	\begin{itemize}
	\item for all integers $a,b\in I_1$ and $v\in V(H)$, 
	$\decf^{-1}(\phi_a(v))$ and $\decf^{-1}(\phi_b(v))$ are contained in the same subtree in $\mathcal{C}$.
	\end{itemize}
	Next, we apply Lemma~\ref{lem:treepartition} with $(k,m,p)=(p, |I_1|, y_2)$ and the family $\{\decf^{-1}(V(G_i))\mid i\in I_1\}$. 
	Since $|I_1|=g_1(t,p) = f(p, g_2(t,p))$, 
	there are a subset $I_2$ of $I_1$ of size  $g_2(t,p)$, 
	a set $F$ of $\ell-1$ edges in $\decT$ for some $1\le \ell \le p$, 
	a partition $(A_1, \ldots, A_{\ell})$ of $[p]$, and 
	a bijection $\mu$ from $[\ell]$ to the set of connected components of $\decT-F$ such that 
	for all $j\in [\ell]$ and $i_1, i_2\in I_2$, the minimal subtree of $\decT$ containing vertices of $\{\decf^{-1}(\phi_{i_1}(h_b))\mid b\in A_j\}$ and 
	the minimal subtree of $\decT$ containing vertices of $\{\decf^{-1}(\phi_{i_2}(h_b))\mid b\in A_j\}$ are vertex-disjoint.

	Let $\widehat{C}$ be the set of maximal subtrees $T$ of $\decT$ such that 
	\begin{itemize}
		\item $T$ does not contain a vertex of $S_1$,
		\item $T$ does not contain a vertex of $S_2\cup S_3$ as an internal vertex, and 
		\item $T$ does not contain an edge of $F$.
	\end{itemize}
	By the choice of $\mathcal{C}$ and $F$, 
	we know that there is an injection $\mu^*\colon [\ell] \to \widehat{C}$ such that
	\begin{itemize}
		\item for all $i\in I_2$, $j\in [\ell]$ and  $b\in A_j$, we have $\decf^{-1}(\phi_i(h_b))\in V(\mu^*(j))$, and 
		\item for all $j\in [\ell]$ and $i_1, i_2\in I_2$, the minimal subtree of $\decT$ containing vertices of $\{\decf^{-1}(\phi_{i_1}(h_b))\mid b\in A_j\}$ and 
	the minimal subtree of $\decT$ containing vertices of $\{\decf^{-1}(\phi_{i_2}(h_b))\mid b\in A_j\}$ are vertex-disjoint.
	\end{itemize}

	We discuss, for each $j\in [\ell]$, how to insert the nodes corresponding to vertices of $\{\phi_m(h_b)\mid b\in A_j \}$ by introducing some tree with $|A_j|$ leaves.
	We fix $j\in [\ell]$.
	Let $W$ be the subgraph obtained from $\mu^*(j)$ by recursively choosing a node $t$ such that $\mu^*(j)-t$ contains a component $W'$ having no node of $S_2\cup S_3$ and no node incident with an edge of $F$, and then removing $W'$.

	For each $i\in I_2$, let $Q_i$ be the minimal subtree of $\decT$ containing all nodes of $\{\decf^{-1}(\phi_i(h_b))\mid b\in A_j \}$.
	Clearly, $Q_i$ is contained in $\mu^*(j)$.

	Observe that every leaf of $W$ is either a node of $S_2\cup S_3$ or a node incident with an edge of $F$.
	It implies that $W$ has at most $3t$ leaves and has at most $6t$ nodes of degree 1 or 3.
	We consider two cases, ($\ast$) and ($\ast\ast$) below.

	\medskip
	($\ast$) We first assume that 	there exists a node $w$ in $W$ such that
	there is a component of $\mu^*(j)-w$ that contains no nodes of $S_2\cup S_3$, no nodes incident with an edge of $F$, and fully contains $Q_\alpha$ for some $\alpha\in I_2$. 
	
	By Lemma~\ref{lem:balanced}, there exists an edge $uv$ such that 
	each component of $\mu^*(j)-uv$ fully contains at least $| I_2 | /3 -1$ of subtrees in $\{Q_i\mid i\in I_2\}$. We assume that the distance from $u$ to $w$ in $\decT$ is at most the distance from $v$ to $w$. 
	We subdivide $uv$ into $uqv$.
	
	Let $Y$ be the reduced tree of $\decT$ with respect to $\{\decf^{-1}(\phi_\alpha(h_b))\mid b\in A_j \}\cup \{w\}$.
	We take a copy of $Y$, say $Y_1$, with isomorphism $\eta$ from $Y$ to $Y_1$.
	We add $Y_1$ to $\decT$, identify $q$ with $\eta(w)$, 
	and then extend $\decf$ by assigning, for each $v=\eta(\decf^{-1}(\phi_\alpha(h_b)))$ with $b\in A_j$, 
	$\decf(v):=\phi_m(h_b)$.
	
	\medskip
	($\ast\ast$) 
	Now, we assume that there are no such nodes $w$ in $W$ described in $(\ast)$. It implies that 
	every subtree in $\{Q_i\mid i\in I_2\}$ intersects $W$.
	Since $W$ has at most $6t$ nodes of degree 1 or 3, it has at most $6t-1$ maximal paths without degree 1 or 3 nodes.
	As $|I_2| =g_2(t,p) =g_3(t,p)(6t-1) + (6t)$, 
	there exists a path $P$ in $W$ and a subset $I_3$ of $I_2$ of size at least $g_3(t,p)$ such that 
	\begin{itemize}
		\item the endpoints of $P$ are nodes of degree 1 or 3 in $W$ and 
		all internal nodes of $P$ are nodes of degree 2 in $W$, and
		\item all subtrees of $\{Q_i\mid i\in I_3\}$ intersect $P$ and do not contain the endpoints of $P$. 
	\end{itemize}
	We choose an edge $uv$ of $P$ such that the number of paths in $\{P\cap Q_i\mid i\in I_3\}$ meeting components of $P-uv$ are differ by at most $1$.
	Let $\alpha\in I_3$ such that the distance from $uv$ to $P\cap Q_{\alpha}$ in $P$ is minimum.
	Assume that $u$ is closer to $P\cap Q_{\alpha}$ than $v$.
	
	Let $w_1, w_2$ be the two nodes in $Q_{\alpha}\cap P$ (possibly, $w_1=w_2$) such that they have a neighbor in $V(P)\setminus V(Q_{\alpha})$ and
	$w_1$ is closer to $u$ than $w_2$.
	Let $Y$ be the reduced tree of $\decT$ with respect to $\{\decf^{-1}(\phi_\alpha(h_b))\mid b\in A_j \}\cup \{w_1, w_2\}$.
	We take a copy of $Y$, say $Y_1$, with isomorphism $\eta$ from $Y$ to $Y_1$.
	We add $Y_1$ to $\decT$, subdivide $uv$ into $uu_1v_1v$, remove $u_1v_1$ and identify $u_1$ with $\eta(w_2)$ and identify $v_1$ with $\eta(w_1)$. 
	Also, we extend $\decf$ by assigning, for each $v=\eta(\decf^{-1}(\phi_\alpha(h_b)))$ with $b\in A_j$, 
	$\decf(v):=\phi_m(h_b)$.
	
	\medskip

	We apply the procedure $(\ast)$ or $(\ast\ast)$ for every $j\in [\ell]$. This is the end of the construction.
	Let $(\decT^*, \decf^*)$ be the resulting branch decomposition.
	We verify that it has width at most $\Fsbw(G-V(G_m))$.
	 
	 Let $e$ be an edge of $\decT^*$. 
	 We claim that $\Fsbw(\decT^*,e)\le \Fsbw(\decT)$.	 
	 Suppose for contradiction that $\Fsbw(\decT^*,e)> \Fsbw(\decT)$.
	 Let $\beta:=\Fsbw(\decT)$.
	 Note that since  
	 $\Fsbw(G-V(G_m))\le t$, and we have 
	 \[\Fsbw(\decT^*)\le t + |V(G_m)|\le t+p.\]
	
	Let $(A, B)$ be the vertex partition of $G$ separated by $e$.
	By the assumption that $\Fsbw(\decT^*,e)> \beta$, 
	$G[A, B]$ contains an induced subgraph $Z$ isomorphic to a graph in $\mathcal{F^*}$ that has $2\beta+2$ vertices.
	As $\beta=\Fsbw(\decT)$, $Z$ must contain a vertex of $G_m$.
	
	\setcounter{ourcase}{0}
		\begin{ourcase}[$e$ was an edge of $\decT$.]
	Let $(J_1, J_2)$ be a partition of $[p]$ such that 
	for each $i\in [p]$, $\phi_m(h_i)\in A$ if and only if $i\in J_1$.
	By the procedure $(\ast)$ and $(\ast\ast)$, we can observe that 
	if $e$ is contained in some $\mu^*(j)$ and $B$ contains the vertex set $\{\phi_m(h_b)\mid  b\in A_j\}$, 
	there exists a subset $I_4$ of $I_2$ with $| I_4 |\ge \min ( |I_2| /3 -1 , |I_3| /2 -1)$ such that 
	\begin{itemize}
		\item for all $a\in I_4$ and $b\in A_j$, we have $\phi_a(h_b)\in B$.
	\end{itemize}
	This implies that 
	there exists a subset $I_4$ of $I_2$ with $| I_4 |\ge \min ( |I_2| /3 -1 , |I_3| /2 -1)\ge g_4(t,p)$ such that 
	\begin{itemize}
		\item for all $a\in I_4$, $b\in J_1$, $c\in J_2$, we have $\phi_a(h_b)\in A$ and $\phi_a(h_c)\in B$.
	\end{itemize}
	As $| I_4 |\ge 2t+2p+1$ and $\Fsbw(G)\le t+p$, there exists $x\in I_4$ such that 
	$Z$ does not intersect $G_x$.
	It means that we can replace $G_m$ with $G_x$ to obtain an induced subgraph isomorphic to $Z$, 
	and it implies that $\Fsbw(\decT,e)\ge \beta+1$. This is a contradiction.
	\end{ourcase}

\begin{ourcase}[$e$ is a new edge in $\decT^*$]
	This means that $e$ was obtained when we apply the procedure $(\ast)$ or $(\ast\ast)$.
	We check for the two cases. We recall the terminologies from the procedures.
	
	\begin{itemize}
	\item (Case $(\ast)$. There exists a node $w$ in $W$ such that
	there is a component of $\mu^*(j)-w$ that contains no nodes of $S_2\cup S_3$, no nodes incident with an edge of $F$, and fully contains $Q_\alpha$ for some $\alpha\in I_2$.)
	
	In this case, one of $A$ and $B$, say $A$, is a subset of $V(G_m)$.
	Let $(J_1, J_2)$ be a partition of $[p]$ such that 
	for each $i\in [p]$, $\phi_m(h_i)\in A$ if and only if $i\in J_1$.
	In the minimal subtree of $\decT$ containing $\{\decf^{-1}(\phi_\alpha(h_b))\mid b\in A_j \}$, 
	there exists an edge $e^*$ separating $V(G)$ into $(A^*, B^*)$ where $\{\phi_{\alpha}(h_b)\mid b\in J_1\}\subseteq A^*$ and $\{\phi_{\alpha}(h_b)\mid b\in J_2\}\subseteq B^*$.

	By the construction, $B^*$ contains all vertices of $S$.
	From the choice of the node $w$ and $e^*$, there exists a subset $I_5$ of $I_2$ of size at least $|I_2| /3-1\ge 4t+4p$ such that $B^*$ fully contains components in $\{G_i\mid i\in I_5\}$.
	Since $\Fsbw(G)\le t+p$, by replacing $G_m$ with $G_x$, and replacing components of $\{G_i\mid i\in I_2\}$ used for $Z$ with some components in $\{G_i\mid i\in I_5\}$ that do not intersect $Z$, 
	we can obtain an induced subgraph isomorphic to $Z$ without using vertices of $G_m$. 
	It implies that $\Fsbw(\decT,e^*)\ge \beta+1$, a contradiction.
	
	\medskip
	\item (Case $(\ast\ast)$. There are no nodes $w$ in $W_2$ as in Case $(\ast)$.) 	
	
	If one of $A$ and $B$ is a subset of $V(G_m)$, then an analogous argument from Case $(\ast)$ holds, because $| I_3 |\ge g_3(t,p)\ge 4t+4p$.
	Thus, we may assume that it is not the case. 
	Let $a,b$ be the two endpoints of $P$, 
	without loss of generality, we may assume that $A$ corresponds to the subtree of $\decT^*-uv$ containing $a$.
	Let $(J_1, J_2)$ be a partition of $[p]$ such that 
	for each $i\in [p]$, $\phi_m(h_i)\in A$ if and only if $i\in J_1$.
	
	By the construction, in the minimal subtree of $\decT$ containing $\{\decf^{-1}(\phi_\alpha(h_b))\mid b\in A_j \}\cup \{w_1, w_2\}$, 
	there exists an edge $e^*$ separating $V(G)$ into $(A^*, B^*)$ where  
	where $\{\phi_{\alpha}(h_b)\mid b\in J_1\}\subseteq A^*$ and $\{\phi_{\alpha}(h_b)\mid b\in J_2\}\subseteq B^*$.

	By the choice of $S$ and $F$, 
	for all $j'\neq j$, $\{\decf^{-1}(\phi_\alpha(h_b))\mid b\in A_{j'} \}$ are fully contained in $A^*$, 
	or it is fully contained in $B^*$.
	Also, since $P$ intersects $g_3(t, p)$ paths in $\{P\cap Q_i\mid  i\in I_3\}$, 
	each of $A^*$ and $B^*$ contains at least $(g_3(t,p)-1) /2\ge 4t+4p$  paths in $\{P\cap Q_i\mid  i\in I_3\}$.
	Since $\Fsbw(G)\le t+p$, by replacing $G_m$ with $G_\alpha$, and replacing components of $\{G_i\mid i\in I_3\}$ used for $Z$ with some components in $\{G_i\mid i\in I_3\}$ that do not intersect $Z$, 
	we can obtain an induced subgraph isomorphic to $Z$ without using vertices of $G_m$. 
	It implies that $\Fsbw(\decT,e^*)\ge \beta+1$, a contradiction.
		\end{itemize}
	\end{ourcase}
	We conclude that $\Fsbw(\decT^*)\le \Fsbw(\decT)$.
	\end{proof}
			\fi

\iflong
	Let $T$ be a rooted tree with height $k$, and let $r$ be the root of $T$. For every node $t$ of $T$, we define the \emph{rank} of $t$ as $k-\dist_T(t, r)$.
	Note that the rank of a leaf $t$ with maximum $\dist_T(t, r)$ has rank $1$.
	\fi

\treedepthTheorem
\iflong
	\begin{proof}
	Let $G$ be a graph of treedepth $k$. By Lemma~\ref{lem:component}, we may assume that $G$ is connected. 
	Let $T$ be a rooted forest with height $k=\td(G)$ whose closure contains $G$ as a subgraph. As $G$ is connected, we may assume that $T$ is a tree. 
	Let $r$ denote the root of $T$.
		
	Let $g$ be the function defined in Proposition~\ref{prop:manycomponent}.
	We define $h(1):=1$ and 
	for~${1< j\le k}$, we recursively define
       $h(j):=2^{{h(j-1) \choose 2}} \cdot 2^{(k+j-1)h(j-1)} \cdot g(k+1, h(j-1))$.

	For each $j\in [k]$, we recursively obtain a graph $G^j$ and a rooted tree $T^j$ whose closure contains $G^j$ as a subgraph such that  
	\begin{itemize}
		\item $\Fsbw(G)=\Fsbw(G^j)$, 
		\item $T^j$ has height $k$, and
		\item for every $i\in [j]$ and every node $t$ in $T^j$ of rank $i$, 
		$G^j[V(G^j)\cap V((T^j)_t)]$ has at most $h(i)$ vertices.
	\end{itemize}
	We define $G^1:=G$ and $T^1:=T$. Clearly, $G^1$ and $T^1$ satisfy the above properties. 
	If this holds, then $G^k$ will have the size bounded by a function of $k$, and
	the problem can be solved on $G_k$ using an arbitrary brute-force algorithm.
	
	We assume that $1<j\le k$ and that $G^{j-1}$ and $T^{j-1}$ have been constructed, and we now want to construct $G^j$ and $T^j$ in polynomial time. 
	For each node $t$ of $T^{j-1}$, we denote $(G^{j-1})_t:=G^{j-1}[V(G^{j-1})\cap V((T^{j-1})_t)]$.
	We check whether $(G^{j-1})_t$ has at most $h(j)$ vertices for every node $t$ of rank $j$. 
	If this is true, then we can set $G^j:=G^{j-1}$, $T^j:=T^{j-1}$ and we are done.
	Thus, we may assume that there is a node $t$ of rank $j$ such that $(G^{j-1})_t$ has more than $h(j)$ vertices.
	Assume that $\{t_i \mid  i\in [x]\}$ is the set of all children of $t$ in $T^{j-1}$, and $R=\{r_i\mid  i\in [k+j-1]\}$ is the vertex set of the path from the root $r$ to $t$ in $T^{j-1}$.

	We consider a pair $(H, \gamma)$ of a graph $H$ on at most $h(j-1)$ vertices and a function $\gamma:V(H)\rightarrow 2^{V(R)}$. 
	Up to isomorphism, there are at most $2^{{h(j-1) \choose 2}} \cdot 2^{(k+j-1)h(j-1)}$ such pairs. 
	Intuitively, we use this pair to reflect the information how the corresponding vertex in $(G^{j-1})_{t_i}$ is adjacent to a vertex in $V(R)$.
	
	As each $(G^{j-1})_{t_i}$ has at most $h(j-1)$ vertices and $(G^{j-1})_t$ has more than $h(j)$ vertices, 
	there is a subset $\{p_i \mid  i\in [y]\}$ of $\{t_i\mid  i\in [x]\}$ and a pair $(H, \gamma)$ with $y=g(k+1, h(j-1))$ such that 
	\begin{itemize}
		\item for every $i\in [y]$, there is a graph isomorphism $\phi_i$ from $H$ to $(G^{j-1})_{p_i}$, 
		\item for every $i\in [y]$ and $z\in [k-j+1]$ and $v\in V(H)$, 
		$\phi_i(v)$ is adjacent to $r_z$ if and only if $r_z\in \gamma(v)$.
	\end{itemize}
	Note that by Corollary~\ref{cor:FisTWBounded}, we have $\Fsbw(G^{j-1})=\Fsbw(G)\le \tw(G)+1\le \td(G)+1\le k+1$.
	As $y=g(k+1, h(j-1))$, by Proposition~\ref{prop:manycomponent}, we know that $\Fsbw(G^{j-1})=\Fsbw(G^{j-1}-V((G^{j-1})_{p_y}))$.
	Thus, we can safely reduce the graph by removing $V((G^{j-1})_{p_y})$.
	After applying this procedure recursively, we will obtain a graph $G^j$ and a rooted tree $T^j$
	such that for  every node $t$ in $T^j$ of rank $j$, 
		$G^j[V(G^j)\cap V((T^j)_t)]$ has at most $h(j)$ vertices.
	
	This concludes the theorem.
	\end{proof}
	\fi
	
	\treedepthCorollary	
	
\section{Feedback Edge Set}\label{sec:fes}
In this section we give a linear kernel for \Fsbwprob parameterized by the size 
$k$ of a feedback edge set of the input graph,
when $\Fstar$ is the union of some primal \phsymb families.
\iflong
Since a minimum-size feedback edge set of a graph can easily be computed in polynomial time 
as the complement of a spanning forest of a graph, we assume throughout this section
that we are given a size-$k$ feedback edge set of the input graph.
\fi

If a graph $G$ has many vertices but a small feedback edge set,
then either it has many bridges (in form of `dangling trees') and isolated vertices,
or it has long paths of degree two vertices in $G$,
which we refer to as `unimportant paths'.
This simple observation paves the road to the kernel:
For the former, we can observe that bridges and isolated vertices 
are not important when it comes to computing \Fstar-branchwidth 
\iflong(see Lemma~\ref{lem:fes:bridge} for bridges, the argument for isolated vertices is trivial). \fi 
\ifshort($\star$). \fi 
For the latter, we show that we can always shrink any unimportant path to a 
constant-length subpath without changing the \Fstar-branchwidth of $G$.
It is not surprising that this gives a safe reduction rule when $\Fstar$\xspace
does not contain induced matchings.
For induced matchings however, proving safeness requires quite an amount of detail\iflong,
which comprises a large part of this section\fi.
\iflong This will be done via\fi\ifshort For this we show \fi Lemma~\ref{lem:fes:unimportant:placement} which states that
if $G$ has a long enough unimportant path $P$, 
then we can modify each branch decomposition of $G$ such that the vertices
of a long enough subpath of $P$ appear `consecutively' in it,
without increasing the \Fstar-branchwidth.
In addition to that, we only have to overcome one minor hurdle 
that concerns induced chains of value $2$\iflong (Corollary~\ref{cor:fes:neatly:breaks})\fi.
From there on it is not difficult to argue that 
contracting an edge on $P$ does not change the \Fstar-branchwidth either.

Let us begin. 
We first define unimportant paths,
and what it means for the vertices of a path to appear `consecutively'
in a branch decomposition via the notion of preservation.
\unimportantPathDefinition
\preservesDefinition

We now turn to the main technical lemma needed in the correctness proof of the 
reduction rule alluded to above. 
Note that it only concerns induced matchings.

\FesUnimportantPlacementLemma

\iflong
\begin{proof}
	Let $P = a_1 \ldots a_{8}$.
	For $j \in [3]$, let $a_{j, 1}'' \defeq a_{3(j-1)+1}$ 
	and $a_{j, 2}'' \defeq a_{3(j-1)+2}$.
	For an illustration of $P$ and this notation see Figure~\ref{fig:unimportant-path};
	note that $\{a_{1, 1}'' a_{1, 2}'', a_{2, 1}'' a_{2, 2}'', a_{3, 1}'' a_{3, 2}''\}$
	is an induced matching in $G$.
	For $j \in [3]$, let $Q_j$ be the path in $\decT$ from $\decf(a_{j, 1}'')$ to $\decf(a_{j, 2}'')$.
	Let $\calQ \defeq \{Q_j \mid j \in [3]\}$ and let $H$ be the intersection graph of $\calQ$ in $\decT$.
	Either $H$ is a triangle or it contains a pair of nonadjacent vertices.
	\begin{figure}
		\centering
		\includegraphics[width=.7\textwidth]{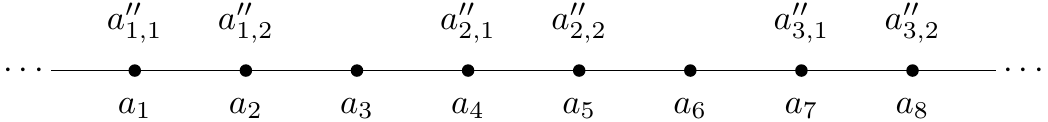}
		\caption{Illustration of the path $P$ and the notation used in Lemma~\ref{lem:fes:unimportant:placement}.}
		\label{fig:unimportant-path}
	\end{figure}
	
	\setcounter{ourcase}{0}
	\begin{ourcase}[$H$ is a triangle]
		Throughout the rest of this case,
		for all $j \in [3]$, $r \in [2]$ 
		we let $a_{j, r}' \defeq a_{j, r}''$.
		Since $H$ is a triangle, the paths $Q_{1}$, $Q_{2}$, and $Q_{3}$ intersect pairwise,
		and by the Helly property of trees 
		we know that they have a common nonempty intersection.
		Suppose they do not share a common edge, but only some vertex $x \in V(T)$.
		Since the endpoints of the three paths are (pairwise distinct) leaves of $\decT$,
		we know that $x$ is an internal vertex, so that each of the paths uses two edges incident with $x$.
		We may assume that $x$ has degree three,
		so by the pigeonhole principle
		each edge incident with $x$ is contained in two of the three paths.
		More precisely, the following holds.
		\begin{observation}\label{obs:clique:struc}
			Either $Q_{1}$, $Q_{2}$, and $Q_{3}$ share a common edge,
			or they share a common vertex $x \in V(\decT)$ such that for each pair of distinct $j, j' \in [3]$,
			$x$ has one incident edge that is contained in $Q_{j}$ and $Q_{j'}$.
		\end{observation}
		
		We choose a vertex $x_1 \in \bigcap_{j \in [3]} V(Q_{j})$ and one of its incident edges $e^\star$,
		according to which case of Observation~\ref{obs:clique:struc} applies.
		\begin{enumerate}
			\item\label{enum:common:edge}
			If the three paths share a common edge, then we let $e^\star$ be one of these edges
			and we choose the endpoint of $e^\star$ that is closer to $\decf(a_{1, 1}')$
			as $x_1$.
			\item\label{enum:common:vtx} 
			Otherwise, we let $x_1$ be the only vertex in the intersection of the three paths 
			and we let $e^\star$ be the edge incident with $x_1$ that is contained in $Q_{2}$ and $Q_{3}$.
		\end{enumerate}
		
		We denote the endpoint of $e^\star$ other than $x_1$ by $x_2$.
		For an illustration of the following construction, see Figure~\ref{fig:clique-sketch}.
		We obtain $(\decT', \decf')$ as follows.
		\begin{itemize}
			\item We remove the leaves to which a vertex on $P$ is mapped from $\decT$.
			If this creates additional leaves other than the endpoints of $e^\star$, then
			we iteratively remove them.
			\item We perform $\card{V(P)}$ many subdivisions on $e^\star$
			and attach a new leaf to each subdivision vertex.
			If after this step, an endpoint of $e^\star$ is a leaf, then we remove it.
			For simplicity, we assume that this did not happen.
			This finishes the construction of the tree $\decT'$.
			Let $\decT''$ be the caterpillar introduced this way between $x_1$ and $x_2$ in $\decT'$.
			\item
			We let $\decf'|_{V(G) \setminus V(P)} = \decf|_{V(G) \setminus V(P)}$,
			and we let $\decf'$ map the vertices in $P$ to the (newly created) leaves of $\decT''$
			in the order in which they appear on $P$,
			with $a_{1, 1}'$ being mapped to the leaf closest to $x_1$ 
			and $a_{3, 2}'$ being mapped to the leaf closest to $x_2$.
		\end{itemize}
		
		\begin{figure}
			\centering
			\includegraphics[width=.95\textwidth]{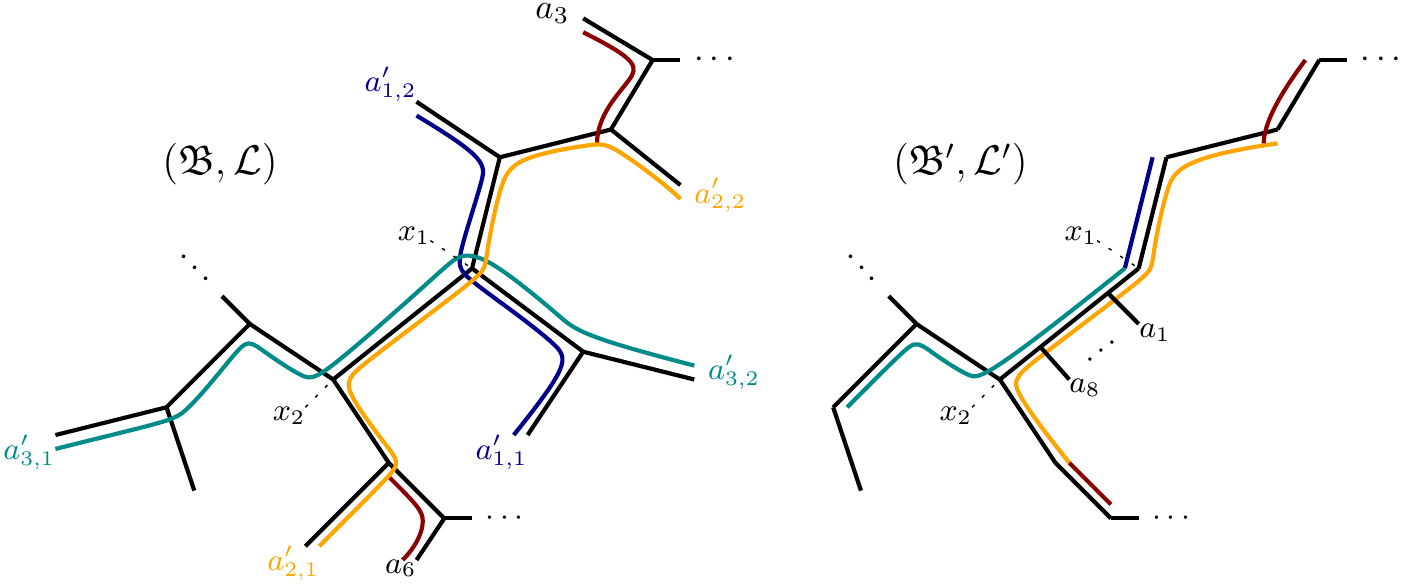}
			\caption{Sketch of the clique case. 
				Note that $a_{1, 1}' = a_1$, $a_{1, 2}' = a_2$, $a_{2, 1}' = a_4$, etc. 
			}
			\label{fig:clique-sketch}
		\end{figure}
		We show that $\mimw(\decT', \decf') \le \mimw(\decT, \decf)$ 
		by exhibiting for each induced matching $M'$ in each cut in $(\decT', \decf')$
		an induced matching $M$ in some cut of $(\decT, \decf)$ of the same size.
		We use the following convention.
		For a cut $(A', B')$ induced by some edge in $\decT'$,
		we call a cut $(A, B)$ induced by some edge in $\decT$ a 
		\emph{corresponding cut} 
		if $A' \setminus V(P) = A \setminus V(P)$ 
		and $B' \setminus V(P) = B \setminus V(P)$.
		Observe that for each cut $(A', B')$ there is always a corresponding cut $(A, B)$.
		
		Let $\decS$ be the minimal subtree of $\decT$ containing $\decf(a_i)$ for all $i \in [8]$
		(note in particular that $\decS$ contains $Q_1$, $Q_2$, and $Q_3$)
		and let $\decS'$ be the subtree of $\decT'$ obtained from $\decS$ by replacing the 
		leaves that are not in $\decT'$ and the edge $x_1x_2$ with the caterpillar between $x_1$ and $x_2$ in $\decT'$.
		Moreover, for all $j \in [3]$, 
		we let $Q_{j}'$ be the path in $S'$ that naturally corresponds to $Q_{j}$ in $\decS$.
		
		First, we observe that every cut induced by an edge in $E(\decT') \setminus E(\decS')$ 
		is also induced by some edge in $\decT$, 
		therefore we can restrict ourselves to the cuts induced by edges in $E(\decS')$.
		Let $\decS'' \defeq \decS' - V(\decT'')$,
		and observe that $\decS''$ has precisely two connected components -- 
		the one containing $x_1$, say $D_1$ and the one containing $x_2$, say $D_2$.
		
		Let $(A', B')$ be any cut induced by an edge in $D_2$.
		We may assume that $V(P) \subseteq B'$.
		We observe that for each $j \in [3]$, the path $Q_{j}'$ has at most one endpoint in $D_2$.
		(Recall that each such path uses the vertex $x_1 \notin V(D_2)$.)
		Let $R$ be the set of all vertices of $P$ 
		mapped to endpoints in the corresponding paths $Q_{j}$ by $\decf$.
		Then, there is a corresponding cut $(A, B)$ in $(\decT, \decf)$ such that
		$A = A' \cup R$ and $B = B' \setminus R$.
		Moreover, we have that 
		\begin{align*}
			E(G[A', B' \setminus V(P)]) = E(G[A \setminus V(P), B \setminus V(P)]).
		\end{align*}
		Observe that there is at most one edge in $G[A', B']$ that is not contained in $G[A, B]$:
		Since all vertices of $P$ are in $B'$, no edge of $P$ crosses the cut $(A', B')$.
		The vertices $a_{1, 1}'$ and $a_{3, 2}'$ have one more neighbor in $V(G) \setminus V(P)$ each;
		but since the cut we are considering is induced by an edge in $D_2$ we know that $a_{1, 1}$ remains in $B$.
		In the case that $a_{3, 2}' \in A$, it might happen that $a_{3, 2}' w \in E(G[A', B']) \setminus E(G[A, B])$,
		where $w$ is the neighbor of $a_{3, 2}'$ other than $a_{3, 1}'$.
		But this implies that $Q_{3}$ uses the edge that induces the cut $(A, B)$, in particular that
		$a_{3, 1}' a_{3, 2}' \in E(G[A, B])$.
		
		Now let $M'$ be an induced matching in $G[A', B']$.
		If $a_{3, 2}'w \notin M'$, then $M' \subseteq E(G[A, B])$ so $M'$ is an induced matching in $G[A, B]$ as well.
		If $a_{3, 2}'w \in M'$, then as we just argued, we can replace it with the edge $a_{3, 1}' a_{3, 2}'$ in $G[A, B]$
		and obtain an induced matching in $G[A, B]$ of the same size as $M'$.
		
		The cuts induced by edges in $D_1$ can be argued similarly, 
		with two exceptions.
		First, the path $Q_{1}$ may be fully contained in $D_1$ 
		(in situation~\ref{enum:common:vtx} described above).
		However, as $Q_{1}$ contains $x_1$, 
		the vertex $a_{1, 1}'$ only switches sides on the cuts induced by edges on
		the path between $x_1$ and the leaf to which $a_{1, 1}'$ is mapped,
		so it can be treated similarly.
		Second, there may be a cut 
		where both $a_{1, 1}'$ and $a_{3, 2}'$ switch sides. 
		But in this case, we can apply the same argument as before to both $a_{1, 1}'$ and $a_{3, 2}'$ at once.
		
		Before we move on to the edges in $\decT''$,
		note that there are two edges, 
		one between $x_1$ and a vertex from $\decT''$, 
		and one between $x_2$ and a vertex from $\decT''$,
		inducing cuts where all vertices of $P$ are on the same side of the cut.
		These can be treated similarly as cuts in the components $D_1$ and $D_2$ 
		(depending on which component they are incident with).
		
		It remains to consider any cut $(A', B')$ induced by an edge in $\decT''$,
		where $A'$ contains the vertices of $G$ 
		that are mapped to leaves in the connected component of $\decT' - V(\decT'')$ containing $x_1$, 
		and $B'$ contains the vertices that are mapped to the component containing $x_2$.
		Clearly, if the edge is incident with a leaf of $\decT''$ then the corresponding cut has a part with only one vertex,
		so we only have to consider edges that are not incident with any leaf.
		We may assume that there is some $j \in [7]$ such that 
		$A' \cap V(P) = \{a_1, \ldots, a_j\}$ and
		$B' \cap V(P) = \{a_{j+1}, \ldots, a_8\}$.
		Note that the edge $a_j a_{j+1}$ is the only edge on $P$ that is contained in $G[A', B']$.
		
		Consider the corresponding cut $(A, B)$ induced by the edge $x_1x_2 \in E(\decT)$ in $(\decT, \decf)$,
		where $A$ contains the vertices that are mapped to leaves in the component of $\decT - x_1 x_2$ containing $x_1$,
		and $B$ contains the vertices that are mapped to leaves in the component containing $x_2$.
		By construction, 
		we know that $a_{1, 1}' = a_1 \in A$, 
		and that the paths $Q_{2}$ and $Q_{3}$ use the edge $x_1x_2$.
		The latter implies that both $a_{2, 1}' a_{2, 2}'$ and $a_{3, 1}' a_{3, 2}'$ 
		are edges in $G[A, B]$.
		
		Let $M'$ be an induced matching in $G[A', B']$.
		Again, we can restrict our attention to the edges incident with vertices in $P$.
		If $a_{1, 1}' w \in M'$, where $w$ is the neighbor of $a_{1, 1}'$ other than $a_{1, 2}'$,
		then we observe that $a_{1, 1}' w \in E(G[A, B])$, since $a_{1, 1}' \in A$ and $w$ has not been moved.
		If $a_{3, 2}' u \in M'$, where $u$ is the neighbor of $a_{3, 2}'$ other than $a_{3, 1}'$,
		then we can replace $a_{3, 2}'u$ with the edge $a_{3, 1}' a_{3, 2}'$.
		Finally, if $a_j a_{j+1} \in M'$, then we can replace $a_j a_{j+1}$ with the edge $a_{2, 1}' a_{2, 2}'$.
		In all of these cases, the resulting set of edges is an induced matching in $G[A, B]$,
		clearly of the same size as $M'$.
		This finishes Case~1.
	\end{ourcase}
	
	\begin{ourcase}[$H$ has a pair of nonadjacent vertices]
		Let $Q_{h_1}$ and $Q_{h_2}$ be the paths corresponding to 
		the pair of nonadjacent vertices such that $h_1 < h_2$.
		For $j \in [2]$, $r \in [2]$, we use the shorthand notation $a_{j, r}' \defeq a_{h_j, r}''$.
		Since $Q_{h_1}$ and $Q_{h_2}$ do not intersect,
		we know that there is some edge $e^\star$ in $\decT$ such that $\decT - e^\star$ has two components $D_1$ and $D_2$
		such that $Q_{h_1}$ is fully contained in $D_1$ and $Q_{h_2}$ is fully contained in $D_2$.
		For $i \in [2]$, let $x_i$ denote the endpoint of $e^\star$ that is contained in $D_i$.
		
		Similarly to above, let $P^\star \subseteq P$ be the unimportant subpath from $a_{1, 1}'$ to $a_{2, 2}'$.
		we obtain a branch decomposition $(\decT', \decf')$ from $(\decT, \decf)$ as follows.
		\begin{itemize}
			\item We remove from $\decT$ all leaves to which a vertex from $P^\star$ is mapped via $\decf$.
			If this creates additional leaves other than the endpoints of $e^\star$, then
			we iteratively remove them.
			\item We perform $\card{V(P^\star)}$ subdivisions on the edge $e^\star$, 
			and attach a new leaf to each subdivision vertex.
			This finishes the construction of $\decT'$; 
			let $\decT''$ be the caterpillar introduced this way between $x_1$ and $x_2$.
			\item We let $\decf'|_{V(G) \setminus V(P^\star)} = \decf|_{V(G) \setminus V(P^\star)}$ and we let $\decf'$
			map the vertices of $P^\star$ to the (newly created) leaves of $\decT''$ in the order in which they appear on $P^\star$,
			with $a_{1, 1}'$ being mapped to the leaf closest to $x_1$ and $a_{2, 2}'$ being mapped to the leaf closest to $x_2$.
		\end{itemize}
		
		\begin{figure}
			\centering
			\includegraphics[width=.95\textwidth]{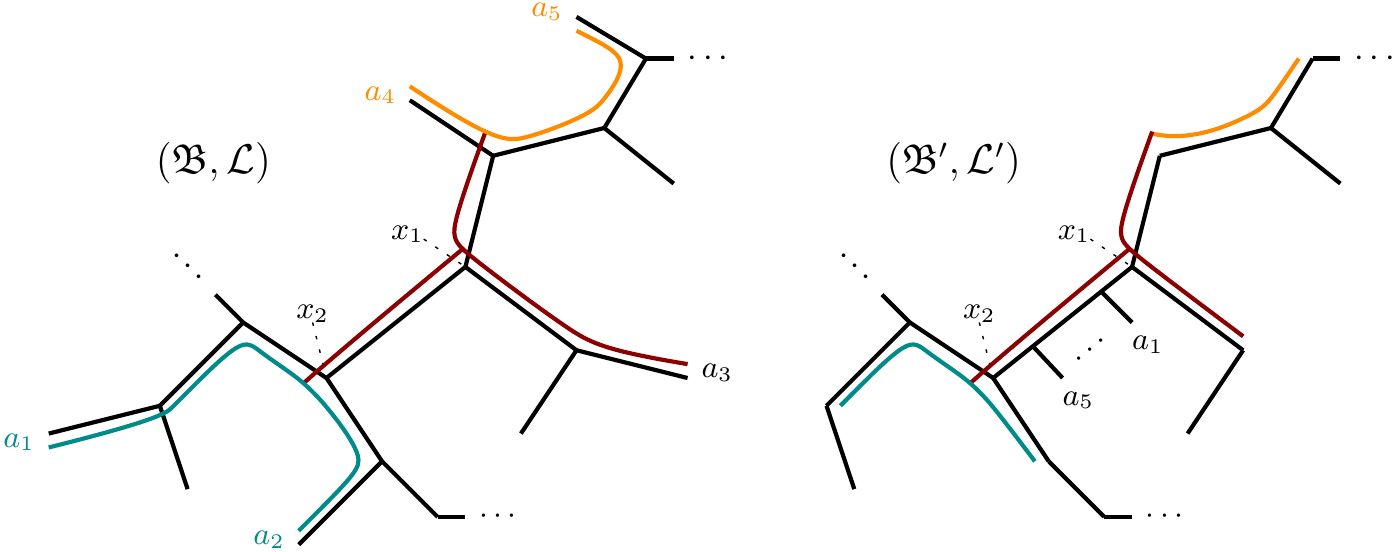}
			\caption{Sketch of the case when $H$ has a pair of nonadjacent vertices.
				We assume that $a_{1, 1}' = a_1$ and that $a_{2, 2}' = 5$. 
				The other cases are similar.}
			\label{fig:is-sketch}
		\end{figure}
		We show that $\mimw(\decT', \decf') \le \mimw(\decT, \decf)$;
		for an illustration of this case see Figure~\ref{fig:is-sketch}.
		Recall that for a cut
		$(A', B')$ induced by an edge in $\decT'$,
		we call a cut $(A, B)$ induced by an edge in $T$ a \emph{corresponding cut}
		if $A' \setminus V(P^\star) = A \setminus V(P^\star)$
		and $B' \setminus V(P^\star) = B \setminus V(P^\star)$.
		
		Let $\decS$ be the minimal subtree of $\decT$ containing $\decf(p)$ for all $p \in V(P^\star)$,
		and let $\decS'$ be the corresponding subtree of $\decT'$, including the caterpillar $\decT''$.
		For all $j \in [2]$, we denote by $Q_{h_j}'$ the path in $\decT'$ that naturally corresponds to $Q_{h_j}$.
		Each cut induced by an edge in $E(\decT') \setminus E(\decS')$ is also induced by some edge in $\decT$, 
		so we can focus on cuts induced by edges in $\decS'$.
		For $j \in [2]$, we denote by $D_j'$ the component of $\decS' - V(\decT'')$ containing $x_j$,
		and by $D_j$ the component of $\decS - x_1x_2$ containing $x_j$.
		
		We now consider cuts induced by edges in $D_1'$ and 
		we remark that the cuts induced by edges in $D_2'$ can be dealt with symmetrically.
		In each such cut $(A', B')$, we may assume that $V(P^\star) \subseteq B'$.
		We observe that there are at most two edges crossing the cut that are incident with a vertex in $P^\star$,
		namely one incident with $a_{1, 1}'$ and one incident with $a_{2, 2}'$.
		
		We first consider a cut $(A', B')$ induced by an edge in $Q_{h_1}'$,
		and let $(A, B)$ be the corresponding cut (induced by the copy of that edge in $Q_{h_1}$).
		We have that $a_{2, 2}' \in B$.
		By the choice of the cut, 
		we either have that $a_{1, 1}' \in A$ and $a_{1, 2}' \in B$
		or $a_{1, 1}' \in B$ and $a_{1, 2}' \in A$.
		If $a_{1, 1}' \in B$, then any induced matching in $G[A', B']$
		is also an induced matching in $G[A, B]$.
		Suppose $a_{1, 1}' \in A$ and consider any induced matching $M'$ in $G[A', B']$.
		If no edge of $M'$ is incident with $a_{1, 1}'$,
		then $M'$ is also an induced matching in $G[A, B]$;
		so we may assume that $a_{1, 1}' w \in M'$ for some $w \notin V(P^\star)$.
		But in this case, we have that $a_{1, 2}' \in B$, so $a_{1, 1}' a_{1, 2}' \in E(G[A, B])$,
		which means that $M \defeq M' \setminus \{a_{1, 1}' w\} \cup \{a_{1, 1}' a_{1, 2}'\}$
		is an induced matching in $G[A, B]$.
		
		Next, we consider a cut $(A', B')$ induced by an edge in $D_1' - V(Q_{h_1}')$,
		and let $(A, B)$ be the corresponding cut in $D_1 - V(Q_{h_1})$.
		Throughout the following, we use the notation $P^\star = b_1 \ldots b_\ell$.
		(Note that $b_1 = a_{1, 1}'$ and that $b_\ell = a_{2, 2}'$.)
		We may assume that $V(P^\star) \subseteq B'$,
		that $\{a_{1, 1}', a_{1, 2}'\} \subseteq A$, and that
		$\{a_{2, 1}', a_{2, 1}'\} \subseteq B$.
		Let $M'$ be an induced matching in $G[A', B']$.
		We may assume that there is some $a_{1, 1}' w \in M'$.
		Since $a_{1, 2}' = b_2 \in A$ and $a_{2, 1}' = b_{\ell-1} \in B$,
		and since there is a path $b_2 \ldots b_{\ell-1}$ in $G$,
		we know that there is some $i \in \{2, \ldots, \ell-2\}$
		such that $b_i \in A$ and $b_{i+1} \in B$.
		Since $V(P^\star) \subseteq B'$, 
		no endpoint of an edge in $M'$ is adjacent to $b_i$ or $b_{i+1}$.
		Therefore, $M' \setminus \{a_{1, 1}' w\} \cup \{b_i b_{i+1}\}$ is an induced matching in $G[A, B]$.
		
		It remains to consider a cut $(A', B')$ induced by an edge in $\decT''$.
		In this case there is some $i \in [\ell-1]$ such that 
		we may assume that $\{b_1, \ldots, b_i\} \subseteq A'$
		and that $\{b_{i+1}, \ldots, b_{\ell}\} \subseteq B'$.
		Let $(A, B)$ be the corresponding cut in $(\decT, \decf)$, 
		meaning the one induced by the edge $x_1 x_2$.
		By construction, 
		$\{b_1, b_2\} \subseteq A$ and $\{b_{\ell-1}, b_\ell\} \subseteq B$.
		Let $M'$ be an induced matching in $G[A', B']$.
		If $M'$ contains an edge between $b_1$ (or $b_\ell$) 
		and a vertex not in $P^\star$,
		then this edge is also contained in $G[A, B]$.
		If $M'$ uses the edge $b_i b_{i+1}$, 
		then by the same argument as above,
		we can find an edge $b_j b_{j+1}$ in $G[A, B]$ 
		that we can replace $b_i b_{i+1}$ with
		to obtain an induced matching in $G[A, B]$.
	\end{ourcase}
	This finishes the proof of the lemma.
\end{proof}
\fi

The construction in the previous lemma does not 
increase the $\Fmatch$-branchwidth of a branch decomposition.
\ifshort
It might however introduce induced chain graphs of value $2$ in some cut, which is a technical complication that can, luckily, be dealt with by considering the case of $\Fstar$-branchwidth $1$ separately ($\star$).
\fi
\iflong
It might however introduce induced chain graphs of value $2$ in some cut. 
Before we proceed and give the construction to deal with this case,
we prove an auxiliary lemma that tells us that we need not worry about
larger chains and anti-matchings.
(Observe that an induced anti-matching of value $2$ is an induced matching of value $2$,
so this case is covered by Lemma~\ref{lem:fes:unimportant:placement}.)
\begin{lemma}\label{lem:fes:aux:small}
	Let $G$ be a graph,
	and let $P = a_1 \ldots a_\ell \subseteq G$ be an unimportant path with $\ell \ge 5$,
	and let $P^\star = a_2 \ldots a_{\ell-1}$.
	Let $(A, B)$ be any vertex cut of $G$, and let $H = G[A, B]$.
	Then, in $H$, no vertex of $P^\star$ is contained in
	\begin{itemize}
		\item an induced anti-matching of value at least $3$,
		\item an induced strict chain of value at least $3$.
	\end{itemize}
\end{lemma}
\begin{proof}
	The length of any shortest cycle (in $G$) that contains a vertex from $P$ has length at least $7$.
	For \Fantimatch, we observe that each vertex in an anti-matching of value at least $4$ has degree at least three;
	since $P$ is an unimportant path, all its vertices have degree $2$ in $G$ and therefore they cannot be part of 
	such an anti-matching in $H$.
	An anti-matching of value $3$ is an induced $C_6$, 
	and we observed above that no vertex of $P$ can be in a cycle of that length.
	For \Fchain, we observe that each vertex in a strict chain of value at least $3$ 
	is adjacent to a vertex of degree at least $3$; 
	therefore no vertex of $P^\star$ can be part of such a strict chain.
\end{proof}

Next we work towards a modification of branch decompositions 
that preserves their $\ichainwidth$.
The desired property is captured in the following definition.
\fi
\ifshort
Next, we can develop separate arguments that are designed to handle the case of $\Fchain$-branchwidth when $\Fstar$-branchwidth is at least $2$, and show that these two approaches can be combined together to deal with $\Fstar$-branchwidth for any union $\Fstar$ of \primal \phsymb\ classes ($\star$). Recall that a graph is \emph{bridgeless} if there is no edge $e \in E(G)$
such that $G - e$ has more connected components than $G$; we obtain:
\fi
\iflong
\begin{definition}
	Let $G$ be a graph, $(\decT, \decf)$ be a branch decomposition of $G$,
	and let $P = a_1 \ldots a_\ell \subseteq G$ be a path.
	We say that $(\decT, \decf)$ \emph{neatly breaks} $P$
	if there is an $(s_1, s_\ell)$-path $\decT'$ in $\decT$ such that:
	\begin{itemize}
		\item $\{\decf(a_1) s_1, \decf(a_\ell) s_\ell, \decf(a_{\ell-1}) s_\ell\} \subseteq E(\decT)$.
		\item There is a subpath $\decT_1$ of $\decT'$ such that $\decT[V(\decT_1) \cup \decf(\{a_1, \ldots, a_{\ell-2}\})]$
		is a caterpillar whose corresponding linear order is $a_1, \ldots, a_{\ell-2}$.
	\end{itemize}
\end{definition}
\fi

\iflong
The following lemma shows that neat breakage has the desired property:
if a path is neatly broken by a branch decomposition,
then none of its vertices is part of an induced chain of value $2$ or more in any cut.
\begin{lemma}
	\label{lem:fes:neatly:breaks}
	Let $G$ be a graph, $(\decT, \decf)$ be a branch decomposition of $G$,
	and let $P = a_1 \ldots a_\ell$ be an unimportant path with $\ell \ge 5$.
	If $(\decT, \decf)$ neatly breaks $P$, then
	there is no cut $(A, B)$ induced by $(\decT, \decf)$ such that
	any vertex of $P^\star = a_2 \ldots a_{\ell-1}$ is in an induced strict chain of value at least $2$.
\end{lemma}
\begin{proof}
	Suppose for a contradiction that there is a cut $(A, B)$ induced by $(\decT, \decf)$
	such that a vertex of $P^\star$ is contained in an induced strict chain $C$ of value at least~$2$ in $G[A, B]$.
	By Lemma~\ref{lem:fes:aux:small} we know that $C$ has value $2$.
	It is not possible that $\card{V(C) \cap V(P)} = 1$,
	since each vertex in $P^\star$ only has neighbors in $V(P)$.
	Since $G[A, B]$ contains at most one edge of $P$, we have that 
	$V(P) \cap V(C) = \{a_i, a_{i+1}\}$ for some $i \in \{1, \ldots, \ell-1\}$.
	If $i < \ell - 2$ or $i = \ell - 1$, then
	$\card{E(G[A, B])} \le 2$,
	so we may assume that $i = \ell - 2$.
	Assume up to renaming that $a_{\ell-2} \in A$ and $a_{\ell-1} \in B$.
	The neighbor of $a_{\ell-2}$ other than $a_{\ell-1}$ is in $A$,
	and the neighbor of $a_{\ell-1}$ other than $a_{\ell-2}$ is in $B$,
	therefore $a_{\ell-2}a_{\ell-1}$ cannot be an edge of an induced chain of value~$2$ in $G[A, B]$.
\end{proof}
\fi

\iflong
The previous lemma has the following consequence.
\fi
\iflong
\begin{corollary}
	\label{cor:fes:neatly:breaks}
	Let $G$ be a graph, $(\decT, \decf)$ be a branch decomposition of $G$,
	and let $P = a_1 \ldots a_\ell$ be an unimportant path with $\ell \ge 5$.
	Then, there is a branch decomposition $(\decT', \decf')$ of $G$ such that
	\begin{enumerate}
		\item $\ichainwidth(\decT', \decf') \le \ichainwidth(\decT, \decf)$,
		\item $(\decT', \decf')$ neatly breaks $P$, and
		\item restricted to $G - \{a_2, \ldots, a_{\ell-1}\}$, 
		the set of cuts induced by $(\decT, \decf)$ 
		is equal to the set of cuts induced by $(\decT', \decf')$.
	\end{enumerate}
\end{corollary}
\begin{proof}
	Let $P^\star = a_2 \ldots a_{\ell-1}$.
	We obtain $(\decT', \decf')$ from $(\decT, \decf)$ as follows:
	\begin{itemize}
		\item We remove $\decf(V(P^\star))$ from $\decT$.
		\item We subdivide the edge between $\decf(a_\ell)$ and its neighbor once and
		attach a new leaf to the subdivision vertex, say $t_{\ell-1}$.
		\item We subdivide the edge between $\decf(a_1)$ and its neighbor $\ell-3$ times 
		and attach a new leaf to each subdivision vertex. 
		We denote these new leaves by $t_2, \ldots, t_{\ell-2}$,
		ordered by distance to $\decf(a_1)$.
		\item We let $\decf'|_{V(G) \setminus V(P^\star)} = \decf|_{V(G) \setminus V(P^\star)}$,
		and for $i \in \{2, \ldots, \ell-1\}$, we let $\decf'(a_i) = t_i$.
	\end{itemize}
	By this construction it is clear that $(\decT', \decf')$ neatly breaks $P$,
	and the bound on $\ichainwidth(\decT', \decf')$ follows from Lemma~\ref{lem:fes:neatly:breaks},
	since it asserts that no vertex whose position has changed is in a strict chain of value~$2$ or more.
	Since $G$ has at least one edge, $(\decT, \decf)$ must have a cut of value at least one.
\end{proof}

So far we have \iflong given\fi two modifications of branch decompositions; 
one that preserves $\mimw$ and one that preserves $\ichainwidth$.
However, a family $\Fstar$\xspace may contain both $\Fmatch$ and $\Fchain$,
in which case it would not be clear which construction to use.
Since the only problematic case is when the construction of Lemma~\ref{lem:fes:unimportant:placement}
might introduce a chain of value $2$,
we show that we can treat that case of $\Fstar$-branchwidth $1$ separately ($\star$).
It turns out that in this case we have a linear kernel immediately, 
or we can apply Corollary~\ref{cor:fes:neatly:breaks}.

We will now proceed to prove what can be considered the main technical lemma of this section;
we only require a bit more terminology.
Let $G$ be a graph with feedback edge set $F \subseteq E(G)$.
We call the vertices $V(F)$ \emph{significant}
and each vertex that is not significant \emph{insignificant}.
Furthermore, recall that $G$ is called \emph{bridgeless} if there is no edge $e \in E(G)$
such that $G - e$ has more connected components than $G$.
We obtain:
\fi

\FesMainLemma
\iflong
\begin{proof}
	Throughout the proof, let $P = a_1 \ldots a_\ell$
	and $P^\star = a_2 \ldots a_{\ell-1}$.
	We consider two cases.	
	\setcounter{ourcase}{0}
	\begin{ourcase}[$\Fsbw(G) = 1$]
		We first show that if $\Fstar$ contains either the induced matchings or anti-matchings,
		then the size of $G$ must be small.
		\begin{claim}\label{claim:immediate:kernel}
			Let $\Fstar$ be a union of primal \phsymb classes
			and suppose 
			$(\Fmatch \cup \Fantimatch) \cap \Fstar \neq \emptyset$.
			Let $G$ be a bridgeless graph and let $F \subseteq E(G)$ be a feedback edge set of $G$ of size $k$.
			If $\Fsbw(G) = 1$, then $\card{V(G)} \le 8k - 3$.
		\end{claim}
		\begin{claimproof}
			We argue that if $\card{V(G)} > 8k - 3$, 
			then $G$ has an induced cycle on at least six vertices.
			Since $G$ has no bridges, 
			we know the following for each vertex $v \in V(G)$:
			Either $v$ is significant, or there is an edge $xy \in F$
			such that $v$ is on the path from $x$ to $y$ in $G - xy$.
			There are at most $2k$ significant vertices, 
			which means that there are more than $6k - 3$ insignificant vertices.
			$G - F$ consists of the significant vertices connected into a tree structure by at most $2k-1$ paths.
			Since there are more than $3(2k-1)$ insignificant vertices, 
			we know that at least one of these paths has four vertices.
			This in turn implies that $G$ contains an induced $C_6$.
			
			Now, it is easy to see that any branch decomposition of $G$ has a cut $(A, B)$ 
			such that $\card{A \cap V(C)} \ge 2$ and $\card{B \cap V(C)} \ge 2$ (see Lemma~\ref{lem:balancedtrees}).
			This cut therefore contains an induced matching of value $2$
			and since an induced anti-matching of value $2$ is an induced matching of value $2$,
			the claim follows.
		\end{claimproof}
		
		From now on, we may assume that $(\Fmatch \cup \Fantimatch) \cap \Fstar = \emptyset$,
		which immediately implies that $\Fstar = \Fchain$.
		We prove that $\ichainwidth(G/e) \le \ichainwidth(G)$,
		and note that the other inequality can be shown analogously.
		Let $(\decT, \decf)$ be a branch decomposition of $G$ with $\ichainwidth(\decT, \decf) = 1$.
		We apply Corollary~\ref{cor:fes:neatly:breaks} to obtain a branch decomposition $(\decT'', \decf'')$
		with the properties listed there.
		Now consider $G' = G/a_2a_3 \simeq G/e$,
		and let $(\decT', \decf')$ be a branch decomposition of $G'$ obtained from $(\decT'', \decf'')$
		by replacing $a_2$ and $a_3$ with the vertex resulting from the contraction of $a_2a_3$.
		Let $P'$ denote the path corresponding to $P$ in $G'$,
		and note that $P'$ has length at least $7$.
		Since $(\decT', \decf')$ neatly breaks $P'$, we have that no internal vertex,
		which includes the vertex constructed by contracting $a_2 a_3$ 
		is in any induced strict chain of value at least $2$,
		therefore
		$\ichainwidth(\decT', \decf') \le \ichainwidth(\decT, \decf)$ and so $\ichainwidth(G/e) \le \ichainwidth(G)$.
		The argument showing $\ichainwidth(G) \le \ichainwidth(G/e)$ is analogous,
		and concludes this case.
	\end{ourcase}
	\begin{ourcase}[$\Fsbw(G) \ge 2$]
		We first show that $\Fsbw(G/e) \le \Fsbw(G)$.
		By Lemma~\ref{lem:fes:unimportant:placement},
		there is a branch decomposition $(\decT, \decf)$ of $G$ 
		with $\imatchwidth(\decT, \decf) = \imatchwidth(G)$
		that preserves a subpath of $P$ on five vertices, 
		say $P^\star = a_1 \ldots a_5 \subseteq P$.
		By Lemma~\ref{lem:fes:aux:small},
		no vertex of $P^\star$ can contribute to an induced anti-matching
		or an induced strict chain of value more than $2$,
		therefore $\Fsbw(\decT, \decf) = \Fsbw(G)$.
		We obtain a branch decomposition $(\decT', \decf')$
		of $G' \defeq G/a_2a_3 \simeq G/e$
		with $\Fsbw(\decT', \decf') \le \Fsbw(\decT, \decf)$ 
		as follows.
		Let $a_{23}$ denote the the vertex created during the contraction of $a_2 a_3$
		and $P^{\star\star}$ the path $a_1 a_{23} a_4 a_5$ in $G'$.
		
		\begin{itemize}
			\item We obtain $\decT'$ from $\decT$ by removing the leaf $\decf(a_3)$.
			\item To obtain $\decf'$, we let $\decf'(a_{23}) = \decf(a_2)$ and 
			for all $v \in V(G') \setminus \{a_{23}\}$ we let $\decf'(v) = \decf(v)$.
		\end{itemize}
		
		Let $\decS'$ denote the caterpillar in $\decT'$ to whose leaves the vertices of $P^{\star\star}$ are mapped.
		Consider any cut $(A', B')$ that is induced by an edge in $\decT' - V(\decS')$.
		We may assume that $V(P^{\star\star}) \subseteq A'$.
		Let $(A, B)$ be the corresponding cut in $(\decT, \decf)$
		with $V(P^\star) \subseteq A$, 
		$A \setminus V(P^\star) = A' \setminus V(P^{\star\star})$,
		and $B = B'$.
		Any edge in $G[A', B']$ incident with a vertex in $P^{\star\star}$
		is incident with either $a_1$ or $a_5$;
		and these edges are also contained in $G[A, B]$.
		Therefore, any induced matching in $G[A', B']$ is also an induced matching in $G[A, B]$.
		
		It remains to consider cuts induced by edges in $\decS'$;
		observe that it suffices to consider cuts induced by edges that are not incident with a leaf of $\decS'$.
		Moreover, the only edges in $G'$ that are not in $G$ are $a_1 a_{23}$ and $a_{23} a_4$,
		therefore we may focus on cuts in which one of these edges crosses.
		Let $(A', B')$ be such a cut. 
		If $a_1 a_{23} \in E(G[A', B'])$,
		then we may assume that $A' \cap V(P^{\star\star}) = \{a_1\}$
		and $B' \cap V(P^{\star\star}) = \{a_{23}, a_4, a_5\}$.
		There is a cut $(A, B)$ in $(\decT, \decf)$
		with $A \setminus V(P^\star) = A' \setminus V(P^{\star\star})$,
		$B \setminus V(P^\star) = B' \setminus V(P^{\star\star})$,
		as well as $a_1 \in A$ and $a_2 \in B$.
		Let $M'$ be any induced matching in $G[A', B']$.
		If $a_1 a_{23} \notin G[A, B]$,
		then $M' \subseteq E(G[A, B])$;
		so we may assume that $a_1 a_{23} \in M'$.
		But then, $M' \setminus \{a_1 a_{23}\} \cup \{a_1 a_2\}$ is an induced matching in $G[A, B]$.
		The other case when $a_{23} a_4 \in E(G[A', B'])$
		can be argued similarly with the edge $a_3 a_4$ taking the role of $a_1 a_2$.
		This finishes the argument that $\imatchwidth(\decT', \decf') \le \imatchwidth(\decT, \decf)$.
		For the remaining measures, we observe that $a_{23}$ cannot appear in 
		an induced anti-matching
		or induced strict chain of value more than $2$ by Lemma~\ref{lem:fes:aux:small},
		so we have that $\Fsbw(G/e) \le \Fsbw(G)$.
		
		We now show that $\Fsbw(G) \le \Fsbw(G/e)$.
		First, let $P' = P/e$ be the path in $G/e$ corresponding to $P$
		and note that $P'$ is of length at least $7$.
		Therefore, there is a branch decomposition $(\decT', \decf')$ 
		of $G/e$ with $\Fsbw(\decT', \decf') = \Fsbw(G/e)$ that preserves 
		a subpath of $P'$ on five vertices, say $P^\star = a_1 \ldots a_5$, 
		by Lemmas~\ref{lem:fes:unimportant:placement} and~\ref{lem:fes:aux:small}.
		We subdivide the edge $a_2 a_3$ to obtain a graph $G'' \simeq G$;
		let $b$ denote the subdivision vertex.
		We construct a branch decomposition $(\decT, \decf)$ of $G''$ as follows.
		\begin{itemize}
			\item We insert a new leaf $x$ between $\decf'(a_2)$ and $\decf'(a_3)$ to obtain $\decT$.
			\item for all $v \in V(G/e)$, we let $\decf(v) \defeq \decf'(v)$
			and we let $\decf(b) \defeq x$.
		\end{itemize}
		Using similar arguments as above, we can show that $\Fsbw(\decT, \decf) \le \Fsbw(\decT', \decf')$
		and therefore $\Fsbw(G) \le \Fsbw(G/e)$.
	\end{ourcase}
	This finishes the proof of Lemma~\ref{lem:fes:main}.
\end{proof}
\fi

As the previous lemma requires that our input graph does not have any bridges,
we need one more simple lemma to show that we can safely remove bridges from the
input graph without changing its \Fstar-branchwidth.
\ifshort ($\star$)
\fi
\iflong
\begin{lemma}\label{lem:fes:bridge}
	Let $\Fstar$ be the union of some primal \phsymb families.
	Let $G$ be a graph and let $e \in E(G)$ be a bridge of $G$, and
	let $G_1$ and $G_2$ be the two connected components of $G - e$.
	For $h \in [2]$, let $(\decT_h, \decf_h)$ be a branch decomposition of $G_h$.
	Let $(\decT, \decf)$ be a branch decomposition of $G$ obtained as follows.
	\begin{itemize}
		\item For $h \in [2]$, subdivide an arbitrary edge of $\decT_h$,
		denote the resulting node by $t_h$.
		\item Let $\decT = (V(\decT_1) \cup V(\decT_2), E(\decT_1) \cup E(\decT_2) \cup \{t_1t_2\})$,
		and let $\decf$ be such that 
		$\decf|_{V(G_1)} = \decf_1$ and $\decf|_{V(G_2)} = \decf_2$.
	\end{itemize}
	Then, $\Fsbw(\decT, \decf) \le \max\{1, \Fsbw(\decT_1, \decf_1), \Fsbw(\decT_2, \decf_2)\}$.
\end{lemma}
\begin{proof}
	Let $(A, B)$ be a cut induced by $(\decT, \decf)$, 
	and let $f \in E(\decT)$ denote the corresponding edge.
	Suppose $f \in E(\decT_h)$ for some $h \in [2]$,
	and let $(A', B')$ denote the cut that $f$ induces in $(\decT_h, \decf_h)$.
	Then, both endpoints of $e$ are either in $A$ or in $B$,
	and since there are no other edges between vertices in $G_1$ and $G_2$, 
	we have that $E(G[A, B]) = E(G[A', B'])$.
	If $f \notin E(\decT_h)$ for any $h \in [2]$,
	then $f = t_1 t_2$.
	But in this case, $E(G[A, B]) = \{e\}$ and therefore 
	$\Fstarval(A, B) = 1$.
\end{proof}

The previous lemma also has the following consequence for forests.
\begin{corollary}\label{cor:forest}
	Let $\Fstar$ be a union of primal \phsymb families, 
	and let $G$ be a forest with at least one edge.
	Then, $\Fsbw(G) = 1$, and a branch decomposition 
	witnessing the bound can be found in polynomial time.
\end{corollary}

\fi
\iflong
We are now in shape to give the kernelization algorithm.
The first reduction rule `cleans up' the graph by removing bridges 
and isolated vertices.
It is clear that isolated vertices do not contribute to the \Fstar-branchwidth
of any branch decomposition when \Fstar\xspace is the union of primal \phsymb families.
\fi
\FesBridgeIsolatedReduction

\iflong
By Lemma~\ref{lem:fes:bridge} and triviality, we have the following observation.
\begin{observation}\label{obs:fes:bridge:cor}
	Reduction Rule~\ref{reduction:fes:bridge:isolated} is safe
	and can be executed in polynomial time.
\end{observation}

The safeness of the next rule is based on the technical lemmas of this section.
\fi
\FesUnimportantReduction

\iflong
We justify the safeness of Reduction Rule~\ref{reduction:fes:unimportant} in the proof below,
where we describe how to apply it.

\FesMainTheorem
\begin{proof}
	Let $(G, k)$ be an instance of \Fsbwprob where $G$ has a feedback edge set of size $k$.
	Observe that a feedback edge set of minimum size can be computed in polynomial time,
	therefore we may assume that we have a feedback edge set of size $k$ at hand.
	The algorithm proceeds as follows.
	
	\begin{enumerate}
		\item\label{kernel:fes:step1}
		Apply Reduction Rule~\ref{reduction:fes:bridge:isolated}.
		\item\label{kernel:fes:step2} 
		While $\card{V(G)} > 18k - 8$:
		\begin{enumerate}
			\item\label{kernel:fes:step2a} 
			Find an unimportant path $P \subseteq G$ of length at least $8$.
			\item\label{kernel:fes:step2b} 
			Apply Reduction Rule~\ref{reduction:fes:unimportant},
			i.e.\ contract an edge of $P$ in $G$.
			We denote the resulting instance again by $(G, k)$.
		\end{enumerate}
		\item Return $(G, k)$.
	\end{enumerate}
	
	We now argue that the above algorithm meets the requirements of a kernelization algorithm.
	For Step~\ref{kernel:fes:step1}, it follows from Observation~\ref{obs:fes:bridge:cor}.
	(Note also that Reduction Rule~\ref{reduction:fes:bridge:isolated}
	does not change the fact that $G$
	has a feedback edge set of size $k$.)
	Let us consider Step~\ref{kernel:fes:step2}.
	Observe that thanks to Step~\ref{kernel:fes:step1}, 
	we may assume that $G$ has neither bridges nor isolated vertices.
	Now, if $\card{V(G)} > 18k - 8$, then $G$ must have an unimportant path of length $8$:
	Each vertex of $G$ is either significant, 
	or it lies on a path connecting two significant vertices;
	in other words, it is on an unimportant path.
	Since there are at most $2k$ significant vertices,
	there are more than $16k - 8$ vertices that are spread out among at most $2k-1$ unimportant paths.
	By the pigeonhole principle, 
	one of these unimportant paths must have at least $9$ vertices,
	and therefore length at least $8$.
	
	Now, let $P$ be such an unimportant path as chosen in Step~\ref{kernel:fes:step2a} and 
	let $e \in E(P)$ be the edge that is getting contracted in
	Step~\ref{kernel:fes:step2b}.
	By Lemma~\ref{lem:fes:main}
	and since $\card{V(G)} > 18k - 8 > 8k - 3$ 
	(we may assume\footnote{For if $k = 0$, then $G$ is a forest, 
		in which case Corollary~\ref{cor:forest} implies that 
		$\Fsbw(G) = 1$ if $G$ has an edge and $\Fsbw(G) = 0$ otherwise.} 
	that $k \ge 1$),
	we can conclude that $\Fsbw(G) = \Fsbw(G/e)$.
	A long unimportant path can be found in polynomial time,
	and contracting an edge whose endpoints have degree two takes constant time.
	Therefore, Step~\ref{kernel:fes:step2} also meets the requirements of a kernelization algorithm.
	
	Since in each iteration, the number of vertices in $G$ decreases by one,
	we know that eventually, $\card{V(G)} \le 18k - 8$,
	upon which the algorithm returns $(G, k)$.
	The size bound on the kernelized instance follows since $G$ has $\calO(k)$ vertices,
	and at most $\card{V(G)} - 1 + k = \calO(k)$ edges.
\end{proof}
\fi

\FesMainCorollary
	
	\section{Concluding Remarks}
While our introduction focused predominantly on $\mathcal{F}$-branchwidth acting as a unifying framework for treewidth, clique-width and mim-width, the concept is in fact significantly more powerful than that; for instance, it is easy to see that the treewidth of the complement of the graph~\cite{GomesLS19} is also captured by $\mathcal{F}$-branchwidth (one can simply take the complements of the classes used in Corollary~\ref{cor:FisTWBounded}). On the other hand, one immediate question arising from this work concerns the newly identified combinations of \texttt{\emph{si ph}} classes: could, e.g., $\ichainFamily$- and $\iantimatchFamily$-branchwidth be natural counterparts to mim-width?

On the algorithmic side, one fundamental limitation one needs to keep in mind is that computing $\mathcal{F}$-branchwidth parameterized by the measure itself is $\W[1]$-hard~\cite{SaetherV16}. Hence we cannot hope for a fixed-parameter framework that would first compute an optimal $\mathcal{F}$-branch decomposition, and then proceed it to solve a problem of interest. In this sense, the question tackled by our algorithmic contribution is primarily a conceptual one: when can we compute \emph{all} of the captured width parameters via a unified framework? The most obvious question that remains open in this regard is whether Theorem~\ref{thm:tw_deg} can be generalized to the parameterization by treewidth alone.

A separate question that is no less important is whether we can compute approximately-optimal decompositions for $\mathcal{F}$-branchwidth in polynomial time when the width is bounded by a constant; solving this problem already for mim-width alone would be considered a major breakthrough in the field.

\bibliographystyle{plainurl}
\bibliography{literature}
	
\end{document}